%% file: arxiv_v2.tex
\expandafter\def\csname ver@fixltx2e.sty\endcsname{}
\documentclass[10pt,onecolumn]{IEEEtran}
\input{preamble}

\IEEEoverridecommandlockouts
\usepackage{cite}
\usepackage{amsmath,amssymb,amsfonts}
\usepackage{graphicx}
\usepackage{textcomp}
\usepackage{xcolor}
\usepackage{geometry}
\usepackage{subfigure}
\usepackage{float}
\usepackage[ colorlinks = true,
linkcolor = blue,
urlcolor  = blue,
citecolor = red,
anchorcolor = green,]{hyperref}

\usepackage{braket}
\def\BibTeX{{\rm B\kern-.05em{\sc i\kern-.025em b}\kern-.08em
T\kern-.1667em\lower.7ex\hbox{E}\kern-.125emX}}
\definecolor{Dyellow}{RGB}{254,152,0}
\definecolor{Dgreen}{RGB}{0,176,80}

\newcommand{\red}[1]{\textcolor{red}{#1}}
\begin{document}

\title{Achievable Second-Order Asymptotics for MAC and RAC with Additive Non-Gaussian Noise}

\author{Yiming Wang, Lin Bai, Zhuangfei Wu, and Lin Zhou

\thanks{This paper was partially presented at GLOBECOM 2023 \cite{wang2023globecom}.}
\thanks{The authors are with the School of Cyber Science and Technology, Beihang University, Beijing, China, 100191 (Emails: \{1mwang, l.bai, zhuangfeiwu, lzhou\}@buaa.edu.cn).}
\thanks{Corresponding author: Lin Zhou.}
}

\maketitle

\begin{abstract}
We first study the two-user additive noise multiple access channel (MAC) where the noise distribution is arbitrary. For such a MAC, we use spherical codebooks and either joint nearest neighbor (JNN) or successive interference cancellation (SIC) decoding. Under both decoding methods, we derive second-order achievable rate regions and compare the finite blocklength performance between JNN and SIC decoding. Our results indicate that although the first-order rate regions of JNN and SIC decoding are identical, JNN decoding has better second-order asymptotic performance. When specialized to the Gaussian noise, we provide an alternative achievability proof to the result by MolavianJazi and Laneman (T-IT, 2015). Furthermore, we generalize our results to the random access channel (RAC) where neither the transmitters nor the receiver knows the user activity pattern. We use spherical-type codebooks and a rateless transmission scheme combining JNN/SIC decoding and derive second-order achievability bounds. Comparing second-order achievability results of JNN and SIC decoding in a RAC, we show that JNN decoding achieves a strictly larger first-order asymptotic rate. When specialized to Gaussian noise, our second-order asymptotic results recover the corresponding results of Yavas, Kostina, and Effros (T-IT, 2021) up to second-order.
\end{abstract}


\section{Introduction}
\label{sec:introduction}
Multiple access strategy plays a key role in nowadays communication scenarios such as the Internet of Things and massive machine-type communications~\cite{saad2020iot,federico2014massive}. 
Orthogonal multiple access (OMA) strategies, such as frequency/time/code division multiple access or orthogonal frequency-time division multiple access, are widely applied in 1G to 4G wireless communication systems~\cite{jiaxing2024aoiuav,rex2008bookOMA}.
However, the achievable rate region of OMA is usually inferior to non-orthogonal multiple access (NOMA)~\cite{zhiguo2017noma}. As the data rate requirement grows from 10 Gbps to 1 Tbps in the 6G era~\cite{letaief2019roadmap6G}, it is essential to explore effective NOMA strategies. The uplink transmission of a NOMA system is modeled by the multiple access channel (MAC) proposed by Shannon~\cite{shannon1961two}. In a MAC, multiple users aim to transmit messages to one receiver sharing one resource block in a non-orthogonal way~\cite{Liu2024MAtutorial}. Following Shannon's work, Ahlswede~\cite{Ahlswede1973MultiwayCC} proposed the successive interference cancellation (SIC) decoding to mitigate the inter-user interference and derived the first-order asymptotic capacity region. The capacity region can be achieved by SIC decoding with time sharing or joint nearest neighbor (JNN) decoding~\cite{shannon1959Gaussiancode}. SIC decoding is more computationally efficient since the computational complexity of JNN decoding grows exponentially with the number of users, while that of SIC decoding grows linearly~\cite[Chapter 5]{verdu1998multiuser}.

The above studies focus on the capacity performance of a MAC, assuming that the blocklength approaches infinity. However, the performance gap between the achievable rate at finite blocklength and the asymptotic capacity is tremendous~\cite{PPV}. To close the gap, the second-order asymptotic results have been established to approximate the finite blocklength performance of optimal codes~\cite{zhou2023mono,tan2014mono}. Under JNN decoding, MolavianJazi and Laneman~\cite{molavianjazi2015second} derived the second-order asymptotic bounds for the Gaussian MAC using the dependency testing bound~\cite[Theorem 17]{PPV}. Yavas, Kostina, and Effros~\cite{yavas2021gaussianmac} generalized the above results to obtain the more refined third-order asymptotics using the random coding union (RCU) bound~\cite[Theorem 16]{PPV}, which provides a tighter approximation to the finite blocklength performance.

In MAC studies, one designs and analyzes a coding scheme by assuming that the number and identities of active users that transmit messages are known.  However, in certain practical scenarios~\cite{yongpeng2020massiveaccess,poly2017massiveRac}, the user activities could be unpredictable. It is possible that neither the receiver nor the transmitters know the user activity pattern. This communication scenario is modeled by the random access channel (RAC)~\cite{poly2017massiveRac,poly2017lowComplexityRac} and more specifically the unsourced RAC~\cite{poly2020energyEfficientCodeRAC,Fengler2021massiveMIMOUnsourcedRAC}. In this RAC, there are $K$ potential users that would transmit their messages to a single receiver, out of which only a subset of $k$ users are active in a time slot, while both the number $k$ and the identities of active users are unknown. To estimate the number of active users, we adopt the ``rateless'' coding scheme applied in~\cite{yavas2021rac,yavas2021gaussianmac}. In particular, each user adopts the same spherical-type codebooks, and there is a fixed number of decoding times depending on the number of active users. In the spherical-type codebook, the codewords are separated into several independent sub-codewords that are generated uniformly over spheres with certain diameters. At each of the possible decoding times, the receiver first checks whether the estimated number of active users matches the blocklengths of codewords. If yes, the receiver decodes all messages and broadcasts a one-bit feedback to the transmitters to inform that the current time epoch ends. Otherwise, the active users keep on sending messages. The above transmission and decoding process continues until the maximal decoding time is reached. Since the blocklength varies with the number of active users, there exists a fundamental trade-off between the achievable rates, error probabilities, and the number of active users. Yavas, Kostina, and Effros~\cite{yavas2021rac,yavas2021gaussianmac} derived the third-order achievability results for the unsourced RAC when the additive channel noise is Gaussian and showed that the achievable rate decreases as the number of active users increases when the blocklength and error probabilities are fixed.

Although insightful, the aforementioned studies for MAC~\cite{molavianjazi2015second,yavas2021gaussianmac} and RAC~\cite{yavas2021rac,yavas2021gaussianmac} assume that the distribution of additive channel noise is perfectly known. In practical communication scenarios, obtaining the exact distribution of channel noise is challenging. To solve the above problem, for the point-to-point (P2P) channel, Lapidoth~\cite{lapidoth1996mismatch} proposed a coding scheme using Gaussian codebooks and nearest neighbor decoder to combat an additive arbitrary noise and showed that the capacity of an additive white Gaussian noise (AWGN) channel can be achieved for any channel noise with the same second moment as AWGN. The problem is named \emph{mismatched} channel coding because the above coding scheme, optimal for AWGN, is not necessarily optimal for other additive noise distributions. Lapidoth further generalized the results in the mismatched P2P channel to the mismatched MAC and derived the first-order achievable rates~\cite{lapidoth1996mismatch}. Using the P2P result of Lapidoth~\cite{lapidoth1996mismatch}, Scarlett, Tan, and Durisi~\cite{scarlett2017mismatch} derived the second-order asymptotics and generalized their results to the mismatched interference channel. 

Following the above lines of research on establishing finite blocklength performance of mismatched coding schemes, to close the research gap on uplink multiple user mismatched communication, we study the MAC and RAC and derive achievable second-order asymptotics. Our main contributions are summarized as follows.   

\subsection{Main Contributions}
\label{sec:main contributions}
We first study a two-user additive noise MAC under mild moment constraints on the noise distribution. For such a MAC, we use spherical codebooks with either a JNN or SIC decoder and derive the corresponding second-order achievable rate regions. For the two-user MAC, to analyze the second-order asymptotics, we need to consider all possible cases of boundary rate pairs in the first-order achievable rate region and analyze the corresponding second-order achievable rate regions. Our results show that the first-order rate regions of JNN and SIC decoding are identical for this mismatched MAC. However, JNN decoding possesses a better second-order asymptotic performance compared to SIC decoding. Furthermore, for JNN decoding, we provide an alternative form of second-order asymptotics that covers all cases with a unified formula. Our proposed coding scheme can be applied to a MAC with additive arbitrarily distributed noise with guaranteed theoretical performance, requiring only the moment constraints that can be measured in practice. When specialized to the Gaussian noise, our results are consistent with the corresponding second-order asymptotic results for Gaussian MAC~(cf \cite[Theorem 1]{molavianjazi2015second} and~\cite[Theorem 2]{yavas2021gaussianmac}).

While our proofs build on techniques for the P2P mismatched study~\cite{scarlett2017mismatch} and the Gaussian MAC studies~\cite{yavas2021gaussianmac,molavianjazi2015second}, the contribution of our paper goes beyond \cite{scarlett2017mismatch,yavas2021gaussianmac,molavianjazi2015second}. In contrast to~\cite{scarlett2017mismatch} that focuses on P2P mismatched channels, we address the multi-user mismatched MAC, which requires new analytical tools to handle the multi-dimensional rate regions inherent to the MAC. Furthermore, the proofs in~\cite{yavas2021gaussianmac,molavianjazi2015second} focused on the Gaussian MAC. While we use the function version Berry-Esseen Theorem (see Lemma~\ref{lemma:berry esseen for func}) similarly to \cite{molavianjazi2015second}, the mismatched coding scheme we propose is different from the matched case in~\cite{molavianjazi2015second}. In particular, the dependence-testing bound~\cite[Theorem 17]{PPV} was used in \cite{molavianjazi2015second} while we derive a modified version of the RCU bound~\cite[Theorem 16]{PPV} tailored to non-Gaussian MAC by relating JNN/SIC decoding with maximal mismatched information density decoding. Although \cite{yavas2021gaussianmac} provided a more refined result involving the third-order term, their proof focused on the Gaussian MAC, in particular \cite[Lemma 6]{yavas2021gaussianmac}, which is critical for obtaining the third-order asymptotics. To generalize their result, we introduce a looser but general bound (see Lemma~\ref{lemma:g function}), which does not rely on the noise distribution and leads to desired second-order asymptotics. In summary, to address non-Gaussian MAC, we judiciously adapt the proof techniques for the matched MAC~\cite{yavas2021gaussianmac,molavianjazi2015second} and combine the proof steps with those for the mismatched P2P channel~\cite{scarlett2017mismatch}. Compared with~\cite{scarlett2017mismatch,yavas2021gaussianmac,molavianjazi2015second}, both our proof steps and our results are novel.

Furthermore, we generalize our results to a RAC with an unknown user activity pattern, which generalizes MAC by having an unknown number of active users. We use spherical-type codebooks for encoding and decoding, use the rateless transmission scheme~\cite{yavas2021gaussianmac} to estimate the number of active users, and use the JNN/SIC decoder to decode the messages sent by active users. Using the above coding scheme, we derive second-order achievability bounds for both JNN and SIC decoding. Our results imply that even in the first-order achievable rate region, JNN decoding has strictly better performance. This is because all active users transmit at the same rate, and thus the performance of SIC decoding is dominated by the worst user that has a very low signal-to-interference and noise ratio. Finally, our result for RAC with JNN decoding generalizes the corresponding results for the Gaussian RAC in~\cite[Theorem 4]{yavas2021gaussianmac} to non-Gaussian noise up to second-order asymptotics.

\subsection{Organization for the Rest of the Paper}
\label{sec:organization of the paper}
The rest of the paper is organized as follows. In Section~\ref{sec:Problem Formulation}, we set up the notation, formulate the problems of two-user MAC and RAC, and present necessary definitions and existing results. In Section~\ref{sec: main results and discussions}, we present and discuss our results for MAC and RAC. The proofs of our results for MAC and RAC are presented in Section~\ref{sec: mac proof} and~\ref{sec:rac proof}, respectively. Finally, in Section~\ref{sec:conclusion}, we conclude our paper and discuss future research directions.

\section{problem formulation}
\label{sec:Problem Formulation}
\subsection*{Notation}
\label{subsec:notations}

We use $\bbR$, $\bbR_+$, $\bbN$, and $\bbN_+$ to denote the set of real numbers, positive real numbers, integers, and positive integers, respectively. Random variables are in capital (e.g., $X$) and their realizations are in lowercase (e.g., $x$). Random vectors of length $n\in\bbN$ and their particular realizations are denoted as $X^n:= (X_1, \ldots, X_n)$ and $x^n:=(x_1,\ldots,x_n)$, respectively. Given any integers $(a,b)\in\bbN^2$ such that $a<b$, we use $[a:b]$ to denote the set of integers between $a$ and $b$, we use $[a]$ to denote $[1:a]$ when $a\geq 1$. For a random vector $X^n$, we use $X_a^b$ to denote $(X_a,\ldots,X_b)$. Given any sequence of $M\in\bbN_+$ vectors $(X^n(1),\ldots,X^n(M))$ and any integers $(a,b)\in[M]^2$ such that $a\leq b$, we use $X^n([a:b])$ to denote $\{X^n(i),i\in[a:b]\}$. We use the calligraphic font (e.g., $\calX$) to denote all sets and use $\calN(\mu,\sigma^2)$ to denote the Gaussian distribution with mean $\mu$ and variance $\sigma^2$. We use the bold font (e.g., $\mathbf{X}:=((x_1^n)^\rmT \ldots (x_m^n)^\rmT)$) to denote matrices
and we use $\bzero_n$ and $\mathbf{1}_n$ to denote $n$ dimensions all zero and all one vector, respectively. For two $n$-dimensional vectors $A^n$ and $B^n$, we use $A^n\geq B^n$ to denote that for each $i\in[n]$, $A_i\geq B_i$. We use logarithms with base $e$. We use $\| x^n \| = \sqrt {\sum_{i=1}^n {x_i^2} } $ to denote the $\ell_2$ norm of a vector $x^n \in \mathbb{R}^n$. We use $Q(\cdot)$ to denote the complementary cumulative distribution function (ccdf) of a standard Gaussian distribution, and $Q^{-1}(\cdot)$ to denote its inverse.

\subsection{Two-user MAC}
\subsubsection{System Model and Code Definition}
\begin{figure}[tb]
\centering
\setlength{\unitlength}{0.4cm}
\begin{picture}(25,10)
\put(0,6){\makebox(0,0)[l]{\framebox{\strut User 1}}}
\put(3,6){\vector(1,0){3}}
\put(4.5,6.5){\makebox(0,0){$W_1$}}
\put(6,6){\makebox(0,0)[l]{\framebox{\strut Encoder $f_1$}}}
\put(10.6,6){\line(1,0){4.4}}
\put(13,6.5){\makebox(0,0){$X^n_1$}}
\put(0,0){\makebox(0,0)[l]{\framebox{\strut User 2}}}
\put(3,0){\vector(1,0){3}}
\put(4.5,0.5){\makebox(0,0){$W_2$}}
\put(6,0){\makebox(0,0)[l]{\framebox{\strut Encoder $f_2$}}}
\put(10.6,0){\line(1,0){4.4}}
\put(13,0.5){\makebox(0,0){$X^n_2$}}
\put(15,6){\line(0,-1){6}}
\put(15,3){\vector(1,0){3}}
\put(19,3){\circle{2}}
\put(19,3){\makebox(0,0){\Huge$+$}}
\put(20,3){\vector(1,0){3}}
\put(24,3){\makebox(0,0){$Y^n$}}
\put(19,7){\vector(0,-1){3}}
\put(19,7.5){\makebox(0,0){$Z^n$}}
\put(19,8.5){\makebox(0,0){\red{Arbitrary noise}}}
\end{picture}
\caption{System model of a two-user MAC with arbitrary additive channel noise.}
\label{fig:MAC system model}
\end{figure}
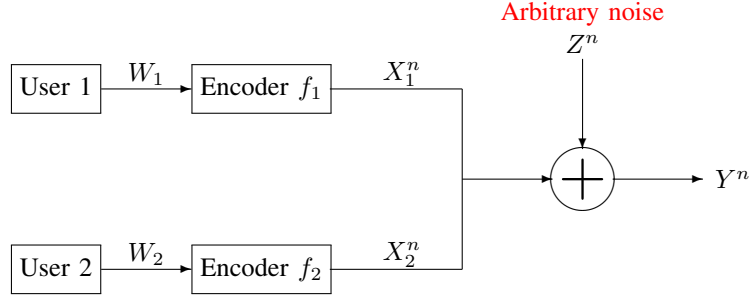

We first consider a two-user MAC with additive arbitrary noise as shown in Fig.~\ref{fig:MAC system model}. For each $j\in[2]$, user $j$ aims to transmit a message $W_j\in[M_j]$ to the receiver. The channel with blocklength $n\in\bbN_+$ is modeled as follows:
\begin{align}\label{eq:MAC channel model}
Y^n=X_1^n(W_1)+X_2^n(W_2)+Z^n,
\end{align}
where $Y^n$ is the signal at the receiver, and $(X^n_1(W_1),X^n_2(W_2))$ are the codewords sent by users 1 and 2 to transmit their messages $(W_1,W_2)$, respectively. The additive channel noise $Z^n$ is generated \iid from an arbitrary unknown distribution subject to the following moment constraints:
\begin{align}
\bbE[Z^2]=1,~\xi:=\bbE[Z^4],~\bbE[Z^6]<\infty\label{eq:noise statics}.
\end{align}
We assume that the second and the fourth moments of additive noise are known, which is possible by measurements. The assumptions on the fourth and sixth moments being finite can be satisfied by most common noises, such as uniform, Gaussian, or Laplace. We consider two coding schemes, one with a JNN decoder and another with an SIC decoder. 

Fix integers $(n,M_1,M_2)\in\bbN_+^3$ and positive real numbers $(P_1,P_2)\in\bbR_+^2$. The coding scheme with the JNN decoder for a two-user MAC is defined as follows.
\begin{definition}
\label{def:MAC JNN coding scheme}
An $(n,M_1,M_2,P_1,P_2)$-JNN code for a two-user MAC consists of:
\begin{itemize}
\item random codebooks $\calC_j$ with $M_j$ codewords $\{ X_j^n(1),\ldots,X_j^n(M_j)\}$ subject to a power constraint $P_j$ for each $j\in[2]$,
\item two encoders $(f_1,f_2)$ such that for each $j\in[2]$ and $m_j\in[M_j]$,
\begin{align}
f_j(m_j):=X_j^n(m_j),
\end{align}
\item a joint nearest neighbor decoder that uses the minimum Euclidean distance decoding, i.e.,
\begin{align}\label{eq:nn decoder}
\Phi(Y^n):=(\hatW_{1},\hatW_{2})=\argmin_{(\barw_1,\barw_2)\in[M_1]\times[M_2]}\left\|{Y^n-X_1^n(\barw_1)-X_2^n(\barw_2)}\right\|^2.
\end{align}
\end{itemize}
\end{definition}

Analogously, when SIC decoding is used, the coding scheme is defined as follows.
\begin{definition}
\label{def:MAC SIC coding scheme}
An $(n,M_1,M_2,P_1,P_2)$-SIC code for a two-user MAC consists of:
\begin{itemize}
\item the same codebooks and encoders in Def.~\ref{def:MAC JNN coding scheme},
\item a successive interference cancellation decoder with two steps. Step 1 applies the nearest neighbor decoder that treats the channel input of one transmitter as noise, i.e.,
\begin{align}\label{eq:SIC Decoder One}
\hatW_1=\Phi_1(Y^n):=\argmin_{\barw_1\in[M_1]}\|{Y^n-X_1^n(\barw_1)}\|^2.
\end{align}
Step 2 subtracts $X_1^n(\hatW_1)$ from the received signal to decode $X_2^n(\hatW_2)$ through another nearest neighbor decoder, i.e.,
\begin{align}\label{eq:SIC Decoder Two}
\hatW_2=\Phi_2(Y^n):=\argmin_{\barw_2\in[M_2]}\big\|{Y^n-X_1^n(\hatW_1)-X_2^n(\barw_2)}\big\|^2.
\end{align}
\end{itemize}
\end{definition}
Note that in Def.~\ref{def:MAC SIC coding scheme}, a fixed decoding order is assumed. In practice, the decoding order can be determined via the values of the signal-to-noise ratio (SNR) of two users. Usually, the message of the stronger user with a higher SNR value is decoded first so that the overall error probability is minimized.

\begin{figure}[tb]
\centering
\includegraphics[width =.3\columnwidth]{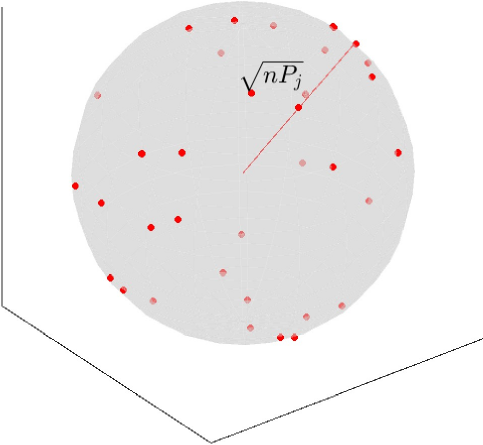}
\caption{Illustration of a spherical codebook in~\eqref{eq:sphericaldist} when $n=3$.}
\label{fig:spherical codebook}
\end{figure}

Inspired by the fact that spherical codebooks achieve better low-latency performance for mismatched channel coding in~\cite{scarlett2017mismatch}, both users are assumed to apply spherical codebooks. Specifically, for each $j\in[2]$ and each $t\in[M_j]$, codeword $X_j^n(t)$ is generated from the uniform distribution over a sphere with radius $\sqrt{nP_j}$, i.e.,
\begin{align}
P_{X_j^n}(x_j^n(t)) := \frac{\delta (\| x_j^n(t) \| ^2 - nP_j)}{S_n(\sqrt{nP_j})}\label{eq:sphericaldist},
\end{align}
where $\delta(\cdot)$ is the Dirac delta function and $S_n(r) = 2\pi^{n/2}$ $r^{n-1}/\Gamma(n/2)$ is the surface area of a radius-$r$ sphere where $\Gamma(\cdot)$ denotes the Gamma function. The spherical codebook is illustrated in Fig.~\ref{fig:spherical codebook} when $n=3$, where each red dot denotes one codeword.

\subsubsection{Performance Metrics}

To evaluate the performance of a code satisfying either Def.~\ref{def:MAC JNN coding scheme} or Def.~\ref{def:MAC SIC coding scheme}, we consider the following ensemble average error probability:
\begin{align}
\rmP_\rme^n:=\Pr \{ (\hatW_1,\hatW_2)\neq(W_1,W_2) \}\leq \varepsilon \label{eq:error prob mac}.
\end{align}
The above probability is averaged over the distributions of both random codebooks $(\calC_1,\calC_2)$, the messages $(W_1,W_2)$, and the additive noise $Z^n$. This definition is consistent with the existing studies for mismatched communications~\cite{lapidoth1996mismatch,scarlett2017mismatch,zhou2019jscc,zhou2019refined,zhou2023sr}.

\begin{figure}[tb]
\centering
\includegraphics[width = .4\columnwidth]{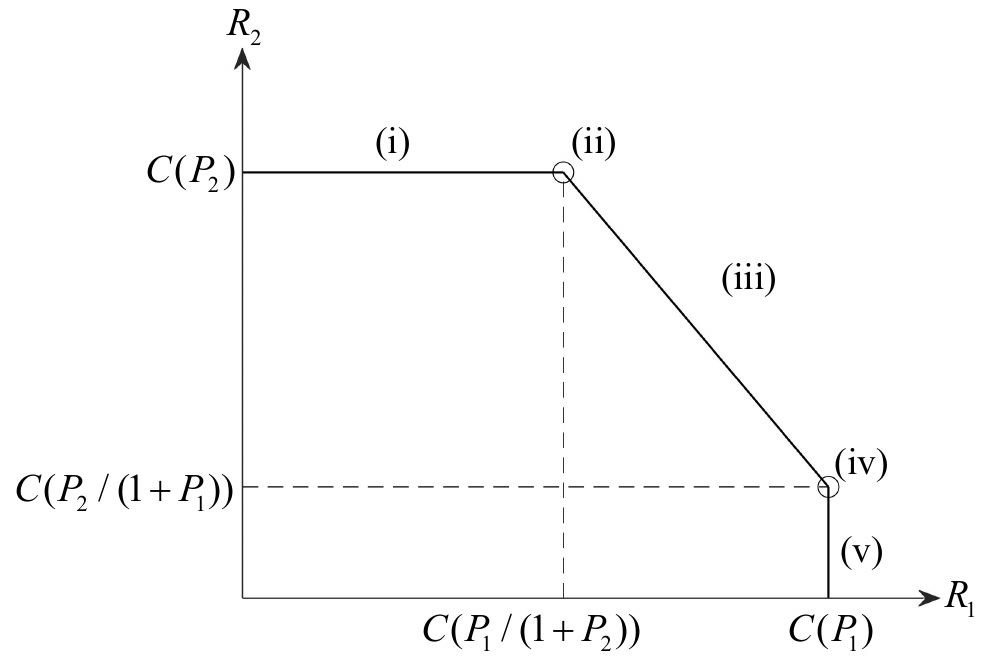}
\caption{Illustration of the first-order rate region for a two-user MAC, where for each $j\in[2]$, $R_j$ and $P_j$ are the first-order rate and transmitter power for user $j$, respectively.}
\label{fig:MACfirstorder}
\end{figure}

For multi-user communication, the definition of theoretical benchmarks is significantly different from the P2P case. In particular, since multiple encoders or multiple decoders appear, the theoretical benchmark is characterized by a multi-dimensional region instead of a one-dimensional minimal or maximal rate. In the following, we define the first- and second-order asymptotics for a two-user MAC.

Fix any $\dag\in\{\rm JNN, SIC\}$. 
\begin{definition}
\label{def:mac first order region}
A rate pair $(R_1,R_2)\in\bbR_+^2$ is said to be first-order $\dag$-achievable for the two-user MAC if there exists a sequence of $(n,M_1,M_2,P_1,P_2)$-$\dag$ codes such that for each $j\in[2]$,
\begin{align}
\liminf\limits_{n\to\infty}\frac{1}{n}\log M_j & \geq R_{j}, \\
\lim_{n\to\infty}\rmP_\rme^n&=0.
\end{align}
The closure of the set of all first-order achievable rate pairs is called the first-order rate region and denoted by $\calR_\dag$.
\end{definition}

The first-order rate region for a two-user Gaussian MAC was established in~\cite{wyner1975mac} (see also~\cite[Chapter 4.6.1]{el2011network}), which was also proved to be the first-order achievable rate region for the two-user MAC with additive non-Gaussian noise under the assumptions that the noise is independent of channel inputs and the second moment of the noise distribution is the same as AWGN~\cite{lapidoth1996mismatch}. The corresponding mismatched capacity region is 
\begin{align}
\label{eq:MAC first order region}
\calR_{\rm JNN}=\calR_{\rm SIC}=\{(R_1,R_2):R_1\leq\rmC(P_1),R_2\leq\rmC(P_2),R_1+R_2\leq\rmC(P_1+P_2)\},
\end{align}
where $C(\cdot)$ is the capacity of an AWGN channel, i.e., for any $P\in\bbR_+$,
\begin{align}
\label{eq:capacity def}
\rmC(P) := \frac{1}{2}\log (1+P),
\end{align}
The above first-order achievable rate region is illustrated in Fig.~\ref{fig:MACfirstorder} and divided into several cases, which will be useful for subsequent low-latency performance characterization. Analogously to~\cite{scarlett2017mismatch}, the key point of mismatched communication is to use a coding scheme optimal for Gaussian noise for data transmission over an additive arbitrary noise channel.

Fix any rate pair $(R_{1}^{*},R_{2}^{*})$ on the boundary of $\calR_\dag$ and any $\varepsilon\in(0,1)$, we now present the definition of second-order asymptotics.
\begin{definition}
\label{def:mac second order region}
A pair $(L_1,L_2)\in\bbR^2$ is said to be second-order $(R_1^*,R_2^*,\varepsilon)$-$\dag$ achievable for the two-user MAC if there exists a sequence of $(n,M_1,M_2,P_1,P_2)$-$\dag$ codes such that for each $j\in[2]$,
\begin{align}
\limsup\limits_{n\to\infty}\frac{1}{\sqrt{n}}(nR_{j}^*-\log M_j) & \leq L_j, \label{eq:mac L}\\
\limsup\limits_{n\to\infty}\rmP_\rme^n&\leq\varepsilon\label{eq:mac 2nd order Pen}.
\end{align}
The closure of all second-order $(R_1^*,R_2^*,\varepsilon)$-$\dag$-achievable pairs is called the second-order achievable rate region and denoted by $\calL_\dag(R_1^*,R_2^*,\varepsilon)$.
\end{definition}
Given a fixed rate pair $(R_1^*,R_2^*)$ on the boundary of the first-order rate region $\calR_\dag$, $(L_1,L_2)$ characterizes the gap of non-asymptotic rate pairs to $(R_1^*,R_2^*)$ up to second-order. This is because \eqref{eq:mac L} implies that $\log M_j\geq nR_j^*-\sqrt{n}L_j$. Thus, the larger the second-order rate region, the larger the finite blocklength transmission rates.

\subsection{RAC with Unknown User Activity Pattern}
\label{sec:RAC defs}

\begin{figure}[tb]
\centering
\includegraphics[width = .6\columnwidth]{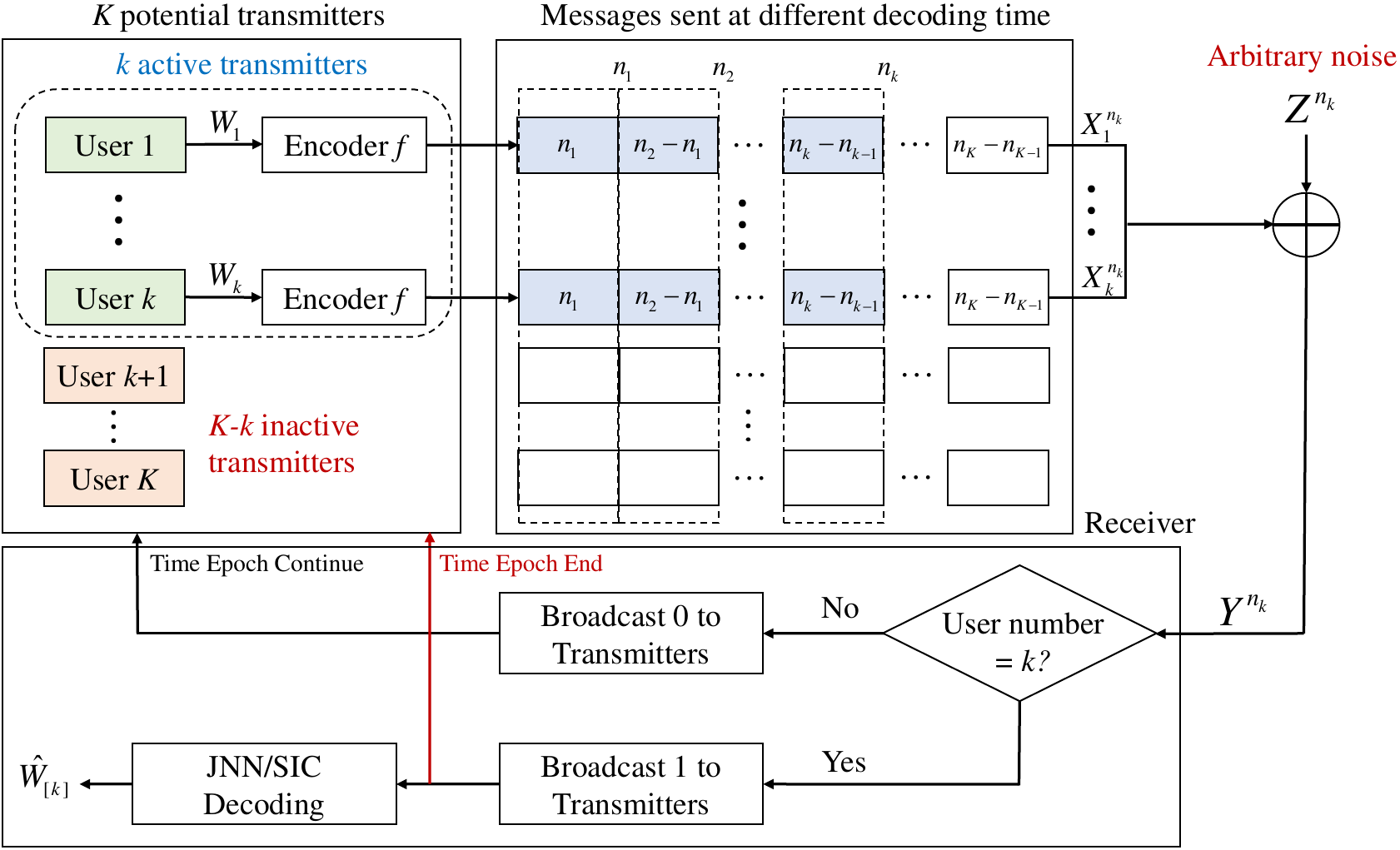}
\caption{Illustration of a RAC applying rateless code with JNN/SIC decoder. }
\label{fig:RAC system model}
\end{figure}

Fix three integers $(K,k,t)\in\bbN_+^3$ and a sequence of stopping times (blocklengths) $\bn:=(n_1,\ldots,n_K)\in\bbN^K$. As shown in Fig.~\ref{fig:RAC system model}, we study the random access channel with $K$ transmitters. In each time epoch, $k$ out of all $K$ transmitters are active and the other $K-k$ transmitters are inactive. Only the active users transmit messages. Without loss of generality, assume that the first $k$ users are active.  The receiver is assumed to be ignorant of the number $k$ and the identities of active users.

The rateless code~\cite{yavas2021gaussianmac} is used by the transmitters and the receiver. In particular, all active transmitters use the same codebook and encoder. For each $t\in[K]$, user $t$ generates a codeword of length $n_K$ and divides the codeword into $K$ sub-codewords with lengths $(n_1, n_2-n_1, \ldots, n_K-n_{K-1})$, respectively. The transmission process is separated into at most $K$ stages, and a one-bit feedback signal is sent from the receiver after each stage. In the first stage, each active user transmits the first sub-codeword of length $n_1$. After receiving a noisy channel output $Y^{n_1}$ that is corrupted by additive arbitrary noise $Z^{n_1}$, the receiver checks whether the estimated number of active users is one. If yes, a one-bit feedback $1$ is sent to all users to terminate the transmission for the current epoch. Otherwise, the receiver broadcasts the other one-bit feedback $0$, signaling each active user to continue the transmission process. For each $t\in[2:K]$, the active users transmitter their $t$-th sub-codewords of length $n_t-n_{t-1}$. The codewords are corrupted by the additive arbitrary noise $Z^{n_t-n_{t-1}}$, yielding noisy channel output $Y^{n_t-n_{t-1}}$. Using $Y^{n_t-n_{t-1}}$, the receiver first checks whether the estimated number of active users is $t$. If yes, the receiver terminates the transmission procedure by broadcasting a feedback signal $1$. Otherwise, the receiver broadcasts $0$ to continue the transmission process until the maximal decoding time $n_K$. When the transmission process is terminated, the receiver uses either a JNN or SIC decoder to estimate the transmitted messages of all active users. It was shown in~\cite[Lemma 1]{yavas2021rac} that channel quality scales inversely with the number of active users under the conditions that the users adopt the same codebook and encoder, and the channel is invariant to permutation of the channel inputs. Our mismatched setting satisfies this condition. Thus, as the active user number increases, to compensate for the deteriorating channel conditions, more information needs to be encoded, which necessitates the hierarchical structure that $n_1<n_2\ldots<n_K$. For each $k\in[K]$, $n_k$ can be expressed in terms of $n_1$, following a formulation analogous to the results of Gaussian RAC~\cite[Eq. (37)]{yavas2021gaussianmac}.

Fix any $k\in[K]$. Assume $k$ active users transmit messages $(W_1,\ldots,W_k)\in[M]^K$. For each $t\in[k]$, the received channel output at time $n_t$ satisfies
\begin{align}
\label{eq:RAC channel output}
Y^{n_t}=\sum_{i=1}^{k}X^{n_t}(W_i)+Z^{n_t},
\end{align}
where $(X^{n_t}(W_1),\ldots,X^{n_t}(W_k))$ are the first $n_t$ symbols of codewords sent by active users, and $Z^{n_t}$ is the additive noise generated i.i.d. from an unknown distribution $P_Z$ satisfying moment constraints in~\eqref{eq:noise statics}.

Fix $K$ positive real numbers $(\lambda_1,\ldots,\lambda_K)\in\bbR_+^K$ and a power constraint $P\in\bbR_+$. Recall that $\bn=(n_1,\ldots,n_K)$ are $K$ fixed blocklengths.

When JNN decoding is used, the mismatched code for RAC is defined as follows. 
\begin{definition}
\label{def:RAC JNN coding scheme}
An $(\bn,M,P)$-JNN code for a RAC consists of:
\begin{itemize}
\item a random codebook $\calC$ with $M$ codewords $\{X^{n_K}(1),\ldots, X^{n_K}(M)\}$ subject to a maximal power constraint $P$,
\item an encoder $f$ such that for each message $w\in[M]$,
\begin{align}
f(w):=X^{n_K}(w).
\end{align}
\item a sequence of $K$ decoders $(\Psi_1,\ldots,\Psi_K)$. For each $t\in[K]$, using $Y^{n_t}$, the decoder $\Phi_t$ first checks whether the estimated number of users equals $t$ by verifying the inequality $\rho_t(Y^{n_t}):=\big|\frac{1}{n_t}\|Y^{n_t}\|^2-(1+tP)\big|\leq\lambda_t$. If the inequality is not satisfied, the decoder notifies all active users to continue transmission. Otherwise, the decoder notifies all users to terminate transmission and applies the following JNN decoder:
\begin{align}\label{eq:RAC JNN decoder}
(\hatW_1,\ldots,\hatW_t)=\Psi_t(Y^{n_t}):=\argmin_{(\barw_1,\ldots,\barw_t)\in[M]^t}\bigg\|{Y^{n_t}-\sum_{i=1}^{t} X^{n_t}(\barw_i)}\bigg\|^2.
\end{align}
\end{itemize}
\end{definition}

When SIC decoding is used, the mismatched code for RAC is defined as follows.
\begin{definition}\label{def:RAC SIC coding shceme}
An $(\bn,M,P)$-SIC code is the same as the code in Def.~\ref{def:RAC JNN coding scheme} except that the JNN decoder in~\eqref{eq:RAC JNN decoder} is replaced by the following SIC decoder. The first message is decoded using nearest neighbor (NN) decoding and treating interference as noise (TIN):
\begin{align}
\hatW_1:=\argmin_{\barw\in[M]}\big\|Y^{n_k}-X^{n_k}(\barw)\big\|^2.
\end{align}
Subsequently, for each $r\in[2:t]$, the message $W_r$ is decoded using interference cancellation with NN and TIN:
\begin{align}\label{eq:RAC SIC decoder}
\hatW_r:=\argmin_{\barw\in[M]}\Big\|Y^{n_k}-\sum_{i=1}^{r-1}X^{n_k}(\hatW_i)-X^{n_k}(\barw)\Big\|^2.
\end{align}
\end{definition}

All users in a RAC adopt the same spherical codebook with codewords $\{X^{n_K}(1),\ldots, X^{n_K}(M)\}$. For each $w\in[M]$, the codeword $X^{n_K}(w)$ is generated from the same distribution and thus, for simplicity, we use $X^{n_K}$ to denote the codeword. The first $n_1$ symbols $X^{n_1}$ of $X^{n_K}$ are generated from the uniform distribution over the sphere with radius $\sqrt{n_1P}$, i.e., for any $x^{n_1}\in\bbR^{n_1}$,
\begin{align}
\label{eq:first n0 codewords distribution}
P_{X^{n_1}}(x^{n_1}):=\frac{\delta \left(\| x^{n_1} \| ^2 - n_1P\right)}{S_{n_1}(\sqrt{n_1P})}.
\end{align}
For each $j\in[2:K]$, the $n_{j-1}+1$ to the $n_j$ symbols $X^{n_j}_{n_{j-1}+1}$ of $X^{n_K}$,  is generated from the uniform distribution over a sphere with radius $\sqrt{(n_j-n_{j-1})P}$, i.e., for any $x^{n_j}_{n_{j-1}+1}\in\bbR^{n_j-n_{j-1}}$,
\begin{align}
\label{eq:nj-1 to nj codewords distribution}
P_{X^{n_j}_{n_{j-1}+1}}(x^{n_j}_{n_{j-1}+1}):=\frac{\delta (\| x^{n_j}_{n_{j-1}+1} \| ^2 - (n_j-n_{j-1})P)}{S_{n_j-n_{j-1}}(\sqrt{(n_j-n_{j-1})P})}.
\end{align}
In Fig.~\ref{fig:sphericaltype codebook}, we illustrate the spherical-type codebook when $K=2$, $n_1=2$ and $n_2=3$. From the above codeword generation process, one observes that each codeword $X^{n_K}$ consists of $K$ independent sub-codewords.

\begin{figure}[tb]
\centering
\includegraphics[width = .4\columnwidth]{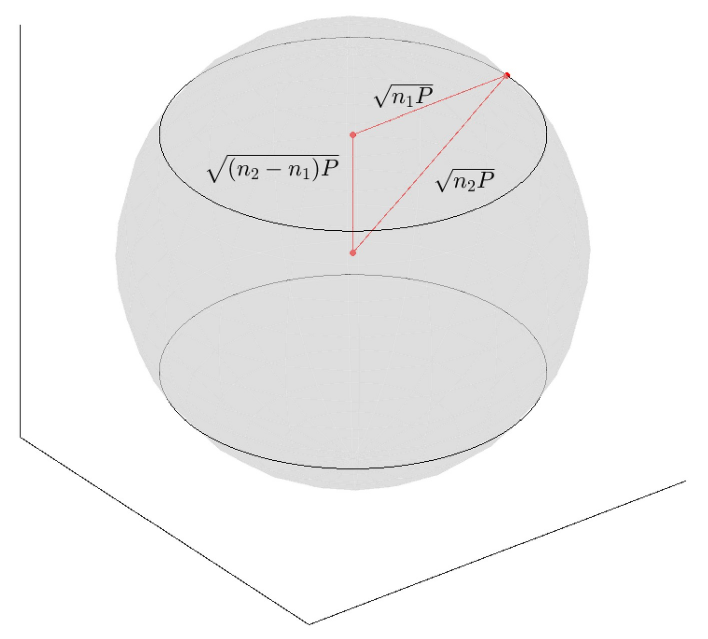}
\caption{Illustration of the spherical-type codebook in~\eqref{eq:first n0 codewords distribution} and~\eqref{eq:nj-1 to nj codewords distribution} when $n_1=2$ and $n_2=3$ (reproduced from~\cite[Fig. 2]{yavas2021gaussianmac}).}
\label{fig:sphericaltype codebook}
\end{figure}

To evaluate the performance of a RAC code, we consider a generalized definition of error probability. When $k\in[K]$ users are active, who transmit messages $(W_1,\ldots,W_k)\in[M]^k$, the error probability is defined as 
\begin{align}
\label{eq:rac error prob}
\rmP_{\rme, k}^n:=\Pr\Big\{\mathrm{perm}(W_1,\ldots,W_k)\neq \mathrm{perm}(\hatW_1,\ldots,\hatW_k)\Big\},
\end{align}
where $\mathrm{perm}$ corresponds to all possible arrangements for the set of messages $(W_1,\ldots,W_k)$. Note that since the identities of users are unknown, a permutation is necessary since the receiver only aims to decode all messages correctly and does not care which user transmits which message. This is consistent with RAC studies in~\cite{yavas2021gaussianmac}.

Since all active users adopt the same codebook and encoder in a RAC, the theoretical benchmark is defined as the maximum codebook size satisfying various error probability constraints under different blocklengths, corresponding to a performance guarantee for all possible numbers of active users. Specifically, given any set of $K$ positive real numbers $\boldsymbol{\varepsilon}:=(\varepsilon_1,\ldots,\varepsilon_K)\in(0,1)^K$ and $\dagger\in(\mathrm{SIC},\mathrm{JNN})$, the finite blocklength theoretical benchmark for a mismatched RAC is defined as
\begin{align}
M^*_\dagger(\bn,\boldsymbol{\varepsilon})
&:=\sup\big\{M\in\bbN_+:~\exists~\mathrm{an~}(\bn,M,P)\mathrm{-}\dagger\mathrm{-code.~s.t.~}\forall~k\in[K],~\rmP_{\rme, k}^{n_k}\leq \varepsilon_k\big\}.
\end{align}
Note that $M_\dagger^*(\bn,\boldsymbol{\varepsilon})$ represents the maximum codebook size for any mismatched RAC code that ensures the error probability is bounded by $\varepsilon_k$ when there are $k$ active users, simultaneously for all $k\in[K]$.

The theoretical benchmark for mismatched RAC differs significantly from that of MAC. This difference arises because, in a MAC, each user employs a distinct codebook and encoder, resulting in different message rates and a multidimensional rate region. In contrast, in a RAC, all active users share the same codebook and encoder, enforcing each user to have the same message rate. Consequently, the theoretical benchmark in a RAC is a one-dimensional value, similar to P2P channel coding.

\section{Main Results and Discussions}
\label{sec: main results and discussions}
In this section, we present second-order achievable rate regions for two-user MAC and RAC with additive arbitrary noise under both JNN and SIC decoding. In particular, we show that JNN and SIC decoding lead to vastly different performance for MAC and RAC. Specifically, both JNN and SIC decoding achieve the first-order rate region for MAC, while JNN achieves a strictly larger first-order rate for RAC. Furthermore, JNN decoding achieves a larger second-order rate region for MAC.

\subsection{Preliminaries}
The following definitions are needed. Given any integer $d\in\bbN_+$, positive semi-definite matrix $\bV$ of dimension $d\times d$ and positive real number $\varepsilon\in(0,1)$, the ccdf of a $d$-dimensional Gaussian random variable $\bS_d\sim\calN(\bzero_d,\bV)$ is defined as~\cite{tan2014mono,zhou2017sr}:
\begin{align}
\rmQ_{\rm inv}(\bV,\varepsilon):=\{ a^d=(a_1,\ldots,a_d) \in \bbR^d: \Pr\{\bS_d\geq a^d \}\geq 1-\varepsilon\}\label{eq:define Qinv}.
\end{align}

Fix power constraints $(P,P_1,P_2)\in\bbR_+^3$. Recall that the noise distribution satisfies the moment constraints in \eqref{eq:noise statics}, especially that the fourth moment is $\xi$. Define the following dispersion functions:
\begin{align}
V(P)&:=\frac{(\xi-1)P^2+4P}{4(1+P)^2},\label{eq:dispersion v}\\
V_{1,2}(P_1,P_2) &:= \frac{(\xi -1)P_1P_2}{4(1+P_1)(1+P_2)},\label{eq:dispersion term} \\
V_{i,12}(P_1,P_2) &:= \frac{(\xi -1)P_i(P_1+P_2)+4P_i}{4(1+P_i)(1+P_1+P_2)},\\
V_{12}(P_1,P_2) &:= V(P_1+P_2)+\frac{P_1P_2}{(1+P_1+P_2)^2},\\
V_{i}(P_1,P_2) &:= \frac{P_i^2(\xi-1+4P_{3-i})+4P_i(1+P_{3-i})^3}{{4(1+P_{3-i})^2(1+P_1+P_2)^2}}.
\end{align}
Furthermore, define the following dispersion matrices:
\begin{align}
\bV_{1}(P_1,P_2)&:=\begin{bmatrix}
V(P_2) & V_{2,12}(P_1,P_2) \\
V_{2,12}(P_1,P_2) & V_{12}(P_1,P_2)
\end{bmatrix},\label{eq:dispersion matrix 1}\\
\bV_{2}(P_1,P_2)&:=\begin{bmatrix}
V(P_1) & V_{1,12}(P_1,P_2) \\
V_{1,12}(P_1,P_2) & V_{12}(P_1,P_2)
\end{bmatrix},\label{eq:dispersion matrix 2}\\
\bV(P_1,P_2)&:=
\begin{bmatrix}
V(P_1) & V_{1,2}(P_1,P_2) & V_{1,12}(P_1,P_2) \\
V_{1,2}(P_1,P_2) & V(P_2) & V_{2,12}(P_1,P_2) \\
V_{1,12}(P_1,P_2) & V_{2,12}(P_1,P_2) & V_{12}(P_1,P_2)
\end{bmatrix}.\label{eq:dispersion matrix}
\end{align}
The above dispersion functions and matrices are used to characterize achievable second-order asymptotics for MAC and RAC.

\subsection{Two-User MAC}
\label{sec:mac results}
We first present an approximation to the non-asymptotic performance of JNN decoding up to second-order. Recall that $\rmC(\cdot)$ is the capacity of an AWGN channel in~\eqref{eq:capacity def}.
\begin{theorem}
\label{theo: mac jnn achie}
For any $\varepsilon \in (0,1)$, there exists an $(n,M_1,$ $M_2,P_1,P_2)$-JNN code such that for $\rmP_\rme^n\leq \varepsilon$,
\begin{align}
\label{eq:mac jnn theo 2}
\begin{bmatrix}
\log M_1 \\
\log M_2 \\
\log M_1M_2
\end{bmatrix}
\in n\bC(P_1,P_2)-\sqrt{n}{\rm Q}_{\rm inv}(\bV(P_1,P_2),\varepsilon)+O(n^{1/4})\mathbf{1}_3,
\end{align}
where
\begin{align}
\bC(P_1,P_2):=
[\rmC(P_1), \rmC(P_2), \rmC(P_1+P_2)]^\rmT.\label{eq:capacity vector def}
\end{align}
\end{theorem}

The proof of Theorem~\ref{theo: mac jnn achie} is provided in Section~\ref{sec:mac jnn proof}, which judiciously combines the proof techniques for AWGN RAC in ~\cite{molavianjazi2015second,yavas2021gaussianmac} and the mismatched P2P channel~\cite{scarlett2017mismatch}.

We make a few remarks. Fix $\alpha\in(0,1)$ and let $\bar{\alpha}=1-\alpha$. First, the result in Theorem~\ref{theo: mac jnn achie} can be written in the second-order asymptotic expression per Def.~\ref{def:mac second order region}, which is divided into five cases, as shown below.
\begin{corollary}
\label{CORO:MAC JNN ACHIEVABILITY}
Depending on the values of boundary rate pairs $(R_{1}^*,R_{2}^*)$ of the first-order achievable rate region $\calR_\mathrm{JNN}$ in~\eqref{eq:MAC first order region}, the second-order achievable rate region $\calL_{\rm JNN}(R_{1}^*,R_{2}^*,\varepsilon)$ satisfies
\begin{itemize}
\item Case (i): $(R_{1}^*,R_{2}^*)=(\alpha \rmC(P_1/(1+P_2)),\rmC(P_2))$
\begin{align}
\calL_{\rm JNN}(R_{1}^*,R_{2}^*,\varepsilon)\supseteq\{(L_1,L_2):L_2\geq\sqrt{V(P_2)}\rmQ^{-1}(\varepsilon)\}\label{eq:MAC JNN case i}.
\end{align}
\item Case (ii): $(R_{1}^*,R_{2}^*)=(\rmC(P_1/(1+P_2)),\rmC(P_2))$
\begin{align}
\calL_{\rm JNN}(R_{1}^*,R_{2}^*,\varepsilon)\supseteq\{(L_1,L_2):(L_2,L_1+L_2)\in \rmQ_{\rm inv}(\bV_{1}(P_1,P_2),\varepsilon)\}\label{eq:MAC JNN caseii}.
\end{align}
\item Case (iii): $(R_{1}^*,R_{2}^*)=(\alpha \rmC(P_1)+\bar{\alpha}\rmC(P_1/(1+P_2)),\alpha \rmC(P_2/(1+P_1))+\bar{\alpha}\rmC(P_2))$
\begin{align}
\calL_{\rm JNN}(R_{1}^*,R_{2}^*,\varepsilon)\supseteq\{(L_1,L_2):L_1+L_2\geq\sqrt{V_{12}(P_1,P_2)}\rmQ^{-1}(\varepsilon)\}\label{eq:MAC JNN caseiii}.
\end{align}
\item Case (iv): $(R_{1}^*,R_{2}^*)=(\rmC(P_1),\rmC(P_2/(1+P_1)))$
\begin{align}
\calL_{\rm JNN}(R_{1}^*,R_{2}^*,\varepsilon)\supseteq\{(L_1,L_2):(L_1,L_1+L_2)\in \rmQ_{\rm inv}(\bV_{2}(P_1,P_2),\varepsilon)\}.
\end{align}
\item Case (v): $(R_{1}^*,R_{2}^*)=( \rmC(P_1),\alpha \rmC(P_2/(1+P_1)))$
\begin{align}
\calL_{\rm JNN}(R_{1}^*,R_{2}^*,\varepsilon)\supseteq\{(L_1,L_2):L_1\geq\sqrt{V(P_1)}\rmQ^{-1}(\varepsilon)\}.
\end{align}
\end{itemize}
\end{corollary}
The results in Corollary~\ref{CORO:MAC JNN ACHIEVABILITY} with five cases enable us to compare the second-order achievable rate regions of JNN and SIC decoding, as presented in Section~\ref{sec:comparison jnn and sic mac}.

Theorem~\ref{theo: mac jnn achie} generalizes the result of the Gaussian MAC in~\cite{molavianjazi2015second} to the case with additive arbitrary noise. When specialized to the Gaussian MAC, $\xi=3$, and the above results are consistent with~\cite[Theorem 1]{molavianjazi2015second}. Although the more refined third-order asymptotic result was characterized~\cite[Theorem 2]{yavas2021gaussianmac}, the result was established for the matched case with Gaussian noise. In contrast, the result in Theorem \ref{theo: mac jnn achie} holds for additive arbitrary noise satisfying the moment constraints, and its proof is based on geometrical analysis, which does not rely on a particular noise distribution.

The proof of Theorem~\ref{theo: mac jnn achie} decomposes the ensemble error probability into probabilities of three error events: $\calE_1$, $\calE_2$, and $\calE_{1,2}$. Specifically, for $j\in[2]$, let $\calE_j$ denote the error event where only message $W_j$ is decoded incorrectly, and let $\calE_{1,2}$ denote the error event where both messages $W_1$ and $W_2$ are decoded incorrectly. Although the proof of Theorem \ref{theo: mac jnn achie} generally follows~\cite{molavianjazi2015second} and~\cite{yavas2021gaussianmac}, a critical step~\cite[Lemma 6]{yavas2021gaussianmac} was established for Gaussian noise. To resolve the above problem, we generalize~\cite[Lemma 6]{yavas2021gaussianmac} to arbitrary noise distribution, although with a looser bound. This way, we can upper bound the probability of three error events as desired and obtain the second-order asymptotics by applying the function version Berry-Esseen Theorem~\cite[Prop. 1]{molavianjazi2015second} (see Lemma~\ref{lemma:berry esseen for func}).

Unfortunately, we could not derive a matching ensemble converse result. The main challenge is clarified as follows. While the error probability concerning separate message error events $(\calE_1,\calE_2)$ can be lower bounded following the techniques of the mismatched P2P channel~\cite[Section IV]{scarlett2017mismatch}, lower bounding the probability of the error event $\calE_{1,2}$ suffers technical challenges. In particular, the analysis requires evaluating the probability density function (pdf) for the sum of two spherical codewords with powers $P_1$ and $P_2$, respectively. However, the summation of these two spherical codewords does not yield a spherical codeword with desired power $P_1+P_2$. Consequently, establishing a matching ensemble converse requires additional efforts beyond \cite{scarlett2017mismatch} and is left as future work.

Given any two real numbers $(a,b)\in\bbR^2$, define the following set:
\begin{align}
\calA(a,b)
&:=\{(L_1,L_2)\in\bbR^2:~L_1\geq a, L_2\geq b\}\label{def:calAab}.
\end{align}
When SIC decoding is used, the corresponding achievable second-order asymptotic result is as follows.
\begin{theorem}
\label{theorem:MAC-SIC ach}
Depending on the values of boundary rate pairs $(R_1^*,R_2^*)$ of the first-order achievable rate region $\calR_\mathrm{SIC}$ in~\eqref{eq:MAC first order region}, the second-order achievable rate region $\calL_{\rm SIC}(R_1^*,R_2^*,\varepsilon)$ satisfies
\begin{itemize}
\item Case (i): $(R_1^*,R_2^*)=(\alpha  C(P_1/(1+P_2)), C(P_2))$
\begin{align}
\label{eq:MAC SIC case i}
\calL_{\rm SIC}(R_1^*,R_2^*,\varepsilon)
&\supseteq\bigcup_{\substack{(\varepsilon_1,\varepsilon_2)\in\bbR_+^2:\\\varepsilon_1+\varepsilon_2\leq \varepsilon}}\calA\Big(\alpha\sqrt{V_1(P_1,P_2)}\rmQ^{-1}(\varepsilon_1),\sqrt{V(P_2)}\rmQ^{-1}(\varepsilon_2)\Big).
\end{align}
\item Case (ii): $(R_1^*,R_2^*)=( C(P_1/(1+P_2)), C(P_2))$
\begin{align}
\label{eq:MAC SIC case ii}
\calL_{\rm SIC}(R_1^*,R_2^*,\varepsilon)\supseteq\bigcup_{\substack{(\varepsilon_1,\varepsilon_2)\in\bbR_+^2:\\\varepsilon_1+\varepsilon_2\leq \varepsilon}}\calA\Big(\sqrt{V_1(P_1,P_2)}\rmQ^{-1}(\varepsilon_1),\sqrt{V(P_2)}\rmQ^{-1}(\varepsilon_2)\Big).
\end{align}
\item Case (iii): $(R_1^*,R_2^*)=(\alpha  C(P_1)+\bar{\alpha} C(P_1/(1+P_2)),\alpha  C(P_2/(1+P_1))+\bar{\alpha} C(P_2))$
\begin{align}
\nn\calL_{\rm SIC}(R_1^*,R_2^*,\varepsilon)\\*
\nn&\supseteq\bigcup_{\substack{(\varepsilon_1,\varepsilon_2)\in\bbR_+^2:\\\varepsilon_1+\varepsilon_2\leq \varepsilon}}\calA\Big((\alpha\sqrt{V_1(P_1,P_2)}+\bar{\alpha}\sqrt{V(P_1)})\rmQ^{-1}(\varepsilon_1),\\*
&\qquad\qquad(\alpha\sqrt{V(P_2)}+\bar{\alpha}\sqrt{V_2(P_1,P_2)})\rmQ^{-1}(\varepsilon_2)\Big).
\end{align}
\item Case (iv): $(R_1^*,R_2^*)=( C(P_1), C(P_2/(1+P_1)))$
\begin{align}
\calL_{\rm SIC}(R_1^*,R_2^*,\varepsilon)\supseteq\bigcup_{\substack{(\varepsilon_1,\varepsilon_2)\in\bbR_+^2:\\\varepsilon_1+\varepsilon_2\leq \varepsilon}}\calA\Big(\sqrt{V(P_1)}\rmQ^{-1}(\varepsilon_1),\sqrt{V_2(P_1,P_2)}\rmQ^{-1}(\varepsilon_2)\Big).
\end{align}
\item Case (v): $(R_1^*,R_2^*)=(  C(P_1),\alpha  C(P_2/(1+P_1)))$
\begin{align}
\calL_{\rm SIC}(R_1^*,R_2^*,\varepsilon)\supseteq\bigcup_{\substack{(\varepsilon_1,\varepsilon_2)\in\bbR_+^2:\\\varepsilon_1+\varepsilon_2\leq \varepsilon}}\calA\Big(\sqrt{V(P_1)}\rmQ^{-1}(\varepsilon_1),\alpha\sqrt{V_2(P_1,P_2)}\rmQ^{-1}(\varepsilon_2)\Big).
\end{align}
\end{itemize}
\end{theorem}
For cases (i) and (v), we exclude $\alpha=0$ and $\alpha=1$, respectively, since these two cases degenerate to the P2P channel characterized in~\cite[Theorem 1]{scarlett2017mismatch}. The proof of case (ii) is available in Section~\ref{sec:mac sic proof}, which combines the proof techniques for mismatched channel coding and mismatched interference channel coding~\cite{scarlett2017mismatch}. As shown in~\cite[Proposition 4.1]{el2011network}, if time sharing two coding schemes with achievable rate pairs $(R_{11},R_{21})$ and $(R_{12},R_{22})$, then the rate pair $(R_1,R_2)=(\alpha R_{11}+\bar{\alpha}R_{12},\alpha R_{21}+\bar{\alpha}R_{22})$ is also achievable for every $\alpha\in[1]$. Specifically, case (i) is proved by time sharing the results for mismatched channel coding, and case (ii), where $\alpha$ denotes the percentage of time that the code for case (ii) is used. Case (iv) is symmetric to case (ii), and case (v) is symmetric to case (i). Case (iii) is proved by time sharing the results of cases (ii) and (iv), where $\alpha$ denotes the percentage of total time where the code for case (iv) is used. 

In a nutshell, when a coding scheme that has guaranteed performance for Gaussian noise is used as a mismatched code, the same first-order asymptotics can be achieved for any additive noise with the same second moment, and the achievable second-order asymptotic result depends on the particular noise distribution through the fourth moment. Since the mismatched code is designed without knowledge of the exact noise distribution, its performance might be far away from the optimal code of the matched case. However, the advantage of a mismatched code is its general validity and robustness.

In the next subsection, we comprehensively compare the performance of JNN and SIC decoding for mismatched MAC.

\subsection{Comparison of Second-order Achievable Rate Regions of JNN and SIC Decoding for MAC}
\label{sec:comparison jnn and sic mac}

Comparing Corollary~\ref{CORO:MAC JNN ACHIEVABILITY} and Theorem~\ref{theorem:MAC-SIC ach}, we find that although both JNN and SIC decoding achieve the same first-order rate region for a two-user mismatched MAC, JNN decoding has a strictly larger second-order achievable rate region. The detailed explanation is as follows.

Due to the symmetry of cases (i) and (v), and cases (ii) and (iv), it suffices to consider cases (i), (ii), and (iii). Given $\dag\in(\rm i,\rm ii,\rm iii)$, let $\calL_{\rm JNN, \dag}$ and $\calL_{\rm SIC,\dag}$ denote inner bounds for the achievable second-order rate region $\calL_{\rm JNN}(R_1^*,R_2^*,\varepsilon)$ and $\calL_{\rm SIC}$ $(R_1^*,R_2^*,\varepsilon)$ in Corollary~\ref{CORO:MAC JNN ACHIEVABILITY} and Theorem~\ref{theorem:MAC-SIC ach}, respectively.

First, consider case (i). For JNN decoding, it follows from~\eqref{eq:MAC JNN case i} that $L_1$ can be arbitrarily small. In contrast, for SIC decoding, it follows from~\eqref{eq:MAC SIC case i} that $L_1$ is lower bounded by $\alpha\sqrt{V_1(P_1,P_2)}\rmQ^{-1}(\varepsilon_1)$. The lower bounds for $L_2$ in JNN and SIC decoding only differ in the parameter of the inverse of ccdf $\rmQ^{-1}(\cdot)$. Since $\rmQ^{-1}(\cdot)$ monotonically decreases in its parameter, it follows from the constraint $\varepsilon\geq\varepsilon_1$ that $\sqrt{V(P_2)}\rmQ^{-1}(\varepsilon_1)\geq\sqrt{V(P_2)}\rmQ^{-1}(\varepsilon)$.
Thus, for case (i), the second-order achievable rate region of JNN decoding is larger, i.e., $\calL_{\rm SIC, i} \subseteq \calL_{\rm JNN, i}$.

Now consider case (ii).Recall the definition of the dispersion function $\bV_1(P_1,P_2)$ in~\eqref{eq:dispersion matrix 1} and define the two dimensional Gaussian random vector $[K_1;K_2]\sim\calN(\cdot;\bzero_2,\bV_1(P_1,P_2))$. It follows from~\eqref{eq:MAC JNN caseii} that the second-order rate region of JNN decoding is equivalently given by
\begin{align}
\label{eq:rephrase case ii}
\calL_{\rm JNN, ii}=\big\{(L_1,L_2):\Pr\{\{K_1\geq L_2\}\cup \{K_2\geq (L_1+L_2)\}\}\leq\varepsilon\big\}.
\end{align}
First, consider extreme cases where one of the second-order rates tends to infinity. When $L_1\rightarrow\infty$, $\Pr\{K_2\leq (L_1+L_2)\}\rightarrow 1$, and
\begin{align}
\calL_{\rm JNN, ii}=\big\{(L_1,L_2):L_2\geq\sqrt{V(P_2)}\rmQ^{-1}(\varepsilon)\big\}.
\end{align}
The above region matches the second-order achievable rate region $\calL_{\rm SIC, ii}$ of SIC decoding in~\eqref{eq:MAC SIC case ii} when $\varepsilon_2\rightarrow\varepsilon$. However, if $\varepsilon_2<\varepsilon$, $\calL_{\rm SIC, ii}\subseteq\calL_{\rm JNN, ii}$ since $\rmQ^{-1}(\varepsilon_2)>\rmQ^{-1}(\varepsilon)$. When $L_2\rightarrow\infty$, the probability term in~\eqref{eq:rephrase case ii} does not depend on $L_1$ anymore, which implies that $L_1$ can be arbitrarily small in $\calL_{\rm JNN, ii}$ for JNN decoding. In contrast, $L_1$ is strictly lower bounded by $\sqrt{V_1(P_1,P_2)}\rmQ^{-1}(\varepsilon_1)$ in~\eqref{eq:MAC SIC case ii} for SIC decoding. Thus, $\calL_{\rm SIC, ii}\subseteq\calL_{\rm JNN, ii}$ when either $L_1$ or $L_2$ is unbounded.

When the second-order rates $L_1$ and $L_2$ are both finite, the comparison is non-trivial. Fix $\varepsilon_1\in(0,1)$ such that $\varepsilon_1<\varepsilon$ and let $L'_1:=\sqrt{V_1(P_1,P_2)}\rmQ^{-1}(\varepsilon_1)$. The set of $L_2$ such that $(L_1',L_2)\in\calL_{\rm JNN, ii}$ is given by
\begin{align}
\calL_{{\rm JNN, ii}, L_2}&=\{L_2:\Pr\{\{K_1\geq L_2\}\cup \{K_2-L'_1\geq L_2\}\}\leq\varepsilon\}\\*
&=\{L_2:\Pr\{\max\{K_1,K_2-L'_1\}\geq L_2\}\leq\varepsilon\}.\label{eq:two gaussian take max}
\end{align}
Correspondingly, the set of $L_2$ such that $(L_1',L_2)\in\calL_{\rm SIC,ii}$ satisfies
\begin{align}
\calL_{{\rm SIC, ii}, L_2}
&=\big\{L_2:~\Pr\{K_1\geq L_2\}\leq \varepsilon_2\big\}.\label{eq:Ljnn ii l2}
\end{align}
The comparison between~\eqref{eq:two gaussian take max} and~\eqref{eq:Ljnn ii l2} is mathematically challenging. Instead, a numerical example can help. Specifically, in Fig.~\ref{fig:caseii}, the second-order rate regions of JNN and SIC decoding are plotted when the boundary rate pair is $(R_1^*,R_2^*)=( C(P_1/(1+P_2)), C(P_2))$, the tolerable error probability is $\varepsilon=0.2$ the power constraints are $P_1=5, P_2=2$ and the fourth moment $\xi=3$. The region above the blue line is the second-order achievable rate region $\calL_{\rm JNN, ii}$ of JNN decoding, and the region above the red dashed line is the second-order achievable rate region $\calL_{\rm SIC, ii}$ of SIC decoding. The pink and deep blue circles denote the cases when $\varepsilon_1\rightarrow0$ and $\varepsilon_2\rightarrow0$, respectively. As observed from Fig.~\ref{fig:caseii}, when $L_1\rightarrow\infty$ or $L_2\rightarrow\infty$, the second-order rates of both JNN and SIC decoding are the same, which is consistent with the theoretical analysis. Furthermore, for a finite $L_1$, the minimum value of $L_2$ in $\calL_{\rm SIC, ii}$ for SIC decoding is strictly smaller than the corresponding minimal value of $L_2$ in $\calL_{\rm JNN, ii}$ for JNN decoding. The above numerical result implies that JNN decoding can be strictly better in terms of low-latency communication performance.

\begin{figure*}[tb]
\centering
\subfigure[Case (ii)]
{
\begin{minipage}{.4\linewidth}
\centering
\includegraphics[scale=0.48]{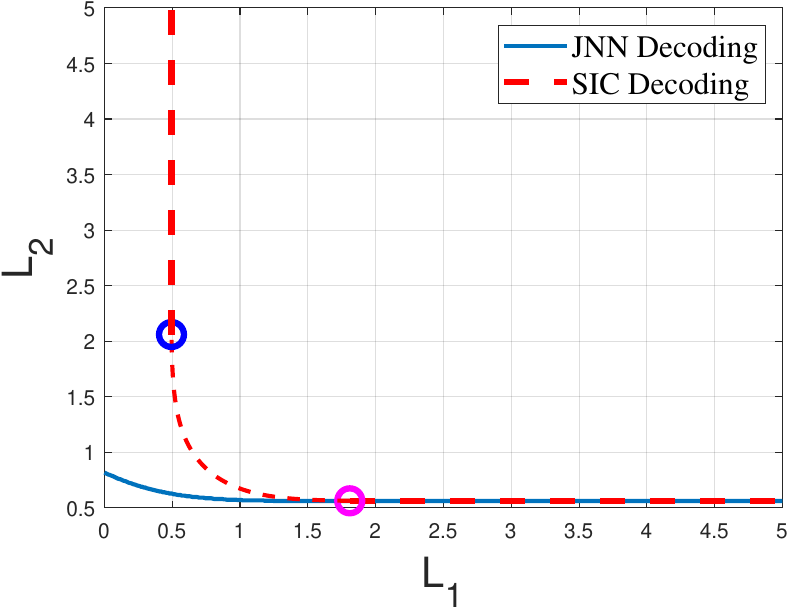}
\label{fig:caseii}
\end{minipage}
}
\subfigure[Case (iii), $\alpha =0.5$]
{
\begin{minipage}{.4\linewidth}
\centering
\includegraphics[scale=0.48]{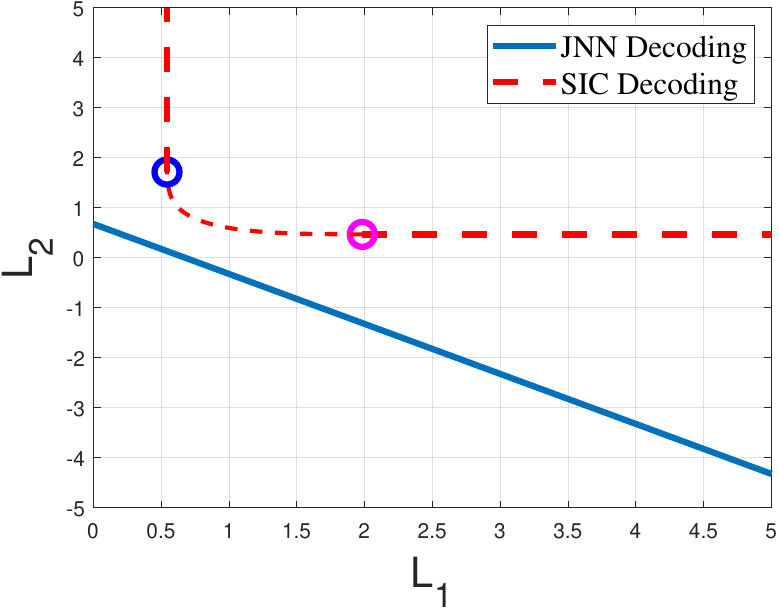}
\label{fig:caseiii}
\end{minipage}
}
\caption{Comparison of the second-order achievable rate regions $\calL_{\rm JNN}$ and $\calL_{\rm SIC}$ when $\varepsilon=0.2, P_1=5, P_2=2$ and $\xi=3$.}
\end{figure*}

For case (iii), the second-order achievable rate region of JNN in~\eqref{eq:MAC JNN caseiii} can be equivalently written as
\begin{align}
\label{eq:rephrase case iii}
\calL_{\rm JNN,iii}=\{(L_1,L_2): \Pr\{K_2\geq L_1+L_2  \}\leq\varepsilon \}.
\end{align}
It follows from~\eqref{eq:rephrase case ii} and~\eqref{eq:rephrase case iii} that $\calL_{\rm JNN, ii}\subseteq\calL_{\rm JNN,iii}$. This is because for any $(L_1,L_2)\in\bbR^2$, the inequality $\Pr\{K_2\geq L_1+L_2\}\leq\Pr\{\{K_1\geq L_2\}\cup \{K_2\geq (L_1+L_2)\}\}$ holds. Analogously, we have $\calL_{\rm JNN,iv}\subseteq\calL_{\rm JNN,iii}$. Furthermore, for SIC decoding, $\calL_{\rm SIC, iii}$ is a linear combination of $\calL_{\rm SIC, ii}$ and $\calL_{\rm SIC, iv}$. In Fig. \ref{fig:caseii}, we have shown numerically that $\calL_{\rm SIC}\subseteq\calL_{\rm JNN}$ for case (ii), which is also true for case (iv) due to symmetry. Thus, the same conclusion holds for case (iii) due to the above arguments. To clarify, we plot the second-order rate regions of JNN and SIC decoding for the same setting as Fig. \ref{fig:caseii} when $\alpha=0.5$ in Fig.~\ref{fig:caseiii}. As observed, the second-order rate region $\calL_{\rm SIC, iii}$ of SIC decoding, denoted by the region above the red dashed line, falls strictly inside the second-order rate region $\calL_{\rm JNN, iii}$ of JNN decoding, denoted by the region above the blue line.

In summary, the low-latency performance of JNN decoding can be strictly better than SIC decoding, although both decoding algorithms achieve the same asymptotic performance as the blocklength (latency) tends to infinity.

\subsection{RAC with Unknown User Activity Pattern}
Fix any power constraint $P\in\bbR_+$. Recall the definition of dispersion function $V(P)$ in~\eqref{eq:dispersion v}. For each $k\in[K]$ and $r\in[k]$, define the following dispersion functions:
\begin{align}
V_{\rm{cr}}(k,P)&:=\frac{k(k-1)P^2}{2(1+kP)^2}\label{eq:Vcr},\\
V_{\rm{rs}}(k,r,P)\nn&:=\frac{1}{4(1+(k-r+1)P)^2(1+(k-r)P)^2}(P^2(\xi-1)+4P(1+(k-r)P)^3\\*
&\qquad+4(k-r)P^3+2(k-r)(k-r-1)P^4).\label{eq:Vrs}
\end{align}
Recall that $\bn=(n_1,\ldots,n_K)\in\bbN^K$ are given blocklengths and $\boldsymbol{\varepsilon}=(\varepsilon_1,\ldots,\varepsilon_K)\in(0,1)^K$ are tolerable error probabilities. 
\begin{theorem}
\label{THEOREM:RAC JNN ACHIEVABILITY}
The finite blocklength transmission rate of a mismatched RAC with JNN decoding satisfies 
\begin{align}
\label{eq:rac jnn ach}
\log M^*(\bn,\boldsymbol{\varepsilon})&\geq \min_{k\in[K]}\frac{1}{k}\big(n_k  C(kP)-\sqrt{n_k(V(kP)+V_{\rm{cr}}(k,P))}\rmQ^{-1}(\varepsilon_{k})+O(\log n_k)\big).
\end{align}
\end{theorem}
The proof of Theorem~\ref{THEOREM:RAC JNN ACHIEVABILITY} is provided in Section~\ref{sec:rac jnn proof}, which combines the proof techniques for Gaussian RAC in~\cite{yavas2021gaussianmac} and mismatched P2P channel in~\cite{scarlett2017mismatch}. Theorem~\ref{THEOREM:RAC JNN ACHIEVABILITY} lower bounds the second-order achievable rate of a mismatched RAC using a JNN decoder. Our result holds for \emph{arbitrary} additive noise satisfying moment constraints~\eqref{eq:noise statics} and better suits the needs of practical communications systems. When specialized to the Gaussian noise, we have $\xi=3$, and Theorem~\ref{THEOREM:RAC JNN ACHIEVABILITY} implies that the coding scheme in Def.~\ref{def:RAC JNN coding scheme} achieves the same second-order rate as the matched case~\cite[Theorems 4]{yavas2021gaussianmac}. Note that we do not consider the case with $k=0$ as in~\cite{yavas2021gaussianmac} because it is the case when there is no active user in the current time epoch, and it affects neither the achievability results nor the comparison between JNN and SIC decoding.

When SIC decoding is used, the corresponding result is as follows.
\begin{theorem}
\label{THEOREM:RAC SIC ACHIEVABILITY}
The finite blocklength transmission rate of a mismatched RAC with SIC decoding satisfies 
\begin{align}
\label{eq:rac sic ach}
\log M^*(\bn,\boldsymbol{\varepsilon})\geq \min_{k\in[K]} \bigg( n_kC\bigg(\frac{P}{1+(k-1)P}\bigg)-\sqrt{n_kV_{\rm{rs}}(k,1,P)}\rmQ^{-1}(\varepsilon_{k})+O(\log n_k)\bigg).
\end{align}
\end{theorem}
The proof of Theorem~\ref{THEOREM:RAC SIC ACHIEVABILITY} is provided in Section~\ref{sec:rac sic proof}, which combines the proof techniques for Gaussian RAC in~\cite{yavas2021gaussianmac} and mismatched P2P interference channel in~\cite{scarlett2017mismatch}.
Theorem~\ref{THEOREM:RAC SIC ACHIEVABILITY} lower bounds the second-order achievable rate of additive non-Gaussian noise RAC with SIC decoding. Since all active users adopt the same codebook with $M$ codewords, the first-order achievable rate in~\eqref{eq:rac sic ach} is the minimum rate to decode the messages of all active users reliably, which corresponds to the decoding of the first message when all the codewords of all other messages are treated as noise. This leads to the signal to interference and noise ratio $\frac{P}{(1+(k-1)P)}$ and thus the first-order rate $ C(\frac{P}{(1+(k-1)P)})$. 

Notice that in Theorems~\ref{THEOREM:RAC JNN ACHIEVABILITY} and~\ref{THEOREM:RAC SIC ACHIEVABILITY}, the expression for $M^*(\bn, \boldsymbol{\varepsilon})$ is given as the minimum over all $k \in [K]$. This is because the lower bound on $\log M$ must be satisfied simultaneously for all blocklengths $(n_1, \ldots, n_K)$ under the corresponding error probability constraints $(\varepsilon_1, \ldots, \varepsilon_K)$. If one is interested in the relation between $\log M$ and a specific number of active users $k \in [K]$ with blocklength $n_k$, the corresponding terms inside the minimum of~\eqref{eq:rac jnn ach} and~\eqref{eq:rac sic ach} capture it.

It follows from~\eqref{eq:rac jnn ach} and~\eqref{eq:rac sic ach} that for mismatched RAC, SIC decoding is strictly inferior to JNN decoding even in the first-order rate. This is different from the mismatched MAC with JNN and SIC, where their first-order asymptotic results align. The reason is that in the RAC, all active users adopt the same codebook with $M$ codewords. This will lead to the fact that the first-order achievable rate of SIC in a RAC is the minimum rate to decode the messages of all active users reliably, corresponding to decoding the first message when all the codewords of all other messages are treated as noise. As a result, the smallest signal-to-interference and noise ratio equals $\frac{P}{(1+(k-1)P)}$, which leads to the first-order rate $ C(\frac{P}{(1+(k-1)P)})$. Compared with the first-order rate of JNN in a RAC, $\frac{C(kP)}{k}> C(\frac{P}{(1+(k-1)P)})$ when $k\geq 2$. The detailed proof is provided in Appendix~\ref{sec:appendix rac first order}.

In contrast, in a MAC, users typically employ distinct codebooks. In this case, SIC can sequentially decode users at their respective maximum rates (ordered by channel gains) and can also be combined with time sharing, avoiding the ``treating interference as noise'' penalty seen in the first-order asymptotic result of a RAC. This explains why SIC and JNN achieve comparable first-order rates in MAC but diverge in RAC under our assumptions. 

Fig.~\ref{fig:rac first comp} illustrates the first-order rates of RAC with JNN and SIC decoding as a function of the number of active users $k$. It can be observed that the asymptotic results for both JNN and SIC decoding decrease monotonically as $k$ increases, and JNN decoding outperforms SIC decoding. Fig.~\ref{fig:rac second comp} presents the second-order rates for RAC with JNN and SIC decoding, highlighting the inferior performance of SIC decoding in the finite blocklength regime as well.

\begin{figure*}[tb]
\centering
\subfigure[Comparison of the first-order rates]
{
\begin{minipage}{.4\linewidth}
\centering
\includegraphics[scale=0.48]{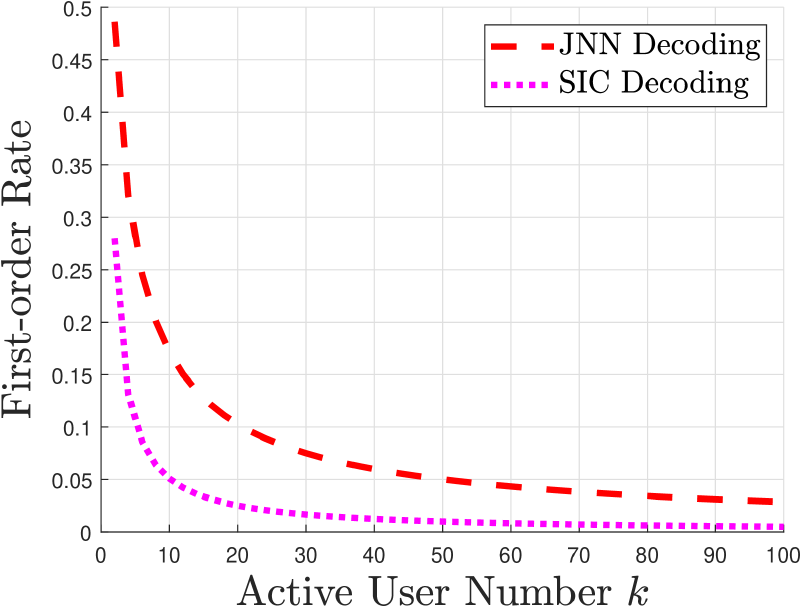}
\label{fig:rac first comp}
\end{minipage}
}
\subfigure[Comparison of the second-order rates]
{
\begin{minipage}{.4\linewidth}
\centering
\includegraphics[scale=0.48]{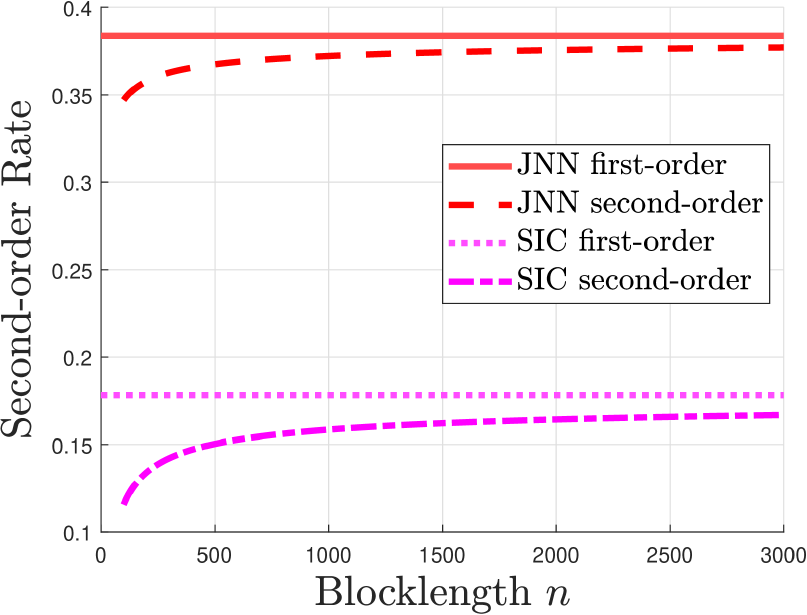}
\label{fig:rac second comp}
\end{minipage}
}
\caption{Comparison of the first- and second-order rates in a RAC for JNN and SIC decoding when $P=3$, $\xi=3$, $k=3$ and $\varepsilon_k=0.1$.}
\end{figure*}

\section{Proof of the Achievable Rate Regions for Two-User MAC (Theorems~\ref{theo: mac jnn achie} and~\ref{theorem:MAC-SIC ach})}
\label{sec: mac proof}

\subsection{Lemmas}
\label{sec:mac proof preliminaries}
In this subsection, we present definitions and key lemmas for the proof of Theorems~\ref{theo: mac jnn achie} and~\ref{theorem:MAC-SIC ach} concerning the low-latency performance of mismatched MAC with JNN and SIC decoding, respectively. 

The following definitions are needed. Given any $(x_1^n,x_2^n,y^n)\in(\bbR^n)^3$, for each $j\in[2]$, define the following mismatched (conditional) information densities
\begin{align}
\tilde{\imath}(x_1^n;y^n)&:=n\rmC\left(\frac{P_1}{1+P_2}\right)+\frac{\|y^n\|^2}{2(1+P_1+P_2)}-\frac{\|y^n-x_1^n\|^2}{2(1+P_2)}\label{eq:mismatch information density},\\
\tilde{\imath}_j^n(x_j^n;y^n|x_{3-j}^n)&:=n\rmC(P_j)+\frac{\| y^n - x_{3-j}^n \| ^2}{2(1+P_j)}-\frac{\|y^n-x_1^n-x_2^n\| ^2}{2},\label{eq:mismatch information density1&2}\\
\tilde{\imath}_{1,2}^n(x_1^n,x_2^n;y^n)&:=n\rmC(P_1+P_2)+\frac{\| y^n\| ^2}{2(1+P_1+P_2)}-\frac{\|y^n-x_1^n-x_2^n\| ^2}{2}\label{eq:mismatch information density1,2}.
\end{align}

With the above definitions, we have the following RCU bound for JNN decoding of mismatched MAC. 
\begin{lemma}
\label{lemma:mismatch RCU JNN}
\textbf{(RCU bound for mismatched MAC with JNN decoder)}
There exists an $(n,M_1,M_2,P_1,P_2)$-JNN code such that the ensemble average probability satisfies
\begin{align}
\nn\rmP_\rme^n \leq \bbE
\Big[ \min
\Big\{ 1,&(M_1-1)\Pr\{ \tilde{\imath}_1^n(\barX_1^n;Y^n|X_2^n)\geq\tilde{\imath}_1^n(X_1^n;Y^n|X_2^n) | X_1^n,X_2^n,Y^n \}\\
\nn+& (M_2-1)\Pr\{ \tilde{\imath}_2^n(\barX_2^n;Y^n|X_1^n) \geq \tilde{\imath}_2^n(X_2^n;Y^n|X_1^n) | X_1^n,X_2^n,Y^n \}\\
+& (M_1-1)(M_2-1)\Pr\{ \tilde{\imath}_{1,2}^n(\barX_1^n,\barX_2^n;Y^n) \geq
\tilde{\imath}_{1,2}^n(X_1^n,X_2^n;Y^n) | X_1^n,X_2^n,Y^n \}\Big\} \Big],\label{eq:mac jnn RCU conclusion}
\end{align}
where $(X_1^n,\barX_1^n,X_2^n,\barX_2^n,Y^n)\sim P_{X_1^n}P_{X_1^n}P_{X_2^n}P_{X_2^n}P_{Y^n|X_1^nX_2^n}$, $P_{X_1^n}$ and $P_{X_2^n}$ are the channel input distribution in~\eqref{eq:sphericaldist} and $P_{Y^n|X_1^nX_2^n}$ is induced by the additive noise channel $Y^n=X_1^n+X_2^n+Z^n$.
\end{lemma}

\begin{proof}
Recall the definition of JNN code in Def.~\ref{def:MAC JNN coding scheme}. A decoding error event occurs if one of the following three events occurs:
\begin{itemize}
\item $\calE_1: \{ \hatW_1 \neq W_1, \hatW_2 = W_2\}$,
\item $\calE_2: \{ \hatW_1 = W_1, \hatW_2 \neq W_2\}$,
\item $\calE_{1,2}: \{ \hatW_1 \neq W_1, \hatW_2 \neq W_2\}$.
\end{itemize}
It follows that
\begin{align}
\rmP_\rme^n &= \Pr\left\{\calE_1 \cup \calE_2 \cup \calE_{1,2}\right\}\\
&= \Pr\left\{\calE_1\right\} + \Pr\left\{\calE_2\right\} + \Pr\left\{\calE_{1,2}\right\}.\label{eq:Pen in three case}
\end{align}

We next bound each term of~\eqref{eq:Pen in three case}. The crux is to use the equivalence between the JNN decoder and the maximum mismatched information density decoder for each error event. Let $(w_1,w_2)\in[M]^2$ be the messages transmitted by the two users. Given that $\calE_1$ happens, i.e. the message $w_2$ is decoded correctly, the JNN decoder in~\eqref{eq:nn decoder} satisfies
\begin{align}
\Phi(Y^n)&=\argmin_{\barw_1\in[M_1]}\|{Y^n-X_1^n(\barw_1)-X_2^n(w_2)}\|^2~\label{eq:nn equivalent first}\\
&=\argmax_{\barw_1\in[M_1]}n\rmC(P_1)+\frac{\| Y^n - X_2^n(w_2) \| ^2}{2(1+P_1)}-\frac{\| Y^n-X_1^n(\barw_1)-X_2^n(w_2) \| ^2}{2}\label{eq:nn equivalent second last}\\
&=\argmax_{\barw_1\in[M_1]}\tilde{\imath}_1^n(X_1^n(\barw_1);Y^n|X_2^n(w_2))\label{eq:nn equivalent last},
\end{align}
where \eqref{eq:nn equivalent second last} follows since the the additive terms are constants when $Y^n$ and $w_2$ are known. The result in~\eqref{eq:nn equivalent last} implies that under error event $\calE_1$, maximizing the mismatched information density $\tilde{\imath}_1^n(X_1^n(\barw_1);Y^n|X_2^n(w_2))$ over $\barw_1\in[M_1]$ is equivalent to JNN decoding. Analogously, maximizing the mismatched information density $\tilde{\imath}_2^n(X_2^n(\barw_2);Y^n|X_1^n(w_1))$ over $\barw_2\in[M_2]$ is equivalent to JNN decoding under error event $\calE_2$.

Given that error event $\calE_{1,2}$ happens, the JNN decoder satisfies
\begin{align}
\Phi(Y^n)&=\argmin_{(\barw_1,\barw_2)\in[M_1]\times[M_2]}\|{Y^n-X_1^n(\barw_1)-X_2^n(\barw_2)}\|^2\\
&=\argmax_{(\barw_1,\barw_2)\in[M_1]\times[M_2]}n\rmC(P_1+P_2)+\frac{\| Y^n \| ^2}{2(1+P_1+P_2)}-\frac{\| Y^n-X_1^n(\barw_1)-X_2^n(\barw_2) \| ^2}{2}\\
&=\argmax_{(\barw_1,\barw_2)\in[M_1]\times[M_2]}\tilde{\imath}_{1,2}^n(X_1^n(\barw_1),X_2^n(\barw_2);Y^n)\label{eq:nn equivalent 12 last}.
\end{align}
The result in~\eqref{eq:nn equivalent 12 last} implies that under error event $\calE_{1,2}$, maximizing the mismatched information density $\tilde{\imath}_{1,2}^n(X_1^n(\barw_1),X_2^n(\barw_2);Y^n)$ over all message pairs $(\barw_1,\barw_2)\in[M_1]\times[M_2]$ is equivalent to JNN decoding.

For ease of notation, given any $(\barw_1,\barw_2)\in[M_1]\times[M_2]$, define the events 
\begin{align}
\calF(\barw_1,\barw_2)
&:=\{\|{Y^n-X_1^n(\barw_1)-X_2^n(\barw_2)}\|^2\leq\|{Y^n-X_1^n(w_1)-X_2^n(w_2)}\|^2\},\\
\calG_1(\barw_1)
&:=\{\tilde{\imath}_1^n(X_1^n(\barw_1);Y^n|X_2^n(w_2))\geq \tilde{\imath}_1^n(X_1^n(w_1);Y^n|X_2^n(w_2))\},\\
\calG_2(\barw_2)
&:=\{\tilde{\imath}_2^n(X_2^n(\barw_2);Y^n|X_1^n(w_1))\geq \tilde{\imath}_2^n(X_2^n(w_2);Y^n|X_1^n(w_1))\},\\
\calG_{1,2}(\barw_1,\barw_2)
&:=\{\tilde{\imath}_{1,2}^n(X_1^n(\barw_1),X_2^n(\barw_2);Y^n)\geq \tilde{\imath}_{1,2}^n(X_1^n(w_1),X_2^n(w_2);Y^n)\}.
\end{align}

It follows from \eqref{eq:Pen in three case},~\eqref{eq:nn equivalent last} and~\eqref{eq:nn equivalent 12 last} that the conditional error probability given that messages $(w_1,w_2)$ are transmitted satisfies
\begin{align}
\nn&\rmP_\rme^n(w_1,w_2)\\*
&=\Pr\Bigg\{\bigcup_{\substack{\barw_1\in[M_1]:\\\barw_1\neq w_1}}\calF(\barw_1,w_2)\Bigg\}+\Pr\Bigg\{\bigcup_{\substack{\barw_2\in[M_2]:\\\barw_2\neq w_2}}\calF(w_1,\barw_2)\Bigg\}+\Pr\Bigg\{\bigcup_{\substack{(\barw_1,\barw_2)\in[M_1]\times[M_2]:\\\barw_1\neq w_1,\barw_2\neq w_2}}\calF(\barw_1,\barw_2)\Bigg\}\label{eq:mac jnn rcu 1}\\
&=\Pr\Bigg\{\bigcup_{\substack{\barw_1\in[M_1]:\\\barw_1\neq w_1}}\calG_1(\barw_1)\Bigg\}+\Pr\Bigg\{\bigcup_{\substack{\barw_2\in[M_2]:\\\barw_2\neq w_2}}\calG_2(\barw_2)\Bigg\}+\Pr\Bigg\{\bigcup_{\substack{(\barw_1,\barw_2)\in[M_1]\times[M_2]:\\\barw_1\neq w_1,\barw_2\neq w_2}}\calG_{1,2}(\barw_1,\barw_2)\Bigg\}
\label{eq:mac jnn rcu 2}\\
\nn&=\bbE\Big[(M_1-1)\Pr\{ \tilde{\imath}_1^n(\barX_1^n;Y^n|X_2^n)\geq\tilde{\imath}_1^n(X_1^n;Y^n|X_2^n)|X_1^n,X_2^n,Y^n\}\\
\nn&\qquad+ (M_2-1)\Pr\{ \tilde{\imath}_2^n(\barX_2^n;Y^n|X_1^n) \geq \tilde{\imath}_2^n(X_2^n;Y^n|X_1^n)|X_1^n,X_2^n,Y^n\}\\
&\qquad+ (M_1-1)(M_2-1)\Pr\{ \tilde{\imath}_{1,2}^n(\barX_1^n,\barX_2^n;Y^n) \geq
\tilde{\imath}_{1,2}^n(X_1^n,X_2^n;Y^n)|X_1^n,X_2^n,Y^n\}\Big],\label{eq:mac jnn rcu 3}
\end{align}
where~\eqref{eq:mac jnn rcu 1} follows from the definition of JNN decoder and three error events, \eqref{eq:mac jnn rcu 2} follows due to the equivalence between JNN decoding and maximum mismatched information density decoding, and \eqref{eq:mac jnn rcu 3} follows from the union bound and the fact given any $(w_1,\barw_1,w_2,\barw_2)\in[M_1]^2\times[M_2]^2$, the joint distribution of random variables $(X_1^n(w_1),X_1^n(\barw_1),X_2^n(w_2),X_2^n(\barw_2),Y^n)$ is the same and given by $P_{X_1^n}P_{X_1^n}P_{X_2^n}P_{X_2^n}P_{Y^n|X_1^nX_2^n}$. 

The proof of Lemma~\ref{lemma:mismatch RCU JNN} by using~\eqref{eq:mac jnn rcu 2} is completed by noting that $\rmP_\rme^n(w_1,w_2)\leq 1$ for any $(w_1,w_2)\in[M_1]\times[M_2]$ and thus
\begin{align}
\rmP_\rme^n
&=\bbE[\rmP_\rme^n(W_1,W_2)]\leq \bbE[\min\{1,\rmP_\rme^n(W_1,W_2)\}].
\end{align}
\end{proof}

The proof of Lemma~\ref{lemma:mismatch RCU JNN} builds on the proof of the RCU bound under matched P2P case~\cite[Theorem 6]{PPV}. In the mismatched case, we show the equivalence of JNN decoding and the maximization of the mismatched information density. Thus, by replacing the information density in~\cite[Theorem 6]{PPV} with mismatched information density, we establish the RCU bound under mismatched MAC with JNN decoding in Lemma~\ref{lemma:mismatch RCU JNN}.
Analogously, we have the following RCU bound for SIC decoding.
\begin{lemma}
\label{lemma:mismatch RCU SIC}
\textbf{(RCU bound for mismatched MAC with SIC decoder)}
There exists an $(n,M_1,M_2,P_1,P_2)$-SIC code such that the ensemble average probability satisfies
\begin{align}
\nn\rmP_\rme^n 
&\leq \bbE
\Big[ \min
\Big\{ 1,(M_1-1)\Pr\{ \tilde{\imath}^n(\barX_1^n;Y^n)\geq\tilde{\imath}^n(X_1^n;Y^n) | X_1^n,Y^n \}\Big\} \Big]\\
&\qquad+\bbE
\Big[ \min
\Big\{ 1,(M_2-1)\Pr\{ \tilde{\imath}_2^n(\barX_2^n;Y^n|X_1^n)\geq\tilde{\imath}_2^n(X_2^n;Y^n|X_1^n) | X_1^n,X_2^n,Y^n \}\Big\} \Big],\label{eq:mac sic RCU conclusion}
\end{align}
where the joint distribution of $(X_1^n,\barX_1^n,X_2^n,\barX_2^n,Y^n)$ satisfies
\begin{align}
P_{X_1^n,\barX_1^n,X_2^n,\barX_2^n,Y^n}(x_1^n,\barx_1^n,x_2^n,\barx_2^n,y^n)=P_{X_1^n}(x_1^n)P_{X_1^n}(\barx_1^n)P_{X_2^n}(x_2^n)P_{X_2^n}(\barx_2^n)P_{Y^n|X_1^nX_2^n}(y^n|x_1^n,x_2^n).
\end{align}
$P_{X_1^n}$ and $P_{X_2^n}$ are the channel input distribution in~\eqref{eq:sphericaldist} and $P_{Y^n|X_1^nX_2^n}$ is the additive noise channel in~\eqref{eq:MAC channel model}.
\end{lemma}

The proofs of Lemmas~\ref{lemma:mismatch RCU JNN} and~\ref{lemma:mismatch RCU SIC} judiciously combine the results in~\cite[Theorem 16]{PPV} and~\cite[Theorem 6]{effros2020mac}. However, in the mismatched scenario, the proof differs from the previous ones by leveraging the equivalence between the JNN decoder and the maximum mismatched information density decoder.

Given any $(x_1^n,x_2^n,y^n)$ and any $t\in\bbR_+$, define the following probability functions
\begin{align}
g_1(t;y^n,x_2^n) &:= \Pr\left\{\tilde{\imath}_1^n(\barX_1^n;y^n|x_2^n)\geq t\right\},\label{eq:define g1}\\
g_2(t;y^n,x_1^n) &:= \Pr\left\{\tilde{\imath}_2^n(\barX_2^n;y^n|x_1^n)\geq t\right\},\label{eq:define g2}\\
g_{1,2}(t;y^n) &:= \Pr\left\{\tilde{\imath}_{1,2}^n(\barX_1^n,\barX_2^n;y^n)\geq t\right\},\\
g(t;y^n) &:= \Pr\{\tilde{\imath}^n(\barX_1^n;y^n)\geq t\},\label{eq:define g}
\end{align}
where $\barX_j^n$ is distributed according to $P_{X_j^n}$ in~\eqref{eq:sphericaldist}. Each probability term concerns a mismatched (conditional) information density.

For subsequent analyses, define the following constants that depend on $(P_1,P_2)$:
\begin{align}
\beta_i&:=\frac{27\sqrt{\pi}(1+P_i)}{\sqrt{8(1+2P_i)}},~i\in[2],\\
\beta_{1,2}&:=\frac{9(P_1+P_2)}{2\pi\sqrt{2P_1P_2}}.
\end{align}

The following lemma upper bounds the above probabilities $g(\cdot)$, generalizing~\cite[Lemma 6]{yavas2021gaussianmac} to the mismatched MAC with additive arbitrary noise. It is important to note that Lemma \ref{lemma:g function} differs from~\cite[Lemma 6]{yavas2021gaussianmac} in that the $1/\sqrt{n}$ term is absent. Note that \cite[Lemma 6]{yavas2021gaussianmac} was established under Gaussian noise (~\cite[Section IV.E]{tantomamichel2015}) and the $1/\sqrt{n}$ term is essential for deriving the more refined third-order asymptotic term. However, for second-order asymptotics, the results in the following lemma suffice.
\begin{lemma}
\label{lemma:g function}
Given any positive real number $t\in\bbR_+$,
\begin{align}
g_1(t;y^n,x_2^n) &\leq \beta_1\exp(-t),\label{eq:g function 1}\\
g_2(t;y^n,x_1^n) &\leq \beta_2\exp(-t),\label{eq:g function 2}\\
g_{1,2}(t;y^n) &\leq \beta_{1,2}\exp(-t),\label{eq:g function 12}\\
g(t;y^n)&\leq \beta_{1,2}\exp(-t)\label{eq:g function}.
\end{align}
\end{lemma}
\begin{proof}
We first prove~\eqref{eq:g function 1}. Fix any $(i,j)\in[2]\times[n]$. Let $\tilX_{i,j}\sim\calN(0,P_i)$, $\tilX_i^n=(\tilX_{i,1},\ldots,\tilX_{i,n})$, $\tilY_j\sim\calN(0,1+P_1+P_2)$, and $\tilY^n=(\tilY_1,\ldots,\tilY_n)$. Let $P_{\tilX_i^n}$ denote the distribution of $\tilX_i^n$ and let $\tilY^n=\tilX_1^n+\tilX_2^n+\tilZ^n$, where $\tilZ^n$ is generated i.i.d. from $\calN(0,1)$. Let $P_{\tilX_1^n\tilX_2^nY^n}$ be the joint distribution of $(\tilX_1^n,\tilX_2^n,\tilY^n)$ and let any other (conditional) distributions concerning $(\tilX_1^n,\tilX_2^n,Y^n)$ be induced by this joint distribution. In particular, $P_{\tilY^n|\tilX_1^n\tilX_2^n}$ denotes the channel law for the Gaussian MAC. It follows from~\eqref{eq:mismatch information density1&2} that for any $(x_1^n,x_2^n,y^n)\in\bbR^{3n}$,
\begin{align}
\tilde{\imath}_1^n(x_1^n;y^n|x_2^n)=\log\frac{P_{\tilY^n|\tilX_1^n,\tilX_2^n}(y^n|x_1^n,x_2^n)}{P_{\tilY^n|\tilX_2^n}(y^n|x_2^n)}.
\end{align}
Let $Y_\rmG^n=X_1^n+X_2^n+\tilZ^n$ be the channel output of a Gaussian MAC when the channel inputs are spherical codewords $(X_1^n,X_2^n)$ that are generated by $P_{X_1^n}$ and $P_{X_2^n}$ in~\eqref{eq:sphericaldist}, respectively. Let the joint distribution of $(X_1^n,X_2^n,Y_\rmG^n)$ be $P_{X_1^nX_2^nY_\rmG^2}$ and let other (conditional) distributions concerning $(X_1^n,X_2^n,Y_\rmG^n)$ be induced by this joint distribution. For any $(x_1^n,x_2^n,y^n)\in\bbR^{3n}$, define the following information density:
\begin{align}
\imath_\rmG(x_1^n;y^n|x_2^n):=\log\frac{P_{Y_\rmG^n|X_1^n,X_2^n}(y^n|x_1^n,x_2^n)}{P_{Y_\rmG^n|X_2^n}(y^n|x_2^n)}.\label{eq: def iG}
\end{align}
It follows from the definition in~\eqref{eq:define g1} that $g_1$ can be upper bounded as follows:
\begin{align}
g_1(t;y^n,x_2^n)&=\Pr_{P_{X_1^n}}\Bigg\{\log\frac{P_{\tilY^n|\tilX_1^n,\tilX_2^n}(y^n|\barX_1^n,x_2^n)}{P_{\tilY^n|\tilX_2^n}(y^n|x_2^n)}\geq t\Bigg\}\\
&=\Pr_{P_{X_1^n}}\Bigg\{\log\frac{P_{\tilY^n|\tilX_1^n,\tilX_2^n}(y^n|\barX_1^n,x_2^n)}{P_{Y_\rmG^n|X_2^n}(y^n|x_2^n)}\geq t-\log\frac{P_{Y_\rmG^n|X_2^n}(y^n|x_2^n)}{P_{\tilY^n|\tilX_2^n}(y^n|x_2^n)}\Bigg\}\\
&\leq\Pr_{P_{X_1^n}}\Bigg\{\log\frac{P_{\tilY^n|\tilX_1^n,\tilX_2^n}(y^n|\barX_1^n,x_2^n)}{P_{Y_\rmG^n|X_2^n}(y^n|x_2^n)}\geq t-\log\beta_1\Bigg\}\label{eq: g1 proof 4}\\
&=\Pr_{P_{X_1^n}}\Bigg\{\log\frac{P_{Y_\rmG^n|X_1^n,X_2^n}(y^n|\barX_1^n,x_2^n)}{P_{Y_\rmG^n|X_2^n}(y^n|x_2^n)}\geq t-\log\beta_1\Bigg\}\label{eq: g1 proof 5}\\
&=\Pr_{P_{X_1^n}}\{\imath_\rmG(\barX_1^n;y^n|x_2^n)\geq t-\log\beta_1\}\label{eq: g1 proof 6}\\
&=\int_{-\infty}^\infty P_{X_1^n}(\barx_1^n)\bbo\{\imath_\rmG(\barx_1^n;y^n|x_2^n)\geq t-\log\beta_1\}\rmd \barx_1^n,\label{eq: g1 proof 7}
\end{align}
where~\eqref{eq: g1 proof 4} follows from \cite[Lemma 1]{yavas2021gaussianmac}, which implies that $\frac{P_{Y_\rmG^n|X_2^n}(y^n|x_2^n)}{P_{\tilY^n|\tilX_2^n}(y^n|x_2^n)}\leq \beta_1$,~\eqref{eq: g1 proof 5} follows since both $P_{Y_\rmG^n|X_1^n,X_2^n}$ and $P_{\tilY^n|\tilX_1^n,\tilX_2^n}$ denote the distribution of the Gaussian noise $\tilZ^n$, and~\eqref{eq: g1 proof 6} follows from the definition in~\eqref{eq: def iG}.

Bounding~\eqref{eq: g1 proof 7} follows the same steps as in~\cite[Eq. (3.310)-(3.320)]{polyanskiy2010thesis} (see also~\cite[Eq. (27)-(30)]{tantomamichel2015}) by performing a change-of-measure. Recall that $\barX_1^n,X^n_2$ follow the spherical distribution in~\eqref{eq:sphericaldist} and $Y_\rmG^n$ is the Gaussian MAC with spherically distributed inputs. It follows that for any $(\barx_1^n,x_2^n,y^n)\in\bbR^{3n}$,
\begin{align}
P_{X_1^n}(\barx_1^n)&=P_{X_1^n|X_2^n}(\barx_1^n|x_2^n)\label{eq: change of measure 1}\\
&=P_{X_1^n|Y^n_\rmG,X_2^n}(\barx_1^n|y^n,x_2^n)\frac{P_{X_1^n|X_2^n}(\barx_1^n|x_2^n)}{P_{X_1^n|Y^n_\rmG,X_2^n}(\barx_1^n|y^n,x_2^n)}\label{eq: change of measure 2}\\
&=P_{X_1^n|Y^n_\rmG,X_2^n}(\barx_1^n|y^n,x_2^n)\exp(-\imath_\rmG(\barx_1^n;y^n| x_2^n))\label{eq: change of measure},
\end{align}
where~\eqref{eq: change of measure 1} follows since $P_{X_1^n}$ and $P_{X_2^n}$ are independent of each other, and~\eqref{eq: change of measure} follows from the definition of information density in~\eqref{eq: def iG}. Given any $\eta\in\bbR_+$, combining~\eqref{eq: g1 proof 7} and~\eqref{eq: change of measure}, we obtain
\begin{align}
g_1(t;y^n,x_2^n)
\nn&=\int_{-\infty}^\infty P_{X_1^n|Y^n_\rmG,X_2^n}(\barx_1^n|y^n,x_2^n)\exp(-\imath_\rmG(\barx_1^n;y^n| x_2^n))\\
&\qquad\qquad\times\bbo\{\imath_\rmG(\barx_1^n;y^n|x_2^n)\geq t-\log\beta_1\}\rmd \barx_1^n\\
&=\bbE_{P_{X_1^n|Y^n_\rmG,X_2^n}}[\exp(-\imath_\rmG(\barX_1^n;y^n|x_2^n))\bbo\{\imath_\rmG(\barX_1^n;y^n|x_2^n)\geq t-\log\beta_1\}]\\
&\nn\leq\sum_{l=0}^\infty \exp(-(t-\log\beta_1+l\eta))\\*
&\qquad\quad\times\Pr\big\{t-\log\beta_1+l\eta\leq\imath_\rmG(\barX_1^n;y^n|x_2^n)\leq t-\log\beta_1+(l+1)\eta\big\}\\
&\leq\sum_{l=0}^\infty \exp(-(t-\log\beta_1+l\eta))\\
&=\frac{\exp(-t+\log\beta_1)}{1-\exp(-\eta)}\label{eq: sum of geometric sequence}\\
&\leq\beta_1\exp(-t),
\end{align}
where~\eqref{eq: sum of geometric sequence} follows from the summation of a geometric sequence. The proofs for \eqref{eq:g function 2} to \eqref{eq:g function} are similar and are thus omitted.
\end{proof}

In our proof, the function version Berry-Esseen theorem by MolavianJazi and Laneman~\cite[Prop. 1]{molavianjazi2015second} is used. For completeness, we recall the theorem in the following lemma. Fix any integer $l\in\bbN$. For each $i\in[n]$, let $\bX_i=(X_{i,1},\ldots,X_{i,d})$ be a $d$-dimensional random vector generated from any given pdf. Furthermore, let $\bt=(t_1,\ldots,t_d)\in\bbR^d$ be any vector and let $\boldf(\bt)=(f_1(\bt),\ldots,f_l(\bt)):\bbR^d\rightarrow\bbR^l$ be an $l$-component vector-function with continuous second-order partial derivatives in a neighborhood of $\bt=\bbE[\bX_1]$, whose Jacobian matrix $\bJ=\{J_{i,j}\}_{i\in[l],j\in[d]}$ satisfies
\begin{align}
J_{i,j}:=\frac{\partial f_i(\bt)}{\partial t_j}\bigg|_{\bt=\bbE[\bX_1]}.
\end{align}
Assume that the third absolute moment of the random vector is finite, i.e., $\bbE[\|\bX_1-\bbE[\bX_1]\|^3]<\infty$, and assume that the covariance matrix $\cov[\bJ \bX_1]$ is positive definite.

\begin{lemma}[Function Version Berry-Esseen Theorem]
\label{lemma:berry esseen for func}
There exists a finite positive constant $B\in\bbR_+$ such that, for any convex Borel-measurable set $\calD\in\bbR^l$ and for all $n\in\bbN_+$,
\begin{align}
\Bigg|\Pr\Bigg\{\sqrt{n}\Bigg(\boldf\Bigg(\frac{1}{n}\sum_{t=1}^n\bX_t\Bigg)-\boldf(\bbE[\bX_1])\Bigg)\in\calD\Bigg\}-\Pr\{Z^l\in\calD\}\Bigg|\leq\frac{B}{\sqrt[4]{n}},
\end{align}
where $Z^l\sim\calN(\bzero_l,\bJ\cov[\bX_1]\bJ^\rmT)$ is an $l$-dimensional Gaussian random vector.
\end{lemma}

\subsection{Proof of Theorem~\ref{theo: mac jnn achie} (MAC with JNN decoding)}
\label{sec:mac jnn proof}
To prove Theorem \ref{theo: mac jnn achie}, we judiciously combine the proof techniques in~\cite{molavianjazi2015second,yavas2021gaussianmac}. However, a critical step to get the third-order asymptotic result is~\cite[Lemma 6]{yavas2021gaussianmac}, which cannot be applied directly to the mismatched case, as it was established for the matched case with AWGN. To address this problem, we derive a weaker bound in Lemma~\ref{lemma:g function}, which is applicable for arbitrary noise distributions encountered in the mismatched MAC.

\subsubsection{Step 1: upper bound the ensemble error probability by RCU}
Using the channel law $Y^n=X_1^n+X^2_n+Z^n$, it follows from \eqref{eq:mismatch information density} that
\begin{align}
\tilde{\imath}_1:=\tilde{\imath}^n_1(X_1^n;Y^n|X_2^n) - n\rmC(P_1) &= \frac{\parallel X_1^n + Z^n \parallel^2}{2(1+P_1)} - \frac{\parallel Z^n \parallel^2}{2}\\
&= \frac{nP_1+\parallel Z^n \parallel^2+ 2\langle X_1^n,Z^n \rangle}{2(1+P_1)} - \frac{\parallel Z^n \parallel^2}{2}\\
&= \frac{(n-\parallel Z^n \parallel^2)P_1+ 2\langle X_1^n,Z^n \rangle}{2(1+P_1)}.\label{eq:i matrix 1}
\end{align}
Similarly, it follows from~\eqref{eq:mismatch information density1&2} and~\eqref{eq:mismatch information density1,2} that
\begin{align}
\tilde{\imath}_2&:=\tilde{\imath}^n_2(X_2^n;Y^n|X_1^n) - n\rmC(P_2)= \frac{(n-\parallel Z^n \parallel^2)P_2+ 2\langle X_2^n,Z^n \rangle}{2(1+P_2)},\label{eq:i matrix 2} \\
\tilde{\imath}_{1,2}&:=\tilde{\imath}^n_{1,2}(X_1^n,X_2^n;Y^n) - n\rmC(P_1+P_2)=\frac{(n-\parallel Z^n \parallel^2)(P_1+P_2)+ 2\langle X_1^n+X_2^n,Z^n \rangle+ 2\langle X_1^n,X_2^n \rangle}{2(1+P_1+P_2)}\label{eq:i matrix 12}.
\end{align}

Recall that $\beta_1,\beta_2,\beta_{1,2}$ are constants in Lemma \ref{eq:g function}. Define the event
\begin{align}
\calS &:=\left\{\begin{bmatrix}
\tilde{\imath}_1^n(X_1^n;Y^n|X_2^n) \\
\tilde{\imath}_2^n(X_2^n;Y^n|X_1^n) \\
\tilde{\imath}_{1,2}^n(X_1^n,X_2^n;Y^n)
\end{bmatrix}\geq\log\begin{bmatrix}
4M_1\beta_1\sqrt{n} \\
4M_2\beta_2\sqrt{n} \\
2M_1M_2\beta_{1,2}\sqrt{n}
\end{bmatrix}\right\}.
\label{def:calS:mac}
\end{align}
For ease of notation, given random variables $(X_1^n,X_2^n,Y^n)\sim P_{X_1^n}P_{X_2^n}P_{Y^n|X_1^nX_2^n}$ in Lemma \ref{lemma:mismatch RCU JNN}, define the following random variables
\begin{align}
G_1:=g_1(\tilde{\imath}_1(X_1^n;Y^n|X_2^n);Y^n,X_2^n),\\
G_2:=g_2(\tilde{\imath}_2(X_2^n;Y^n|X_1^n);Y^n,X_1^n),\\
G_{1,2}:=g_{1,2}(\tilde{\imath}^n_{1,2}(X_1^n,X_2^n;Y^n);Y^n).
\end{align}

Similar to~\cite[Eq. (86)-(90)]{yavas2021gaussianmac}, it follows from Lemma~\ref{lemma:mismatch RCU JNN} that the ensemble average error probability satisfies
\begin{align}
\rmP_\rme^n 
&\leq \bbE[ \min\{ 1,(M_1-1)G_1+(M_2-1)G_2+(M_1-1)(M_2-1)G_{1,2} \}]\label{eq:boundingpen1}\\
\nn&\leq \bbE[ \min\{ 1,(M_1-1)G_1+(M_2-1)G_2+(M_1-1)(M_2-1)G_{1,2} \}\bbo(\calS^\rmc)]\\*
&\qquad+\bbE[ \min\{ 1,(M_1-1)G_1+(M_2-1)G_2+(M_1-1)(M_2-1)G_{1,2} \}\bbo(\calS) ]\\
&\leq\Pr\{\calS^\rmc\}+\bbE[ \{ (M_1-1)G_1+(M_2-1)G_2+(M_1-1)(M_2-1)G_{1,2} \}\bbo(\calS) ]\label{eq:boundingpen11}\\
\nn&\leq\Pr\{\calS^\rmc\}+M_1\bbE[G_1 \bbo(\tilde{\imath}_1(X_1^n;Y^n|X_2^n)\geq\log(4M_1\beta_1\sqrt{n}))]\\*
&\nn\qquad+M_2\bbE[G_2\bbo(\tilde{\imath}_{2}(X_2^n;Y^n|X_1^n)\geq\log(4M_2\beta_2\sqrt{n}))]\\*
&\qquad+M_1M_2\bbE[G_{1,2}\bbo(\tilde{\imath}_{1,2}(X_1^n,X_2^n;Y^n)\geq\log(2M_1M_2\beta_{1,2}\sqrt{n}))]\label{eq:boundingpen2},
\end{align}
where~\eqref{eq:boundingpen1} follows since the newly defined random variables $G_1$ $G_2$ and $G_{1,2}$ correspond to three terms in right-hand side of the RCU bound in~\eqref{eq:mac jnn RCU conclusion}, \eqref{eq:boundingpen2} follows from the union bound and the fact that $\calS$ is the intersection of three events.

Using Lemma \ref{lemma:g function}, it follows that
\begin{align}
\nn&M_1\bbE[G_1 \bbo(\tilde{\imath}_1(X_1^n;Y^n|X_2^n)\geq\log(4M_1\beta_1\sqrt{n}))]\\*
&\leq \beta_1 M_1 \bbE\big[\exp(-\tilde{\imath}_1(X_1^n;Y^n|X_2^n))\bbo\big(\tilde{\imath}_1(X_1^n;Y^n|X_2^n)\geq\log(4M_1\beta_1\sqrt{n})\big)\big]\\
&\leq \beta_1 M_1\exp(-\log(4M_1\beta_1\sqrt{n}))\\
&=\frac{1}{4\sqrt{n}}\label{mac:remain:1}.
\end{align}
Similarly,
\begin{align}
M_2\bbE[G_2\bbo(\tilde{\imath}_{2}(X_2^n;Y^n|X_1^n)\geq\log(4M_2\beta_2\sqrt{n}))]
&\leq \frac{1}{4\sqrt{n}}\label{mac:remain:2},\\
M_1M_2\bbE[G_{1,2}\bbo(\tilde{\imath}_{1,2}(X_1^n,X_2^n;Y^n)\geq\log(2M_1M_2\beta_{1,2}\sqrt{n}))]
&\leq \frac{1}{2\sqrt{n}}\label{mac:remain:3}.
\end{align}

It remains to bound the probability $\Pr\{\calS^\rmc\}$. It follows from the definition of $\calS$ in \eqref{def:calS:mac} that
\begin{align}
\Pr\{\calS^\rmc\}
&\leq 1-\Pr\left\{\begin{bmatrix}\tilde{\imath}_1\\\tilde{\imath}_2\\\tilde{\imath}_{1,2}\end{bmatrix}\geq \log\begin{bmatrix}
4M_1\beta_1\sqrt{n} \\
4M_2\beta_2\sqrt{n} \\
2M_1M_2\beta_{1,2}\sqrt{n}
\end{bmatrix}-n\bC(P_1,P_2)\right\}\label{eq:Pen upper bound calAc}.
\end{align}

\subsubsection{Step 2: deriving the non-asymptotic form of the second-order result}
Analogously to~\cite[Eq. (57)-(73)]{molavianjazi2015second}, we bound~\eqref{eq:Pen upper bound calAc} using the function version Berry-Esseen theorem (see Lemma \ref{lemma:berry esseen for func}). For ease of notation, let
\begin{align}
\bi&:=[\tilde{\imath}_1;\tilde{\imath}_2;\tilde{\imath}_3],\\
\bm{\tau}&:=\log\begin{bmatrix}
4M_1\beta_1\sqrt{n} \\
4M_2\beta_2\sqrt{n} \\
2M_1M_2\beta_{1,2}\sqrt{n}
\end{bmatrix}-n\bC(P_1,P_2).\label{eq:tau}
\end{align}
For each $j\in[2]$, let $\tilX^n_j \sim \calN(\cdot;\bzero_n,\bone_n)$. It follows from \cite[Eq. (57)]{molavianjazi2015second} that spherical codewords $X_1^n$ and $X_2^n$ can be expressed as functions of i.i.d. random variables. Specifically, for each $(i,j)\in[n]\times[2]$,
\begin{align}
X_{ji}= \sqrt{nP_j}\frac{\tilX_{ji}}{\| \tilX^n_j \|}.\label{eq:spherical to iid}
\end{align}
For each $i \in [n]$, define the following random variables:
\begin{align}
&A_{1i}:= 1- Z_i^2, A_{2i}:=\sqrt{P_1}\tilX_{1i}Z_i, A_{3i}:=\tilX_{1i}^2 - 1,\label{eq:defineA}\\
&A_{4i}:= \sqrt{P_2}\tilX_{2i}Z_i, A_{5i}:= \tilX_{2i}^2 - 1, A_{6i}:=\sqrt{P_1P_2}\tilX_{1i}\tilX_{2i},\nn
\end{align}
and let $\bA_i := (A_{1i},\dots,A_{6i})$. Note that $\bA_{[n]}:=(\bA_1,\ldots,\bA_n)$ are zero-mean i.i.d. random vectors with finite third moments. Furthermore, for each $i\in[n]$, the covariance matrix of $\bA_i$ is
\begin{align}
\label{eq:mac jnn cov}
\cov(\bA_i) = \diag[\xi -1,P_1,2,P_2,2,P_1P_2].
\end{align}
Given any real numbers $\bk:=(k_1,\dots,k_6)\in\bbR^6$, define three functions:
\begin{align}
f_1(\bk) & := \frac{1}{2(1+P_1)}\left(P_1k_1+\frac{2k_2}{\sqrt{1+k_3}}\right),\\
f_2(\bk) & := \frac{1}{2(1+P_2)}\left(P_2k_1+\frac{2k_4}{\sqrt{1+k_5}}\right),\\
f_3(\bk) & := \frac{1}{2(1+P_1+P_2)}\Big((P_1+P_2)k_1 + \frac{2k_2}{\sqrt{1+k_3}} + \frac{2k_4}{\sqrt{1+k_5}} + \frac{2k_6}{\sqrt{(1+k_3)(1+k_5)}}\Big),\label{eq:define f_1}
\end{align}
and let $\boldf(\bk) := [f_1(\bk),f_2(\bk),f_3(\bk)]^{\rm T}$ denote the vector of functions. It follows from above definitions that $\boldf\big(\frac{1}{n}\sum_{i=1}^{n}\bA_i\big) = \frac{1}{n}\bi$.

The Jacobian matrix of the function vector $\bf(\cdot)$ around $\bbE[\bA_1] = \bzero_6$ satisfies
\begin{align}
\label{eq:mac jnn Jmatrix}
\bJ(\bA_1) = \begin{bmatrix}
\frac{P_1}{2(1+P_1)} & \frac{1}{1+P_1} & 0 & 0 & 0 & 0 \\
\frac{P_2}{2(1+P_2)} & 0 & 0 & \frac{1}{1+P_2} & 0 & 0 \\
\frac{P_1+P_2}{2(1+P_1+P_2)} & \frac{1}{1+P_1+P_2} & 0 & \frac{1}{1+P_1+P_2} & 0 & \frac{1}{1+P_1+P_2}
\end{bmatrix}.
\end{align}
Recall the definition of dispersion matrix $\bV(P_1,P_2)$ in~\eqref{eq:dispersion matrix}. It follows that for each $i\in[n]$,
\begin{align}
\nn&\bJ(\bA_i)\cov(\bA_i)\bJ(\bA_i)^\rmT\\*
&=\begin{bmatrix}
\frac{(\xi -1)P_1^2+4P_1}{4(1+P_1)^2} & \frac{(\xi -1)P_1P_2}{4(1+P_1)(1+P_2)} & \frac{(\xi -1)P_1(P_1+P_2)+4P_1}{4(1+P_1)(1+P_1+P_2)} \\
\frac{(\xi -1)P_1P_2}{4(1+P_1)(1+P_2)} & \frac{(\xi -1)P_2^2+4P_2}{4(1+P_2)^2} & \frac{(\xi -1)P_2(P_1+P_2)+4P_2}{4(1+P_2)(1+P_1+P_2)} \\
\frac{(\xi -1)P_1(P_1+P_2)+4P_1}{4(1+P_1)(1+P_1+P_2)}  & \frac{(\xi -1)P_2(P_1+P_2)+4P_2}{4(1+P_2)(1+P_1+P_2)} & \frac{(\xi -1)(P_1+P_2)^2+4(P_1+P_2+P_1P_2)}{4(1+P_1+P_2)^2}\\
\end{bmatrix}\\
&=\bV(P_1,P_2).
\end{align}
Let $\bS_3:=[S_1,S_2,S_3]^\rmT \sim \calN(\cdot;\bzero_3,\bV(P_1,P_2))$ denote the three-dimensional Gaussian random vector with mean zero and covariance matrix $\bV(P_1,P_2)$. It follows from \eqref{eq:Pen upper bound calAc} and Lemma~\ref{lemma:berry esseen for func} that
\begin{align}
\Pr\{\calS^\rmc\}&
\leq 1-\Pr\left\{\boldf\left(\frac{1}{n}\sum_{i=1}^{n}\bA_i\right)\geq \frac{\bm{\tau}}{n}\right\}\\
&\leq 1-\Pr\left\{ \bS_3\geq \bm{\tau}/\sqrt{n}\right\} +O(n^{-1/4}).\label{eq:apply berry mac jnn}
\end{align}
Combining~\eqref{eq:boundingpen2} and~\eqref{mac:remain:1} to~\eqref{mac:remain:3}, the error probability of the proposed $(\bn,M,P)$-code satisfies
\begin{align}
\rmP_\rme^n
&\leq \Pr\{\calS^\rmc\}+n^{-1/2}\label{mac:jnn:start-1}\\
&=1-\Pr\left\{ \bS_3\geq \bm{\tau}/\sqrt{n}\right\} +O(n^{-1/4})\label{mac:jnn:start0}.
\end{align}
Recall the definition of $\rmQ_{\rm inv}$ in~\eqref{eq:define Qinv}. To make~\eqref{mac:jnn:start0} less than $\varepsilon$, one needs
\begin{align}
\label{mac:jnn:start1}
\bm{\tau}/\sqrt{n}\in \rmQ_{\rm inv}(\bV(P_1,P_2),\varepsilon+O(n^{-1/4})).
\end{align}
Thus, using the definition of $\tau$ in~\eqref{eq:tau} and the Taylor series expansion
of $\rmQ_{\rm inv}$ in~\cite[Lemma 13]{kostina2019macracsourcecoding}, we obtain~\eqref{eq:mac jnn theo 2} in Theorem~\ref{theo: mac jnn achie}.

\subsubsection{Step 3: deriving the asymptotic form of the second-order result}
To compare the second-order achievable rate regions of JNN and SIC decoding, we then analyze various cases of boundary rate pairs to get Corollary~\ref{CORO:MAC JNN ACHIEVABILITY}. Fix a rate pair $(R_1^*,R_2^*)$ on the boundary of $\calR_{\rm JNN}$ (see~\eqref{eq:MAC first order region}). To bound the second-order rate region $\calL_{\rm JNN}$ (see Def.~\ref{def:mac second order region}), given real numbers $(L_1,L_2)\in\bbR^2$, let
\begin{align}
\log M_1 = nR_1^*-\sqrt{n}L_1,\\
\log M_2 = nR_2^*-\sqrt{n}L_2.
\end{align}
It follows from \eqref{eq:tau} and~\eqref{eq:apply berry mac jnn} that
\begin{align}
\nn&\limsup_{n\to\infty}\rmP_\rme^n\\*
&\leq 1-\liminf_{n\to\infty}\Pr\{ \bS_3\geq \bm{\tau}/\sqrt{n}\}\\
&=1-\liminf_{n\to\infty}\Pr\left\{\begin{bmatrix}
S_1\\
S_2\\
S_3
\end{bmatrix}\geq\begin{bmatrix}
\sqrt{n}(R_1^*-C(P_1))-L_1\\
\sqrt{n}(R_2^*-C(P_2))-L_2\\
\sqrt{n}(R_1^*+R_2^*-C(P_1+P_2))-(L_1+L_2)
\end{bmatrix}\right\}\label{eq:three dimension inequal}.
\end{align}
For subsequent analyses, define the following three events
\begin{align}
\calA_1&:=\{S_1\geq\sqrt{n}(R_1^*-C(P_1))-L_1\}\label{eq:define cals1},\\
\calA_2&:=\{S_2\geq\sqrt{n}(R_2^*-C(P_2))-L_2\}\label{eq:define cals2},\\
\calA_3&:=\{S_3\geq\sqrt{n}(R_1^*+R_2^*-C(P_1+P_2))-(L_1+L_2)\}.\label{eq:define cals3}
\end{align}
With the above definition, \eqref{eq:three dimension inequal} is equivalent to
\begin{align}
\limsup_{n\to\infty}\rmP_\rme^n
&\leq 1-\liminf_{n\to\infty}\Pr\{\calA_1\cap\calA_2\cap\calA_3\}.\label{eq:three s inequal}
\end{align}

Fix $\alpha\in(0,1)$. First consider case (i) where $(R_1^*,R_2^*)=(\alpha C(P_1/(1+P_2)),C(P_2))$. It follows that 
\begin{align}
R_1^*&<C(P_1),\\
R_1^*+R_2^*&<C(P_1/(1+P_2))+C(P_2)=C(P_1+P_2),
\end{align}
and thus
\begin{align}
\lim_{n\to\infty}\sqrt{n}(R_1^*-C(P_1))-L_1\rightarrow&-\infty,\label{eq:mac s1 infty}\\
\lim_{n\to\infty}\sqrt{n}(R_1^*+R_2^*-C(P_1+P_2))-(L_1+L_2)\rightarrow&-\infty.
\end{align}
As a result,
\begin{align}
\lim_{n\to\infty}\Pr\{\calA_1\}=\lim_{n\to\infty}\Pr\{\calA_3\}=1.
\end{align}
This way, the right-hand side of \eqref{eq:three s inequal} is further upper bounded as
\begin{align}
\nn&1-\liminf_{n\to\infty}\Pr\{\calA_1\cap\calA_2\cap\calA_3\}\\*
&\leq 1-\liminf_{n\to\infty}\big(\Pr\{\calA_2\}-\Pr\{\calA_1^\rmc\}-\Pr\{\calA_3^\rmc\}\big)\label{eq:pr is 1}\\
&=1-\Pr\{S_2\geq-L_2\}\label{eq:take n limt}\\
&=\Pr\{S_2\geq L_2\},\label{eq: mac jnn case i final}
\end{align}
where~\eqref{eq:pr is 1} follows since for 
\begin{align}
\Pr\{\calA_1\cap\calA_2\cap\calA_3\}
&=\Pr\big\{\calA_2\cap(\calA_1\cap\calA_3)\big\}\\
&\geq \Pr\{\calA_2\}-\Pr\{\calA_2\cap(\calA_1\cap\calA_3)^\rmc\}
\\
&\geq\Pr\{\calA_2\}-\Pr\{(\calA_1\cap\calA_3)^\rmc\}\\
&=\Pr\{\calA_2\}-\Pr\{\calA_1^\rmc\cup\calA_3^\rmc\}\\
&\geq \Pr\{\calA_2\}-\Pr\{\calA^\rmc\}-\Pr\{\calA_3^\rmc\},
\end{align}
\eqref{eq:take n limt} follows since $R_2^*=C(P_2)$ for case (i), and~\eqref{eq: mac jnn case i final} follows since $S_2\sim\calN(0,V(P_2))$ is a zero-mean Gaussian random variable. To make~\eqref{eq: mac jnn case i final} less than $\varepsilon$, one requires
\begin{align}
L_2\geq\sqrt{V(P_2)}\rmQ^{-1}(\varepsilon).
\end{align}
It follows that for case (i), the second-order rate region of mismatched MAC with JNN decoding satisfies
\begin{align}
\calL_{\rm JNN}(R_1^*,R_2^*,\varepsilon)\supseteq \big\{(L_1,L_2):L_2\geq\sqrt{V(P_2)}\rmQ^{-1}(\varepsilon)\big\}.
\end{align}

Recall the definition of the dispersion function $\bV_1(P_1,P_2)$ in~\eqref{eq:dispersion matrix 1}. Note that $[
S_2,S_3]^\rmT \sim \calN(\bzero_2,\bV_1(P_1,P_2))$ is a two-dimensional Gaussian random vector, which contains the last two random variables in $\bS=[S_1;S_2;S_3]$. Now consider case (ii), where $(R_1^*,R_2^*)=(C(P_1/(1+P_2)),C(P_2))$. Since $R_1^*<C(P_1)$, \eqref{eq:mac s1 infty} still holds, which leads to $\lim_{n\to\infty}\Pr\{\calA_1\}=1$.  Thus, the right-hand side of \eqref{eq:three s inequal} satisfies
\begin{align}
\nn&\liminf_{n\to\infty}(1-\Pr\{\calA_1\cap\calA_2\cap\calA_3\})\\*
&\leq \liminf_{n\to\infty}(1-\Pr\{\calA_2\cap\calA_3\}-\Pr\{\calA_1^\rmc\})\label{eq:caseii first eq}\\
&=1-\Pr\{\{S_2\geq-L_2\}\cap\{S_3\geq-(L_1+L_2)\}\}\label{eq:caseii second eq}\\
&=1-\Pr\{\{S_2\leq L_2\}\cap\{S_3\leq L_1+L_2\}\}\label{eq:caseii third eq}\\
&=1-\Pr\Bigg\{\begin{bmatrix}
S_2\\
S_3
\end{bmatrix}\leq\begin{bmatrix}
L_2\\
L_1+L_2
\end{bmatrix}\Bigg\},\label{eq: mac jnn case ii final}
\end{align}
where~\eqref{eq:caseii first eq} follows similarly as~\eqref{eq:pr is 1},~\eqref{eq:caseii second eq} follows from the values of boundary rate pairs $(R_1^*,R_2^*)=(C(P_1/(1+P_2)),C(P_2))$ for case (ii), and \eqref{eq:caseii third eq} follows form symmetry since $S_2$ and $S_3$ are zero-mean Gaussian random variables.
To make~\eqref{eq: mac jnn case ii final} less than $\varepsilon$, one requires
\begin{align}
\Pr\Bigg\{\begin{bmatrix}
S_2\\
S_3
\end{bmatrix}\leq\begin{bmatrix}
L_2\\
L_1+L_2
\end{bmatrix}\Bigg\}\geq1-\varepsilon.
\end{align}
Recall the definition of $\rmQ_{\rm inv}$ in~\eqref{eq:define Qinv}. The corresponding second-order rate region of mismatched MAC with JNN decoding for case (ii) satisfies
\begin{align}
\calL_{\rm JNN}\supseteq \{(L_1,L_2):(L_2,L_1+L_2)\in \rmQ_{\rm inv}(\bV_{1}(P_1,P_2),\varepsilon)\}.
\end{align}

Next consider case (iii), where $(R_1^*,R_2^*)=(\alpha C(P_1)+\bar{\alpha}C(P_1/(1+P_2)),\alpha C(P_2/(1+P_1))+\bar{\alpha}C(P_2))$ for some $\alpha\in(0,1)$. Since $R_1^*<C(P_1)$ and $R_2^*<C(P_2)$, it follows that
\begin{align}
\lim_{n\to\infty}\Pr\{\calA_1\}=\lim_{n\to\infty}\Pr\{\calA_2\}=1.
\end{align}
Analogously to~\eqref{eq:pr is 1}-\eqref{eq: mac jnn case i final}, the right hand side of \eqref{eq:three s inequal} satisfies
\begin{align}
\liminf_{n\to\infty}1-\Pr\{\calA_1\cap\calA_2\cap\calA_3\}=\Pr\{S_3\geq L_1+L_2\}.\label{eq: mac jnn case iii final}
\end{align}
To make~\eqref{eq: mac jnn case iii final} less than $\varepsilon$, one requires
\begin{align}
L_1+L_2\leq\sqrt{V_{12}(P_1,P_2)}\rmQ^{-1}(\varepsilon).
\end{align}
Thus, the second-order rate region of mismatched MAC with JNN decoding for case (iii) satisfies
\begin{align}
\calL_{\rm JNN}(R_1^*,R_2^*,\varepsilon)\supseteq\{(L_1,L_2):L_1+L_2\geq\sqrt{V_{12}(P_1,P_2)}\rmQ^{-1}(\varepsilon)\}.
\end{align}
The analyses for cases (iv) and (v) are omitted due to their symmetry with cases (ii) and (i), respectively. The proof of Corollary~\ref{CORO:MAC JNN ACHIEVABILITY} is now completed.

\subsection{Proof for Case (ii)  of Theorem~\ref{theorem:MAC-SIC ach} (MAC with SIC decoding)}
\label{sec:mac sic proof}
Fix any $\varepsilon\in(0,1)$ and $(\varepsilon_1,\varepsilon_2)\in\bbR_+^2$ such that $\varepsilon_1+\varepsilon_2\leq \varepsilon$. It follows from Lemma~\ref{lemma:mismatch RCU SIC} that the error probability $\rmP_\rme^n$ of the mismatched code in Def. \ref{def:MAC SIC coding scheme} satisfies
\begin{align}
\nn\rmP_\rme^n 
&\leq \bbE
\Big[ \min
\Big\{ 1,(M_1-1)\Pr\{ \tilde{\imath}^n(\barX_1^n;Y^n)\geq\tilde{\imath}^n(X_1^n;Y^n) | X_1^n,Y^n \}\Big\} \Big]\\
&\qquad+\bbE
\Big[ \min
\Big\{ 1,(M_2-1)\Pr\{ \tilde{\imath}_2^n(\barX_2^n;Y^n|X_1^n)\geq\tilde{\imath}_2^n(X_2^n;Y^n|X_1^n) | X_1^n,X_2^n,Y^n \}\Big\} \Big]\label{proof:mac:sic:step0}.
\end{align}

For simplicity, let the first and second terms of the right-hand side of \eqref{proof:mac:sic:step0} be denoted as $\rmP_{\rme,1}$ and $\rmP_{\rme,2}$, respectively. It follows that
\begin{align}
\rmP_{\rme,1}
&=\bbE[\min\{1,M_1\Pr\{\tilde{\imath}_1(\barX_1^n;Y^n)\geq\tilde{\imath}_1(X_1^n;Y^n)|X_1^n,Y^n\}\}] \label{eq:using RCU}\\
&\leq\bbE[\min\{1,M_1g(\tilde{\imath}_1(X_1^n;Y^n),Y^n)\}]\label{eq:take g in mac sic}\\
&\leq\Pr\left\{M_1\beta e^{-\tilde{\imath}_1(X_1^n;Y^n)}\geq\frac{1}{\sqrt{n}}\right\}+\frac{1}{\sqrt{n}}\label{eq:using gfunction}\\
\nn&\leq\Pr\bigg\{\frac{\|X_2^n+Z^n\|^2}{2(1+P_2)}-\frac{\|X_1^n+X_2^n+Z^n\|^2}{2(1+P_1+P_2)}\\*
&\qquad\qquad\qquad\geq nC\left(\frac{P_1}{1+P_2}\right)-\log M_1-\log(\beta\sqrt{n})\bigg\}+\frac{1}{\sqrt{n}},\label{eq:take i in}
\end{align}
where~\eqref{eq:take g in mac sic} follows from the definition of function $g$ in~\eqref{eq:define g},~\eqref{eq:using gfunction} follows the result in Lemma~\ref{lemma:g function} and~the inequality in \cite[Eq. (37)]{scarlett2017mismatch}, which states that for any real-valued random variable $J$ and integer $n\in\bbN$,
\begin{align}
\bbE[\min\{ 1,J \}]\leq\Pr\left\{J\geq\frac{1}{\sqrt{n}}\right\}+\frac{1}{\sqrt{n}},\label{eq:min 1 j}
\end{align}
\eqref{eq:take i in} follows by taking logarithm at both sides of the inequality and using the definition of mismatched information density of $\tilde{\imath}_1(\cdot)$ in~\eqref{eq:mismatch information density} and the channel model $Y^n=X_1^n+X_2^n+Z^n$. For ease of notation, define two terms:
\begin{align}
I&:=(1+P_1+P_2)\|X_2^n+Z^n\|^2-(1+P_2)\|X_1^n+X_2^n+Z^n\|^2,\label{eq:LHS}\\
J&:=2(1+P_1+P_2)(1+P_2)\bigg(nC\left(\frac{P_1}{1+P_2}\right)-\log M_1-\log(\beta\sqrt{n})\bigg).\label{eq:RHS}
\end{align}
This way, \eqref{eq:take i in} is equivalent to
\begin{align}
\label{eq: mac sic e1}
\rmP_{\rme,1}\leq\Pr\{I\geq J\}+\frac{1}{\sqrt{n}}.
\end{align}

For each $(i,j)\in[n]\times [2]$, let $\tilX_{ji}$ be generated from the Gaussian distribution with mean zero and variance one, i.e., $\calN(\cdot;0,1)$. For subsequent analyses, for each $i\in[n]$, define the following random variables:
\begin{align}
&B_{1i}:= Z_i^2-1, B_{2i}:=\sqrt{P_1}\tilX_{1i}Z_i, B_{3i}:=\tilX_{1i}^2 - 1,\\
&B_{4i}:= \sqrt{P_2}\tilX_{2i}Z_i, B_{5i}:= \tilX_{2i}^2 - 1, B_{6i}:=\sqrt{P_1P_2}\tilX_{1i}\tilX_{2i},\nn
\end{align}
and let $\bB_i := (B_{1i},\dots,B_{6i})$ denote the random vector. It follows that $(\bB_1,\ldots,\bB_n)$ are a sequence of i.i.d. random vectors with finite third moment and for each $i\in[n]$,
\begin{align}
\cov(\bB_i) =\diag\begin{bmatrix}
\xi-1,P_1,2,P_2 ,2 ,P_1P_2\end{bmatrix}.
\end{align}

Given any real numbers $\bb:=(b_1,\dots,b_6)\in\bbR^6$, define the function:
\begin{align}
f(\bb):=P_1b_1-2(1+P_2)\frac{b_2}{\sqrt{1+b_3}}+2P_1\frac{b_4}{\sqrt{1+b_5}}-2(1+P_2)\frac{b_6}{\sqrt{1+b_3}\sqrt{1+b_5}}.
\end{align}
The Jacobian matrix of $f(\cdot)$ at $\bbE[\bB_1]=\bzero_6$ is 
\begin{align}
\bJ(\bB)=\begin{bmatrix}P_1,-2(1+P_2),0,2P_1,0,-2(1+P_2)\end{bmatrix}.
\end{align}
Thus, we have
\begin{align}
\bJ(\bB)\cov(\bB_1)\bJ(\bB)^{\rmT}
&=P_1^2(\xi-1+4P_2)+4P_1(1+P_2)^3=:\sigma^2.
\end{align}
Furthermore, let
\begin{align}
\nn nf\left(\frac{1}{n}\sum_{i=1}^{n}\bB_i\right)
&=P_1(\|Z^n\|^2-n)+2P_1\langle X_2^n,Z^n\rangle-2(1+P_2)\langle X_1^n,Z^n\rangle\\*
&\qquad-2(1+P_1)\langle X_1^n,X_2  ^n\rangle\\
&=I.
\end{align}
It follows from using~\eqref{eq:LHS} to \eqref{eq: mac sic e1} and the function version Berry-Esseen theorem (Lemma \ref{lemma:berry esseen for func}) that 
\begin{align}
\rmP_{\rme,1}&\leq\Pr\{I\geq J\}+\frac{1}{\sqrt{n}}\label{eq:before apply Qinv}\\
&\leq \rmQ \left(\frac{J}{\sqrt{n}\sigma}\right)+O\left(n^{-1/4}\right).\label{eq:apply Qinv}
\end{align}
To make~\eqref{eq:apply Qinv} less than $\varepsilon_1$, one requires
\begin{align}
\frac{J}{\sqrt{n}\sigma}\geq \rmQ^{-1}\left(\varepsilon_1+O\left(n^{-1/4}\right)\right).\label{eq:apply taylor}
\end{align}
Applying Taylor Series expansion of $\rmQ^{-1}(\cdot)$ on $\varepsilon_1$ and using the definition of $J$ in~\eqref{eq:RHS}, the maximal value of $M_1$ such that $\rmP_{\rme,1}\leq \varepsilon_1$ satisfies
\begin{align}
\log M_1 &= nC\left(\frac{P_1}{1+P_2}\right)-\sqrt{nV_{1}(P_1,P_2)}\rmQ^{-1}(\varepsilon_1)+O(\log n).\label{eq: lower bound for sic 1}
\end{align}
Thus, a lower bound $L_1\geq \sqrt{nV_{1}(P_1,P_2)}\rmQ^{-1}(\varepsilon_1)$ is obtained.

Next we consider $\rmP_{\rme,2}$. Given that message $W_1$ is decoded correctly, the remaining error probability satisfies
\begin{align}
\rmP_{\rme,2}
&=\bbE[\min\{1,M_2\Pr\{\tilde{\imath}_2(\barX_2^n;Y^n|X_1^n)\geq\tilde{\imath}_2^n(X_2^n;Y^n|X_1^n)|X_1^n,X_2^n,Y^n\}\}]\label{eq:mac sic e2 rcu}\\*
&\leq\bbE[\min\{1,M_2g_2(\tilde{\imath}_2(X_2^n;Y^n|X_1^n),Y^n)\}]\label{eq:mac sic e2 take g}\\*
&\leq\Pr\left\{M_2\beta_2e^{-\tilde{\imath}_2(X_2^n;Y^n|X_1^n)}\geq\frac{1}{\sqrt{n}}\right\}+\frac{1}{\sqrt{n}}\label{eq:apply g func 2}\\*
&\leq\Pr\bigg\{\frac{\|Z^n\|^2}{2}-\frac{\|X_2^n+Z^n\|^2}{2(1+P_2)}\geq nC(P_2)-\log M_2-\log(\beta_2\sqrt{n})\bigg\}+\frac{1}{\sqrt{n}}\label{eq:same form as p2p},
\end{align}
where~\eqref{eq:mac sic e2 take g} follows from the definition of $g_2(\cdot)$ in~\eqref{eq:define g2},~\eqref{eq:apply g func 2} follows from the result in Lemma~\ref{lemma:g function} and the inequality in~\eqref{eq:min 1 j}, and~\eqref{eq:same form as p2p} follows from the definition of mismatched information density $\tilde{\imath}_2(\cdot)$ in~\eqref{eq:mismatch information density1,2}. 

Note that~\eqref{eq:same form as p2p} has the same form as the mismatched P2P case  in~\cite[Eq. (42)]{scarlett2017mismatch}. Thus, it follows from \cite[Theorem 1]{scarlett2017mismatch} that to ensure $\rmP_{\rme,2}\leq \varepsilon_2$, the maximal value of $M_2$ satisfies
\begin{align}
\log M_2&=nC(P_2)-\sqrt{nV(P_2)}\rmQ^{-1}(\varepsilon_2)+O(\log n).\label{eq: lower bound for sic 2}
\end{align}
A lower bound $L_2\geq \sqrt{nV(P_2)}\rmQ^{-1}(\varepsilon_2) $ is thus obtained.

The proof for case (ii) of Theorem~\ref{theorem:MAC-SIC ach} is completed by taking a union over all possible $(\varepsilon_1,\varepsilon_2)\in\bbR_+^2$ such that $\varepsilon_1+\varepsilon_2\leq \varepsilon$.

\section{Proof of Second-Order Asymptotics for RAC (Theorems~\ref{THEOREM:RAC JNN ACHIEVABILITY} and~\ref{THEOREM:RAC SIC ACHIEVABILITY})}
\label{sec:rac proof}
\subsection{Preliminaries}
\label{sec:rac proof preliminaries}
In this section, we present preliminaries to prove achievable second-order asymptotics for RAC in Theorems~\ref{THEOREM:RAC JNN ACHIEVABILITY} and~\ref{THEOREM:RAC SIC ACHIEVABILITY}.
Fix an integer $k\in[K]$. Fix any integers $(t,n_k)\in[k]\times \bbN$ and channel input-output vectors $(x_{[k]}^{n_k}, y^{n_k})\in(\bbR^{n_k})^{k+1}$. The generalized version of mismatched (conditional) information densities is defined as
\begin{align}
\tilde{\imath}^{n_k}_t(x_{[t]}^{n_k};y^{n_k}|x_{[t+1:k]}^{n_k})&:=n_kC(tP)+\frac{\| y^{n_k}-\sum_{i\in[t+1:k]}x^{n_k}_i \|^2}{2(1+tP)}-\frac{\| y^{n_k}-\sum_{i\in[k]}x^{n_k}_i \|^2}{2},\label{eq:RAC mis id jnn}\\
\tilde{\imath}_t^{n_k}(x_t^{n_k};y^{n_k}|x_{[t-1]}^{n_k})&:={n_k}C\bigg(\frac{P}{1+(k-t)P}\bigg)+\frac{\|y^{n_k}-\sum_{i\in[t-1]}x_i^{n_k}\|^2}{2(1+(k-t+1)P)}-\frac{\|y^{n_k}-\sum_{i\in[t]}x_i^{n_k}\|^2}{2(1+(k-t)P)}\label{eq:RAC mis id sic},
\end{align}
where $x_{[t]}^{n_k}=(x_1^{n_k},\ldots,x_t^{n_k})$ and $x_{[t+1:k]}^{n_k}=(x_{t+1}^{n_k},\ldots,x_k^{n_k})$.
When 
$t=k$, $x_{[t+1:k]}^{n_k}=0$ and $\sum_{i\in[t+1:k]}x^{n_k}_i=0$; when $t=1$, $x_{[t-1]}^{n_k}=0$ and $\sum_{i\in[t-1]}x_i^{n_k}=0$. The above definitions are consistent with the information density in a Gaussian RAC~\cite[Eq. (10)]{yavas2021gaussianmac}.  Analogously to \eqref{eq:define g}, given any positive real number $s\in\bbR_+$, define the following probability functions
\begin{align}
\tilg_t(s;y^{n_k},x_{[t+1:k]}^{n_k})&:=\Pr\left\{ \tilde{\imath}_t(\bar{X}_{[t]}^{n_k};y^{n_k}|x_{[t+1:k]}^{n_k})\geq s\right\}\label{eq:g function def rac jnn},\\
\barg_t(s;y^{n_k},x_{[t-1]}^{n_k})&:=\Pr\left\{ \tilde{\imath}_t(\bar{X}_{t}^{n_k};y^{n_k}|x^{n_k}_{[t-1]})\geq s\right\}\label{eq:g function def rac sic},
\end{align}
where $\bar{X}_{[t]}^{n_k}=(\barX_1^{n_k},\ldots,\barX_t^{n_k})$ is a random vector with distributions~\eqref{eq:first n0 codewords distribution} and~\eqref{eq:nj-1 to nj codewords distribution}.

Analogous to the proof of Lemma~\ref{lemma:g function}, we obtain the following result.
\begin{lemma}
\label{lemma:g function rac}
For each $t\in[k]$, there exists positive real numbers $(\beta_{t,1},\beta_{t,2})\in\bbR_+^2$ that only depend on $P$ such that for any $s\in\bbR_+$,
\begin{align}\label{eq:g function lemma result}
\tilg_t(s;y^{n_k},x_{[t+1:k]}^{n_k})\leq \beta_{t,1}\exp(-s),\\
\barg_t(s;y^{n_k},x_{[t-1]}^{n_k})\leq \beta_{t,2}\exp(-s).
\end{align}
\end{lemma}

\subsection{Proof of Theorem~\ref{THEOREM:RAC JNN ACHIEVABILITY} (RAC with JNN decoding)}
\label{sec:rac jnn proof}
In this section, we prove Theorem~\ref{THEOREM:RAC JNN ACHIEVABILITY} that gives the achievability bound of a mismatched RAC with JNN decoding. The proof judiciously combines the techniques in~\cite{scarlett2017mismatch,yavas2021gaussianmac}. While the decomposition of the error events follows a similar framework to~\cite{yavas2021gaussianmac}, our proof is different in bounding the ensemble error probability. Since the specific upper bound in~\cite[Eq. (195)-(197)]{yavas2021gaussianmac} was established for the Gaussian RAC, we provide a weaker bound in Lemma~\ref{lemma:g function rac} to handle arbitrary noise distributions.

Fix any $k\in[K]$. It follows from the mismatched code in Section~\ref{sec:RAC defs} that when there are $k$ active users among $K$ users, a decoding error event occurs if one of the following events occurs:
\begin{itemize}
\item $\calE_{\rm rep}$: two different active users transmit the same message,
\item $\calE_{{\rm time}}$: the receiver stops at a decoding time $n_t<n_k$ for some $t<k$,
\item $\calE_{{\rm msg}}$: the receiver stops at $n_t$ but decodes transmitted messages incorrectly.
\end{itemize}
Thus, the error probability satisfies
\begin{align}
\label{eq:rac error prob def}
\rmP_{\rme, k}^n
&\leq \Pr\{\calE_{\rm rep}\cup \calE_{{\rm time}}\cup\calE_{{\rm msg}}\}\\
&\leq \Pr\{\calE_{{\rm rep}}\}+\Pr\{\calE_{{\rm time}}\mathrm\cap\calE_{{\rm rep}}^\rmc\}+\Pr\{\calE_{{\rm msg}}\mathrm\cap\calE_{{\rm rep}}^\rmc\}\label{eq:rac error event condition}.
\end{align}

Notice that we require the uniqueness of codewords from each active user since every user adopts the same codebook and encoder, and thus, the repetition of codewords will require extra treatment. Such assumption is consistent with~\cite{yavas2021gaussianmac}, and the probability $\Pr\{\calE_{\rm rep}\}$ satisfies
\begin{align}
\Pr\{\calE_{\rm rep}\} 
&\leq \sum_{(i,j)\in[K]^2:i\neq j}\Pr\{W_j=W_j\}\\
&\leq \sum_{(i,j)\in[K]^2:i\neq j} \sum_{m\in[M]}\Pr\{W_i=W_j=m\}\\
&=\sum_{(i,j)\in[K]^2:i\neq j} M\Pr\{W_i=W_j=1\}\\
&= \sum_{(i,j)\in[K]^2:i\neq j}M\frac{1}{M^2}\\
&=\frac{K(K-1)}{M}\label{eq:rac rcu first term bound}.
\end{align}

We next bound $\Pr\{\calE_{{\rm time}}\mathrm\cap\calE_{{\rm rep}}^\rmc\}$. When $k\in[K]$ users are active, it follows from Def.~\ref{def:RAC JNN coding scheme} that
\begin{align}
\Pr\{\calE_{\rm time}\mathrm\cap\calE_{{\rm rep}}^\rmc\}
&\leq \Pr\bigg\{ \bigcup_{t:1\leq t< k} \bigg\{ \bigg|\frac{1}{n_t}\|Y^{n_t}\|^2-(1+tP)\bigg|\leq \lambda_t \bigg\}\bigcup\bigg\{ \bigg|\frac{1}{n_k}\|Y^{n_k}\|^2-(1+kP)\bigg|> \lambda_k \bigg\}\bigg\}\label{eq:rac rcu second term},
\end{align}
where the first part concerns the case where the receiver stops too early, and the second part concerns the case where the receiver does not stop at $n_k$. In Appendix~\ref{sec:appendix etime}, we prove that
\begin{align}
\Pr\{\calE_{\rm time}\mathrm\cap\calE_{{\rm rep}}^\rmc\}\leq O(n_k^{-\frac{1}{2}}).\label{eq:rac_second_error_term_bound}
\end{align}

Finally, we bound the probability $\Pr\{\calE_{{\rm msg}}\mathrm\cap\calE_{{\rm rep}}^\rmc\}$ that the receiver decodes messages of active users wrongly when the number of active users $k$ is estimated correctly.Let $(W_1,\ldots,W_k)\in[M]^k$ be the transmitted messages. Recall that the codebook consists of codewords $(X_1^{n_K},\ldots,X_M^{n_K})$. For any $k\in[K]$, $t\in[k]$ and $w^t=(w_1,\ldots,w_t)\in[M]^t$, $X^{n_k}(w^t)=(X_{w_1}^{n_k},\ldots,X_{w_t}^{n_k})$ denotes a collection of $t$ codewords, each with length $n_k$. Similarly, for any $w_{t+1}^k=(w_{t+1},\ldots,w_k)$,  $X^{n_k}(w_{t+1}^k)=(X_{w_{t+1}}^{n_k},\ldots,X_{w_k}^{n_k})$.

For each $t\in[k]$, let $\calE_t$ denote the error event, where $t$ out of $k$ messages are decoded incorrectly. It follows from the union bound that
\begin{align}
\Pr\{\calE_{{\rm msg}}\mathrm\cap\calE_{{\rm rep}}^\rmc\}
&\leq \Pr\bigg\{ \bigcup_{t\in[k]}\calE_t \mathrm\cap\calE_{{\rm rep}}^\rmc\bigg\}\\
&\leq \sum_{t\in[k]}\Pr\big\{\calE_t\mathrm\cap\calE_{{\rm rep}}^\rmc\big\}\label{eq:RCU first union bound}.
\end{align}

Fix any $t\in[k]$. Given that the error event $\calE_t\mathrm\cap\calE_{{\rm rep}}^\rmc$ happens, assume that the first $t$ messages $(W_1,\ldots,W_t)$ out of all messages $(W_1,\ldots,W_k)$ are decoded wrongly. In this case, the JNN decoder operates as follows:
\begin{align}
\Psi_k(Y^{n_k})&=\argmin_{(\barw_1,\ldots,\barw_t)\in[M]^t}\Big\|{Y^{n_k}-\sum_{i=1}^{t} X^{n_k}(\barw_i)-\sum_{j=t+1}^{k} X^{n_k}(W_j)}\Big\|^2 \\
&\nn= \argmax_{(\barw_1,\ldots,\barw_t)\in[M]^t} \frac{n_k}{2}\log(1+tP)+\frac{\| Y^{n_k}-\sum_{j=t+1}^{k}X^{n_k}(W_j) \|^2}{2(1+tP)}\\
&\qquad\qquad\qquad\;\;\;\;\;-\frac{\| Y^{n_k}-\sum_{i=1}^{t}X^{n_k}(\barw_i)-\sum_{j=t+1}^{k}X^{n_k}(W_j) \|^2}{2}\\
&=\argmax_{(\barw_1,\ldots,\barw_t)\in[M]^t}\tilde{\imath}_t^{n_k}(X^{n_k}(\barw^t);Y^{n_k}|X^{n_k}(W_{t+1}^k))\label{eq:nn equivalent RAC},
\end{align}
where \eqref{eq:nn equivalent RAC} follows from the definition of $\tilde{\imath}_t^{n_k}(\cdot)$ in \eqref{eq:RAC mis id jnn}, $\barw^t=(\barw_1,\ldots,\barw_t)$, $X^{n_k}(\barw^t)=(X^{n_k}(\barw_1),\ldots,X^{n_k}(\barw_t))$, $W_{t+1}^k=(W_{t+1},\ldots,W_k)$ and $X^{n_k}(W_t^k)=(X^{n_k}(W_{t+1}),\ldots,X^{n_k}(W_k))$.

The result in~\eqref{eq:nn equivalent RAC} implies that under error event $\calE_t\mathrm\cap\calE_{{\rm rep}}^\rmc$, maximizing the mismatched information density $\tilde{\imath}_s^{n_k}(X^{n_k}(\barw^t);Y^{n_k}|X^{n_k}(W_t^k))$ over $\barw_{[t]}\in[M]^t$ is equivalent to JNN decoding. For ease of notation, given $\barw^t\in[M]^t$, define the event 
\begin{align}
\calB(\barw^t)
:=\Big\{\tilde{\imath}_t^{n_k}(X^{n_k}(\barw^t);Y^{n_k}|X^{n_k}(W_{t+1}^k))\geq\tilde{\imath}_t^{n_k}(X^{n_k}(W^t);Y^{n_k}|X^{n_k}(W_{t+1}^k))\Big\}.
\end{align}

It follows that
\begin{align}
\nn&\Pr\{\calE_t\mathrm\cap\calE_{{\rm rep}}^\rmc\}\\*
&\leq \binom{k}{t}\Pr\Big\{ \bigcup_{\barw^t\in[M]^t:~\barw^t\neq W^t} \calB(\barw^t)\Big\}\label{eq:RAC rcu 1}\\
&\leq \bbE\bigg[\binom{k}{t}\Pr\Big\{ \bigcup_{\barw^t\in[M]^t:~\barw^t\neq W^t} \calB(\barw^t)|X^{n_k}(W^k),Y^{n_k}\Big\}\bigg]\label{eq:RAC rcu 2}\\
&\leq \bbE\Bigg[\binom{k}{t}\binom{M-k}{t} 
\Pr\{\calB(\barw^t)|X^{n_k}(W^k),Y^{n_k}\}
\Bigg]\label{eq:RAC rcu 3}\\
&\leq\bbE\bigg[\min\bigg\{1,\binom{k}{t}\binom{M-k}{t}\Pr\{\calB(\barw^t)|X^{n_k}(W^k),Y^{n_k}\}
\bigg\}\bigg]\label{eq:bounding decoding error prob-1}\\
&\leq \bbE\bigg[ \min\bigg\{ 1,\binom{k}{t}M^t\tilg_t(\tilde{\imath}_t(X^{n_k}(W^t);Y^{n_k}|X^{n_k}(W_{t+1}^k));Y^{n_k},X^{n_k}(W_{t+1}^k)) \bigg\} \bigg]\label{eq:bounding decoding error prob-2}\\
&\leq \bbE\bigg[ \min\bigg\{ 1, \binom{k}{t}M^t\beta_{t,1}\exp(-\tilde{\imath}_t^{n_k}(X^{n_k}(W^t);Y^{n_k}|X^{n_k}(W_{t+1}^k)))\bigg\} \bigg]\label{eq:bounding decoding error prob-3}\\
&\leq  \Pr\Bigg\{ \binom{k}{t}M^t\beta_{t,1}\exp(-\tilde{\imath}_t^{n_k}(X^{n_k}(W^t);Y^{n_k}|X^{n_k}(W_{t+1}^k)))\geq1/\sqrt{n_k} \Bigg\}+\frac{1}{\sqrt{n_k}} ,\label{eq:bounding decoding error prob-4}
\end{align}
where~\eqref{eq:RAC rcu 1} follows due to symmetry and the fact that there are $k\choose t$ possibilities of decoding $t$ messages incorrectly among $k$ messages, \eqref{eq:RAC rcu 2} follows from the law of total expectation, \eqref{eq:RAC rcu 3} follows from symmetry of the codebook, \eqref{eq:bounding decoding error prob-1} follows since any probability is no greater than 1,  \eqref{eq:bounding decoding error prob-2} follows from the definition of $\tilg$ in~\eqref{eq:g function def rac jnn} and the fact that $\binom{M-k}{t}\leq M^t$, \eqref{eq:bounding decoding error prob-3} follows from the result in Lemma~\ref{lemma:g function rac}, and~\eqref{eq:bounding decoding error prob-4} follows from the inequality in~\eqref{eq:min 1 j}. 

The probability term in~\eqref{eq:bounding decoding error prob-4} can be further upper bounded as follows:
\begin{align}
\nn&\Pr\left\{ \binom{k}{t}M^t\beta_{t,1}\exp(-\tilde{\imath}_t^{n_k}(X^{n_k}(W^t);Y^{n_k}|X^{n_k}(W_{t+1}^k)))\geq\frac{1}{\sqrt{n_k}} \right\}\\
&\leq\Pr\left\{\tilde{\imath}_t^{n_k}(X^{n_k}(W^t);Y^{n_k}|X^{n_k}(W_{t+1}^k))\leq \log\left( \binom{k}{t}M^t\beta_{t,1}\sqrt{n_k} \right)  \right\}\label{eq:rac before be 1}\\
&\nn\leq\Pr\Bigg\{(\|Z^{n_k}\|^2-n_k)tP-2\sum_{i=1}^{t}\langle X_i^{n_k},Z^{n_k} \rangle-2\sum_{1\leq i<j\leq t}\langle X_i^{n_k},X_j^{n_k} \rangle \\
&\qquad\qquad\qquad\qquad\geq2(1+tP)\Bigg(n_kC(tP)- \log\Bigg( \binom{k}{t}M^t\beta_{t,1}\sqrt{n_k} \Bigg) \Bigg)  \Bigg\},\label{eq:rac before be 2}
\end{align}
where~\eqref{eq:rac before be 1} follows by taking logarithm at both side and~\eqref{eq:rac before be 2} follows from the definition of mismatched information density in~\eqref{eq:RAC mis id jnn} and the fact that each codeword is generated independently from the same distribution.

Fix any $i\in[n]$. Define the set $\calD_t:=\{(a,b)\in[t]^2:~a\neq b\}$. For each $l\in[t]$ and $(p,q)\in\calD_t$, define the following zero-mean random variables:
\begin{align}
\tilA_{1,i} := Z_i^2-1,~
\tilA_{2,i}^{(l)}:=\sqrt{P}\tilX_{l,i}Z_i,~\tilA_{3,i}^{(l)}:=\tilX^2_{l,i}-1,~
\tilA_{4,i}^{(p,q)}:=P\tilX_{p,i}\tilX_{q,i},\label{eq:rac jnn a last}
\end{align}
and let 
\begin{align}
\tilde{\bA}_{i} := (\tilA_{1,i},\{\tilA_{2,i}^{(l)}\}_{l\in[t]},\{\tilA_{3,i}^{(l)}\}_{l\in[t]},\{\tilA_{4,i}^{(p,q)}\}_{(p,q)\in\calD_t}).\label{eq:def tilbAvec}
\end{align}
Thus $\tilde{\bA}_{[n_k]}=(\tilde{\bA}_{1},\ldots,\tilde{\bA}_{n_k})$ are zero-mean  i.i.d. random vectors with finite third moments. The covariance matrix of $\tilde{\bA}_{i}$ is 
\begin{align}
\label{eq:cov til bA}
\cov(\tilde{\bA}_{1})=\diag\begin{bmatrix}
\xi -1,P\cdot\bone_t,2\cdot\bone_t,P^2\cdot\bone_{t(t-1)}
\end{bmatrix}.
\end{align}

Given any $a_1\in\bbR$, $\ba_2=(a_{2,1},\ldots,a_{2,t})\in\bbR^t$, $\ba_3=(a_{3,1},\ldots,a_{3,t})$ and $\ba_4=(a_4^{p,q})_{(p,q)\in\calD_t}\in\bbR^{t(t-1)}$, define the following function:
\begin{align}
\label{eq:f_s definition}
\tilf(a_1,\ba_2,\ba_3,\ba_4)
&:=tPa_1-\sum_{l\in[t]}\frac{2a_{2,t}}{\sqrt{1+a_{3,l}}}-\sum_{(p,q)\in\calD_t}\frac{2a_4^{(p,q)}}{\sqrt{1+a_{3,p}}\sqrt{1+a_{3,q}}}.
\end{align}
It follows that
\begin{align}
n_k\tilf\bigg(\frac{1}{n_k}\sum_{i\in[n_k]}\tilde{\bA}_{i}\bigg) = (\|Z^{n_k}\|^2-n_k)tP-2\sum_{i\in[t]}\langle X_i^{n_k},Z^{n_k} \rangle-2\sum_{(i,j)\in\calD_t}\langle X_i^{n_k},X_j^{n_k} \rangle,
\end{align}
which is exactly the left hand side of~\eqref{eq:rac before be 2}. The Jacobian of $\tilf$ evaluated at $\bzero$ is
\begin{align}
\bJ=\begin{bmatrix}
tP,-2\cdot\bone_t,\bzero_t,-2\cdot\bone_{t(t-1)}
\end{bmatrix}.
\end{align}
Thus,
\begin{align}
\bJ\cov(\tilde{\bA}_{1})\bJ^{\rm T}
&=tP((\xi+1)tP-2P+4)=:\tilde{\sigma}_t^2.
\end{align}

It follows from the function version Berry-Esseen Theorem (Lemma~\ref{lemma:berry esseen for func}) and the results in \eqref{eq:bounding decoding error prob-4} and~\eqref{eq:rac before be 2} that
\begin{align}
\label{eq:berry-esseen result}
\Pr\{\calE_t\mathrm\cap\calE_{{\rm rep}}^\rmc\}
&\leq \rmQ\Bigg(\frac{2(1+tP)\Big(n_kC(tP)- \log\Big( \binom{k}{t}M^t\beta_{t,1}\sqrt{n_k}\Big) \Big)}{\tilde{\sigma}_t\sqrt{n_k}}\Bigg) + O\left(n_k^{-1/4}\right).
\end{align}

Fix any $(\varepsilon_{k,1},\ldots,\varepsilon_{k,k})\in(0,1)^k$ such that $\sum_{t\in[k]}\varepsilon_{k,t}\leq \varepsilon_k$. Choose $M$ such that
\begin{align}
\log M&=\min_{t\in[k]}\frac{1}{t}\bigg(n_k C(tP)-\sqrt{n_k(V(tP)+V_{\rm{cr}}(t,P))}\rmQ^{-1}(\varepsilon_{k,t})+O(\log n_k)\bigg)\label{eq:rac jnn min}.
\end{align}
It follows from \eqref{eq:berry-esseen result} that for each $t\in[k]$,
\begin{align}
\Pr\{\calE_t\mathrm\cap\calE_{{\rm rep}}^\rmc\}\leq\varepsilon_{k,t}\label{rac:ach:dep:upp}
\end{align}

Since $\frac{1}{t}C(tP)>\frac{1}{k}C(kP)$ for $t<k$~\cite[Lemma 1]{yavas2021gaussianmac}, when $n_k$ is sufficiently large, the minimum of~\eqref{eq:rac jnn min} is achieved by $t=k$.

Thus, when $n_k$ is large enough, it follows from \eqref{eq:rac error event condition}, \eqref{eq:rac rcu first term bound}, \eqref{eq:rac_second_error_term_bound}, \eqref{eq:rac jnn min} and \eqref{rac:ach:dep:upp} that
\begin{align}
\log M&=\frac{1}{k}\bigg(n_k C(kP)-\sqrt{n_k(V(kP)+V_{\rm{cr}}(k,P))}\rmQ^{-1}(\varepsilon_{k,k})+O(\log n_k)\bigg),
\end{align}
and
\begin{align}
\rmP_{\rme,k}^{n_k}\leq \varepsilon_k.
\end{align} 
The proof of Theorem~\ref{THEOREM:RAC JNN ACHIEVABILITY} is completed by letting $\varepsilon_{k,t}\rightarrow 0$ for every $t\in[k-1]$ and $\varepsilon_{k,k}\rightarrow\varepsilon_k$.

\subsection{Proof of Theorem~\ref{THEOREM:RAC SIC ACHIEVABILITY} (RAC with SIC decoding)}
\label{sec:rac sic proof}

In this section, we prove Theorem~\ref{THEOREM:RAC SIC ACHIEVABILITY} that gives an achievability bound of a mismatched RAC with SIC decoding.
The error probability term in~\eqref{eq:rac error prob} can also be written as~\eqref{eq:rac error event condition}:
\begin{align}
\rmP_{\rme,k}^n\leq \Pr\{\calE_{{\rm rep}}\}+\Pr\{\calE_{{\rm time}}\cap\calE_{{\rm rep}}^\rmc\}+\Pr\{\calE_{{\rm msg}}\cap\calE_{{\rm rep}}^\rmc\},
\end{align}
where the first two terms is already bounded by $O(n_k^{-\frac{1}{2}})$ according to~\eqref{eq:rac rcu first term bound} and~\eqref{eq:rac_second_error_term_bound}. Thus, we only need to bound the third term that depicts the case when decoding messages are incorrect using SIC decoding. Since the SIC decoder decodes each codeword separately, the error occurs when the first codeword error occurs, i.e.,
\begin{align}
\Pr\{\calE_{\rm msg}\cap\calE_{\rm rep}^\rmc\}&=\Pr\{\{\hatW_1,\ldots,\hatW_k\}\neq\{W_1,\ldots,W_1\}\}\\
&=\Pr\{\{\hatW_1\neq W_1\}\cup\{\hatW_1= W_1,\hatW_2\neq W_2\}\cup\ldots\cup\{\{\hatW_i= W_i\}_{i=1}^{k-1},\hatW_k\neq W_k\}\}\\
&\leq\sum_{r=1}^k\Pr\{\{\hatW_i=W_i\}_{i=1}^{r-1},\hatW_r\neq W_r\}.\label{eq:msger in rac sic}
\end{align}
Then for any $r\in[k]$, the probability that message $W_r$ decoded wrongly is
\begin{align}
\nn&\Pr\{\{\hatW_i=W_i\}_{i=1}^{r-1},\hatW_r\neq W_r\}\\
&\leq\bbE[\min\{1,M\Pr\{\tilde{\imath}_r^{n_k}(X^{n_k}(\barW_r);Y^{n_k}|X^{n_k}(W^{r-1}))\geq \tilde{\imath}_r^{n_k}(X^{n_k}(W_r);Y^{n_k}|X^{n_k}(W^{r-1}))|X^{n_k}(W^{r}),Y^{n_k}\}\}]\label{eq:rac sic msger 1}\\
&\leq\Pr\{M\barg_r(\tilde{\imath}_r^{n_k}(X^{n_k}(W_r);Y^{n_k}|X^{n_k}(W^{r-1})),Y^{n_k};X^{n_k}(W^{r-1}))\geq1/\sqrt{{n_k}}\}+1/\sqrt{{n_k}}\label{eq:rac sic msger 2}\\
&\leq\Pr\{M\barG_r\exp(-\tilde{\imath}_r^{n_k}(X^{n_k}(W_r);Y^{n_k}|X^{n_k}(W^{r-1})))\geq1/\sqrt{{n_k}}\}+1/\sqrt{{n_k}}\label{eq:rac sic msger 3}\\
&\nn=\Pr\bigg\{(1+(k-r+1)P)\bigg\|Y^{n_k}-\sum_{i=1}^{r}X_i^{n_k}\bigg\|^2-(1+(k-r)P)\bigg\|Y^{n_k}-\sum_{i=1}^{r-1}X_i^{n_k}\bigg\|^2\geq\\ 
&\qquad\quad\;\;2(1+(k-r+1)P)(1+(k-r)P)\bigg({n_k}\rmC\bigg(\frac{P}{1+(k-r)P}\bigg)-\log M\barG_r-\log\sqrt{n_k}\bigg)\bigg\}+1/\sqrt{{n_k}}\label{eq:RAC SIC LHS},
\end{align}
where~\eqref{eq:rac sic msger 1} follows from RCU bound,~\eqref{eq:rac sic msger 2} follows from the definition of $\barg_r$ in~\eqref{eq:g function def rac sic} and the inequality in~\eqref{eq:min 1 j},~\eqref{eq:rac sic msger 3} follows from Lemma~\ref{lemma:g function rac}, and~\eqref{eq:RAC SIC LHS} follows from the definition of mismatched information density in~\eqref{eq:RAC mis id sic}.
For any $l\in[k]$ and $(p,q)\in[k]^2,p\neq q$, define the following zero-mean i.i.d. random variables for subsequent analyses:
\begin{align}
\nn&\tilB_{1,i}:=1-Z_i^2,
\tilB_{2,i}:=\sqrt{P}\tilX_{r,i}Z_i,
\tilB_{3,i}:=(\tilX_{r,i})^2-1,\\
&\tilB_{4,i}^{(l)}:=P\tilX_{r,i}\tilX_{l,i},
\tilB_{5,i}^{(l)}:=\sqrt{P}\tilX_{l,i}Z_i,
\tilB_{6,i}^{(l)}:=(\tilX_{l,i})^2-1,
\tilB_{7,i}^{(p,q)}:=P\tilX_{p,i}\tilX_{q,i},\label{eq:rac sic a last}
\end{align}
let $\tilde{\bB}_{i} := (\tilB_{1,i},\tilB_{2,i},\tilB_{3,i},\{\tilB_{4,i}^{(l)}\}_{l=1}^{k},\{\tilB_{5,i}^{(l)}\}_{l=1}^{k},\{\tilB_{6,i}^{(l)}\}_{l=1}^{k},\{\tilB_{7,i}^{(p,q)}\}_{p,q=1}^{k})$. Thus $\tilde{\bB}_{[n_k]}=(\tilde{\bB}_{1},\ldots,\tilde{\bB}_{n_k})$ are zero-mean  i.i.d. random vectors with finite third moments. One can verify that the covariance matrix of $\tilde{\bB}_{1}$ is 
\begin{align}
\cov(\tilde{\bB}_{1})=\diag\begin{bmatrix}
\xi -1,P,2,P^2\cdot\bone_k,P\cdot\bone_k,2\cdot\bone_k,P^2\cdot\bone_{k(k-1)}
\end{bmatrix}.
\end{align}
Furthermore, define the following function:
\begin{align}
\nn &\barf(a_1,a_2,a_3,\{a_4^{(l)}\}_{l=r+1}^k,\{a_5^{(l)}\}_{l=r+1}^k,\{a_6^{(l)}\}_{l=r+1}^k,\{a_7^{(p,q)}\}_{p,q=r+1}^k)\\*
&:=\nn P\Bigg(a_1-2\sum_{l=r+1}^{k}\frac{a_5^{(l)}}{\sqrt{1+a_6^{(l)}}}-2\sum_{r+1\leq p<q\leq k}\frac{a_7^{(p,q)}}{\sqrt{1+a_6^{(p)}}\sqrt{1+a_6^{(q)}}}\Bigg)\\
&\qquad+2(1+(k-r)P)\Bigg(\frac{a_2}{\sqrt{1+a_3}}+\sum_{l=r+1}^{k}\frac{a_4^{(l)}}{\sqrt{1+a_3}\sqrt{1+a_6^{(l)}}}\Bigg).\label{eq: RAC SIC f}
\end{align}
One can verify that
\begin{align}
\nn-nf\bigg(\frac{1}{n}\sum_{i=1}^{n}\tilde{\bB}_i\bigg)&=P(\|Z^n\|^2-n)+2P\sum_{i=r+1}^{k}\Braket{X_i^n,Z^n}+2P\sum_{r+1\leq i<j\leq k}\Braket{X_i^n,X_j^n}\\
&\quad-2(1+(k-r)P)\bigg(\Braket{X_r^n,Z^n}+\sum_{i=r+1}^{k}\Braket{X_r^n,X_i^n}\bigg).
\end{align}
which is the left-hand side of the probability term in~\eqref{eq:RAC SIC LHS}. The Jacobian of $\barf$ evaluated at $\bzero$ is
\begin{align}
\bJ=\begin{bmatrix}
P, 2(1+(k-r)P), 0, 2(1+(k-r)P)\cdot\bone_k, -2P\cdot\bone_k, 0\cdot\bone_k, -2P\cdot\bone_{k(k-1)}   
\end{bmatrix}.
\end{align}
Thus, for any $r\in[k]$, the dispersion term for the Berry-Esseen Theorem in Lemma~\ref{lemma:berry esseen for func} is given by
\begin{align}
\bar{\sigma_r}:=\bJ\cov(\tilde{\bB}_{1})\bJ^{\rm T}=P^2(\xi-1)+4(1+(k-r)P)^2P+4(k-r)(1+(k-r)P)^2P^2+4(k-r)P^3+4P^4\binom{k-r}{2}.
\end{align}
Applying the Berry-Esseen Theorem in Lemma~\ref{lemma:berry esseen for func} and invoking~\eqref{eq:RAC SIC LHS}, we obtain
\begin{align}
\nn\Pr\{\{\hatW_i=W_i\}_{i=1}^{r-1},\hatW_r\neq W_r\}&\leq \rmQ\bigg(\frac{2(1+(k-r+1)P)(1+(k-r)P)(n_k\rmC\big(\frac{P}{1+(k-r)P}\big)-\log M\barG_r-\log\sqrt{n_k})}{\bar{\sigma_r}\sqrt{n_k}}\bigg)\\*
&\qquad+O\left(n_k^{-1/4}\right).\label{eq:rac sic after be}
\end{align}
For any $\varepsilon_{k,r}>0$, let
\begin{align}
\log M&\leq \min_{r\in[k]}n_k \rmC\bigg(\frac{P}{1+(k-r)P}\bigg)-\sqrt{n_kV_{\rm{rs}}(k,r,P)}\rmQ^{-1}(\varepsilon_{k,r})+O(\log n_k),\label{eq:rac sic min}
\end{align}
we have $\Pr\{\{\hatW_i=W_i\}_{i=1}^{r-1},\hatW_r\neq W_r\}\leq\varepsilon_{k,r}$ for all $r\in[k]$. Since $\rmC\bigg(\frac{P}{1+(k-r)P}\bigg)$ is monotonically decreasing with $r$, the minimum of~\eqref{eq:rac sic min} reaches when $r=1$.
Thus, the probability that the message set is decoded wrong in~\eqref{eq:msger in rac sic} can be bounded as
\begin{align}
\Pr\{\calE_{{\rm msg}}\cap\calE_{{\rm rep}}^\rmc\}&\leq\sum_{r=1}^k\Pr\{\{\hatW_i=W_i\}_{i=1}^{r-1},\hatW_r\neq W_r\}\\
&\leq \sum_{r=1}^{k}\varepsilon_{k,r}=\varepsilon_k,\label{eq:rac sic final}
\end{align}
where~\eqref{eq:rac sic final} follows by letting $\varepsilon_{k,r}\rightarrow0$ for every $r\in[2:k]$ and $\varepsilon_{k,1}\rightarrow\varepsilon_k$.The proof of Theorem~\ref{THEOREM:RAC SIC ACHIEVABILITY} is thus completed by combining~\eqref{eq:rac rcu second term} and~\eqref{eq:rac sic final}.

\section{Conclusion}
\label{sec:conclusion}
We studied a two-user MAC and a RAC with additive arbitrary noise, proposed mismatched coding schemes using spherical codebooks with either JNN or SIC decoding, and bounded the achievable second-order rate regions. Our results implied that in a mismatched MAC, although the first-order rate regions of JNN and SIC decoding are identical, the SIC decoding is strictly inferior to JNN decoding in finite blocklength performance by having a smaller second-order rate region. Furthermore, our results showed that in the mismatched RAC, SIC decoding is inferior to JNN decoding even in the first-order rate region. When specialized to AWGN, our results yielded alternative proofs of corresponding results in~\cite{Ahlswede1973MultiwayCC} up to second-order.

There are several avenues for future research. Firstly, we only derived second-order achievability results. Without an ensemble converse result~\cite{scarlett2017mismatch,gallager_ensemble}, we could not check whether our derivations are tight. Thus, it is of great interest to derive a matching ensemble converse result. For this purpose, novel ideas beyond those for mismatched P2P cases~\cite{scarlett2017mismatch,zhou2019jscc,zhou2023sr} and Gaussian MAC/RAC~\cite{molavianjazi2015second,yavas2021gaussianmac} are required. Secondly, to establish the third-order asymptotic result, one needs to refine Lemmas~\ref{lemma:g function} and~\ref{lemma:g function rac} by having a multiplier $\frac{1}{\sqrt{n}}$ as in~\cite[Section IV.E]{tantomamichel2015} and~\cite[Lemma 6]{yavas2021gaussianmac} for the matched case with Gaussian noise. This effort is challenging and thus left as future work. Thirdly, we did not consider fading channels or multi-antenna communication. For practical next-generation communications, MIMO communication over quasi-static fading channels is highly relevant. To generalize our results, the ideas in~\cite{molavianjazi2013fading,weiyang2014mimo,zhou2019multiConnectivity} might be helpful. Finally, we focused on the reliable transmission of messages. In certain applications, covert transmission~\cite{zhou2023covert,yu2021awgncovert,shiyuan2021covertmimo} is critical to ensure communication secrecy. It is worthwhile to generalize our results to covert mismatched multi-user communication scenarios. For this purpose, one could borrow ideas from~\cite{zhou2020classification,zhou2023achievableerrorexponentstwophase,ShaoLun2023secondorderasymptoticscovert,bloch2019secondorderCovert}.

\appendix
\subsection{Proof of inferiority of SIC for mismatched RAC in the first-order rate region}
\label{sec:appendix rac first order}
To justify the inequality $\rmC(kP)/k>\rmC(P/(1+(k-1)P))$ for any integer $k\geq 2$, it suffices to show that $\mu(k):=k\rmC(P/(1+(k-1)P))-\rmC(kP)<0$ for $k\geq 2$. Using the definition of $\rmC(\cdot)$ in~\eqref{eq:capacity def}, $\mu(k)$ satisfies
\begin{align}
\mu(k)&=\frac{1}{2}\log\Big(\frac{1+kP}{1+(k-1)P}\Big)^k-\frac{1}{2}\log(1+kP)\\
&=\frac{1}{2}\log\frac{(1+kP)^{k-1}}{(1+(k-1)P)^k}\\
&=\frac{1}{2}\log\frac{\sum_{t=0}^{k-1}\binom{k-1}{k-1-t}(kP)^t}{\sum_{t=0}^{k}\binom{k}{k-t}((k-1)P)^t},\label{eq:comparing rac sic and jnn}
\end{align}
where~\eqref{eq:comparing rac sic and jnn} follows the polynomial expansion which states that $(1+a)^b=\sum_{t=0}^b\binom{b}{t}a^t$. 

For each $t\in[0:k-1]$, let $\mu_1(t):=\binom{k-1}{k-1-t}(kP)^t$ and $\mu_2(t):=\binom{k}{k-t}((k-1)P)^t$. This way,~\eqref{eq:comparing rac sic and jnn} is equivalent to
\begin{align}
\mu(k)&=\frac{1}{2}\log\frac{\sum_{t=0}^{k-1}\mu_1(t)}{\sum_{t=0}^{k-1}\mu_2(t)+\mu_2(k)}.\label{eq:simp comparing rac sic and jnn}
\end{align}
Note that $\mu_2(k)=((k-1)P)^k>0$. To show $\mu(k)<0$ for $k\geq 2$, it suffices to show that $\mu_2(t)\geq\mu_1(t)$ for each $t\in[0:k-1]$. It follows that
\begin{align}
\frac{\mu_2(t)}{\mu_1(t)}&=\frac{k}{k-t}\frac{(k-1)^t}{k^t}\\
&=\frac{(1-\frac{1}{k})^t}{1-\frac{t}{k}}\\
&=\left\{
\begin{array}{ll}
1 &  t=0,1\\
1+\frac{1}{1-\frac{t}{k}}\sum_{i=1}^{l}\bigg(\binom{t}{2i}\frac{1}{k^{2i}}-\binom{t}{2i+1}\frac{1}{k^{2i+1}}\bigg) &  t=2l+1,l\in\bbN_+, \\ 
1+\frac{1}{1-\frac{t}{k}}\bigg(\sum_{i=1}^{l-1}\bigg(\binom{t}{2i}\frac{1}{k^{2i}}-\binom{t}{2i+1}\frac{1}{k^{2i+1}}\bigg)+\frac{1}{k^{2l}}\bigg)  & t=2l,~l\in\bbN_+.
\end{array}
\right.\label{eq:mu2 over mu1}
\end{align}
For each $i\in\bbN_+$, we have
\begin{align}
\binom{t}{2i}\frac{1}{k^{2i}}-\binom{t}{2i+1}\frac{1}{k^{2i+1}}&=\frac{t!}{(t-2i)!(2i)!k^{2i}}-\frac{t!}{(t-2i-1)!(2i+1)!k^{2i+1}}\\
&=\frac{t!}{(t-2i-1)!(2i)!k^{2i+1}}\bigg(\frac{k}{t-2i}-\frac{1}{2i+1}\bigg)>0,\label{eq:mu2 over mu1 reason}
\end{align}
where~\eqref{eq:mu2 over mu1 reason} follows because $t\leq k-1$, $i\leq l$ and $t\geq 2l$ imply that $k>t-2i$ and thus $\frac{k}{t-2i}>1>\frac{1}{2i+1}$. Therefore,~\eqref{eq:mu2 over mu1} shows that $\mu_2(t)\geq \mu_1(t)$ for each $t\in[0:k-1]$. Invoking~\eqref{eq:simp comparing rac sic and jnn},~\eqref{eq:mu2 over mu1} and~\eqref{eq:mu2 over mu1 reason}, it follows that $\mu(k)<0$ when $k\geq 2$, and thus $\rmC(kP)/k>\rmC(P/(1+(k-1)P))$.

\subsection{Proof of (178)}
\label{sec:appendix etime}
By the union bound, it suffices to prove
\begin{align}
\Pr\bigg\{ \bigg|\frac{1}{n_k}\|Y^{n_k}\|^2-(1+kP)\bigg|> \lambda_k \bigg\}\leq O(n_k^{-\frac{1}{2}}),\label{eq:time lambda k}
\end{align}
show that for each $t\in[k-1]$,
\begin{align}
\Pr\bigg\{ \bigg|\frac{1}{n_t}\|Y^{n_t}\|^2-(1+tP)\bigg|\leq \lambda_t \bigg\}\leq O(n_t^{-\frac{1}{2}}).\label{eq:time lambda t}
\end{align}
Using the channel model concerning $Y^{n_k}$ in~\eqref{eq:RAC channel output}, we have
\begin{align}
\nn&\Pr\bigg\{ \bigg|\frac{1}{n_k}\|Y^{n_k}\|^2-(1+kP)\bigg|> \lambda_k \bigg\}\\
=&\Pr\Bigg\{\Bigg|\frac{1}{n_k}\Bigg(\|Z^{n_k}\|^2-n_k+2\sum_{i=1}^{k}\langle X_i^{n_k},Z^{n_k} \rangle+2\sum_{1\leq i<j\leq k}\langle X_i^{n_k},X_j^{n_k} \rangle\Bigg)\Bigg|>\lambda_k\Bigg\}.\label{eq:time lambda k 2}
\end{align}
Recall the definitions of random variables $\tilA_{1,i},\tilA_{2,i}^{(l)},\tilA_{3,i}^{(l)},\tilA_{4,i}^{(p,q)}$ in~\eqref{eq:rac jnn a last}, the definition of random vector $\tilde{\bA}_{i}$ in~\eqref{eq:def tilbAvec} and define the following function:
\begin{align}
\label{eq:tilf prime definition}
\tilf'(a_1,\{a_2^{(l)}\}_{l=1}^k,\{a_3^{(l)}\}_{l=1}^k,\{a_4^{(p,q)}\}_{p,q=1}^{k}):=a_1+\sum_{l=1}^{k}\frac{2a_2^{(l)}}{\sqrt{1+a_3^{(l)}}}+\sum_{1\leq p<q\leq k}\frac{2a_4^{(p,q)}}{\sqrt{1+a_3^{(p)}}\sqrt{1+a_3^{(q)}}}.
\end{align}
One can verify that 
\begin{align}
\tilf':=\tilf'\bigg(\frac{1}{n_k}\sum_{i=1}^{n_k}\tilde{\bA}_{i}\bigg)
=\frac{1}{n_k}\Bigg(\|Z^{n_k}\|^2-n_k+2\sum_{i=1}^{k}\langle X_i^{n_k},Z^{n_k} \rangle+2\sum_{1\leq i<j\leq k}\langle X_i^{n_k},X_j^{n_k} \rangle\Bigg)\label{eq:defi tilfp}.
\end{align}
The Jacobian of $\tilf'$ evaluated at $\bzero$ is
\begin{align}
\bJ=\begin{bmatrix}
1,2, 0,2
\end{bmatrix}.
\end{align}
Thus, the dispersion term for the Berry-Esseen Theorem in Lemma~\ref{lemma:berry esseen for func} is given by
\begin{align}
\tilde{\sigma}_k^{'2}:= \bJ\cov(\tilde{\bA}_{1})\bJ^{\rm T}=\xi-1+4Pk+4k(k-1)P^2,
\end{align}
where $\cov(\tilde{\bA}_{1})$ is given in~\eqref{eq:cov til bA}. Applying Lemma~\ref{lemma:berry esseen for func}, we can bound~\eqref{eq:time lambda k 2} as follows:
\begin{align}
\Pr\bigg\{ \bigg|\frac{1}{n_k}\|Y^{n_k}\|^2-(1+kP)\bigg|> \lambda_k \bigg\}&\leq2\Big(\rmQ\Big(\frac{n_k^\frac{1}{2}\lambda_k}{\tilde{\sigma}'_k}\Big)+O(n_k^{-\frac{1}{2}})\Big)\label{eq:time lambda k 3}\\
&\leq 2\Big(\exp\Big(-\frac{n_k\lambda^2_k}{2\tilde{\sigma}^{'2}_k}\Big)+O(n_k^{-\frac{1}{2}})\Big)\label{eq:time lambda k 4}\\
&\leq O(n_k^{-\frac{1}{2}}),
\end{align}
where~\eqref{eq:time lambda k 4} follows because $\rmQ(t)\leq\exp(-t^2/2)$ for $t>0$. This verifies~\eqref{eq:time lambda k}.

Similarly, by taking the definition of $\|Y^{n_k}\|^2$ in, the left hand side of~\eqref{eq:time lambda t} can be written as
\begin{align}
\nn&\Pr\bigg\{ \bigg|\frac{1}{n_t}\|Y^{n_k}\|^2-(1+tP)\bigg|\leq \lambda_t \bigg\}\\
=&\Pr\Bigg\{\Bigg|\frac{1}{n_t}\Bigg(\|Z^{n_t}\|^2-n_t+2\sum_{i=1}^{k}\langle X_i^{n_t},Z^{n_t} \rangle+2\sum_{1\leq i<j\leq k}\langle X_i^{n_t},X_j^{n_t} \rangle+n_t(k-t)P\Bigg)\Bigg|\leq\lambda_t\Bigg\}\\
\leq&\Pr\{\tilf'\leq\lambda_t+(t-k)P\}\label{eq:time lambda t 2}\\
\leq&\rmQ\Big(\frac{n_t^{\frac{1}{2}}((k-t)P-\lambda_t)}{\tilde{\sigma}'_k}\Big)+O(n_t^{-\frac{1}{2}}),\label{eq:time lambda t 3}
\end{align}
where~\eqref{eq:time lambda t 2} follows by the definition of $\tilf'$ in~\eqref{eq:defi tilfp} and~\eqref{eq:time lambda t 3} follows by applying Berry-Esseen Theorem as in~\eqref{eq:time lambda k 3}. 
By taking $\lambda_t\in(0,P)$ for any $t\in[k-1]$, we ensure that $(k-t)P-\lambda_t>0$. Thus, \eqref{eq:time lambda t 3} can be further bounded as
\begin{align}
\Pr\bigg\{ \bigg|\frac{1}{n_t}\|Y^{n_k}\|^2-(1+tP)\bigg|\leq \lambda_t \bigg\}
&\leq\exp\Big(-\frac{n_t((k-t)P-\lambda_t)^2}{2\tilde{\sigma}^{'2}_k}\Big)+O(n_t^{-\frac{1}{2}})\label{eq:time lambda t 4}\\
&\leq O(n_t^{-\frac{1}{2}}),
\end{align}
where~\eqref{eq:time lambda t 4} follows because $\rmQ(t)\leq\exp(-t^2/2)$ for $t>0$.
This verifies~\eqref{eq:time lambda t}. 

\bibliographystyle{IEEEbib}
\bibliography{IEEEfull_wym}

\end{document}

%% file: preamble.tex
\usepackage{soul}
\usepackage[mathscr]{eucal}
\usepackage[cmex10]{amsmath}
\usepackage{epsfig,epsf,psfrag}
\usepackage{amssymb,amsmath,amsthm,amsfonts,latexsym}
\usepackage{amsmath,graphicx,bm,xcolor,url}
\usepackage{fixltx2e}
\usepackage{array}
\usepackage{verbatim}
\usepackage{bm}
\usepackage{algorithmic}
\usepackage{algorithm}
\usepackage{verbatim}
\usepackage{textcomp}
\usepackage{epstopdf}
\usepackage{mathrsfs,bbm,overpic}

\newcommand{\openone}{\leavevmode\hbox{\small1\normalsize\kern-.33em1}}

\catcode`~=11 \def\UrlSpecials{\do\~{\kern -.15em\lower .7ex\hbox{~}\kern .04em}} \catcode`~=13 

\allowdisplaybreaks[4]

\newcommand{\nn}{\nonumber}

\newcommand{\calA}{\mathcal{A}}
\newcommand{\calB}{\mathcal{B}}
\newcommand{\calC}{\mathcal{C}}
\newcommand{\calD}{\mathcal{D}}
\newcommand{\calE}{\mathcal{E}}
\newcommand{\calF}{\mathcal{F}}
\newcommand{\calG}{\mathcal{G}}

\newcommand{\calL}{\mathcal{L}}

\newcommand{\calN}{\mathcal{N}}

\newcommand{\calR}{\mathcal{R}}
\newcommand{\calS}{\mathcal{S}}

\newcommand{\calX}{\mathcal{X}}

\newcommand{\ba}{\mathbf{a}}
\newcommand{\bA}{\mathbf{A}}
\newcommand{\bb}{\mathbf{b}}
\newcommand{\bB}{\mathbf{B}}

\newcommand{\bC}{\mathbf{C}}

\newcommand{\boldf}{\mathbf{f}}

\newcommand{\bi}{\mathbf{i}}

\newcommand{\bJ}{\mathbf{J}}
\newcommand{\bk}{\mathbf{k}}

\newcommand{\bn}{\mathbf{n}}

\newcommand{\bS}{\mathbf{S}}
\newcommand{\bt}{\mathbf{t}}

\newcommand{\bV}{\mathbf{V}}

\newcommand{\bX}{\mathbf{X}}

\newcommand{\rmc}{\mathrm{c}}
\newcommand{\rmC}{\mathrm{C}}
\newcommand{\rmd}{\mathrm{d}}

\newcommand{\rme}{\mathrm{e}}

\newcommand{\rmG}{\mathrm{G}}

\newcommand{\rmP}{\mathrm{P}}

\newcommand{\rmQ}{\mathrm{Q}}

\newcommand{\rmT}{\mathrm{T}}

\newcommand{\bbE}{\mathsf{E}}

\newcommand{\bbN}{\mathbb{N}}

\newcommand{\bbR}{\mathbb{R}}

\newcommand{\bbo}{\mathbbm{1}}

\DeclareMathAlphabet{\mathbsf}{OT1}{cmss}{bx}{n}
\DeclareMathAlphabet{\mathssf}{OT1}{cmss}{m}{sl}

\DeclareSymbolFont{bsfletters}{OT1}{cmss}{bx}{n}  
\DeclareSymbolFont{ssfletters}{OT1}{cmss}{m}{n}
\DeclareMathSymbol{\bsfGamma}{0}{bsfletters}{'000}
\DeclareMathSymbol{\ssfGamma}{0}{ssfletters}{'000}
\DeclareMathSymbol{\bsfDelta}{0}{bsfletters}{'001}
\DeclareMathSymbol{\ssfDelta}{0}{ssfletters}{'001}
\DeclareMathSymbol{\bsfTheta}{0}{bsfletters}{'002}
\DeclareMathSymbol{\ssfTheta}{0}{ssfletters}{'002}
\DeclareMathSymbol{\bsfLambda}{0}{bsfletters}{'003}
\DeclareMathSymbol{\ssfLambda}{0}{ssfletters}{'003}
\DeclareMathSymbol{\bsfXi}{0}{bsfletters}{'004}
\DeclareMathSymbol{\ssfXi}{0}{ssfletters}{'004}
\DeclareMathSymbol{\bsfPi}{0}{bsfletters}{'005}
\DeclareMathSymbol{\ssfPi}{0}{ssfletters}{'005}
\DeclareMathSymbol{\bsfSigma}{0}{bsfletters}{'006}
\DeclareMathSymbol{\ssfSigma}{0}{ssfletters}{'006}
\DeclareMathSymbol{\bsfUpsilon}{0}{bsfletters}{'007}
\DeclareMathSymbol{\ssfUpsilon}{0}{ssfletters}{'007}
\DeclareMathSymbol{\bsfPhi}{0}{bsfletters}{'010}
\DeclareMathSymbol{\ssfPhi}{0}{ssfletters}{'010}
\DeclareMathSymbol{\bsfPsi}{0}{bsfletters}{'011}
\DeclareMathSymbol{\ssfPsi}{0}{ssfletters}{'011}
\DeclareMathSymbol{\bsfOmega}{0}{bsfletters}{'012}
\DeclareMathSymbol{\ssfOmega}{0}{ssfletters}{'012}

\newcommand{\tilA}{\tilde{A}}

\newcommand{\tilB}{\tilde{B}}

\newcommand{\tilf}{\tilde{f}}

\newcommand{\tilg}{\tilde{g}}

\newcommand{\hatW}{\hat{W}}

\newcommand{\tilX}{\tilde{X}}

\newcommand{\tilY}{\tilde{Y}}

\newcommand{\tilZ}{\tilde{Z}}

\newcommand{\barf}{\bar{f}}
\newcommand{\barg}{\bar{g}}

\newcommand{\barw}{\bar{w}}
\newcommand{\barx}{\bar{x}}

\newcommand{\barG}{\bar{G}}

\newcommand{\barW}{\bar{W}}
\newcommand{\barX}{\bar{X}}

\newcommand{\iid}{i.i.d.\ }

\DeclareMathOperator*{\argmax}{arg\,max}
\DeclareMathOperator*{\argmin}{arg\,min}

\DeclareMathOperator{\diag}{diag}

\DeclareMathOperator{\cov}{\mathsf{Cov}}

\newcommand{\bzero}{\mathbf{0}}
\newcommand{\bone}{\mathbf{1}}

\newtheorem{theorem}{Theorem} 
\newtheorem{lemma}[theorem]{Lemma}

\newtheorem{corollary}[theorem]{Corollary}
\newtheorem{definition}{Definition}

%% file: arxiv_v2.bbl
\begin{thebibliography}{10}

\bibitem{wang2023globecom}
Y.~Wang, L.~Zhou, and L.~Bai,
\newblock ``Achievable second-order asymptotics for additive non-{Gaussian} {MAC} and {RAC},''
\newblock in {\em IEEE GLOBECOM}, 2023, pp. 3627--3632.

\bibitem{saad2020iot}
W.~Saad, M.~Bennis, and M.~Chen,
\newblock ``A vision of {6G} wireless systems: Applications, trends, technologies, and open research problems,''
\newblock {\em IEEE Netw.}, vol. 34, no. 3, pp. 134--142, 2020.

\bibitem{federico2014massive}
F.~Boccardi, R.~W. Heath, A.~Lozano, T.~L. Marzetta, and P.~Popovski,
\newblock ``Five disruptive technology directions for {5G},''
\newblock {\em IEEE Commun. Mag.}, vol. 52, no. 2, pp. 74--80, 2014.

\bibitem{jiaxing2024aoiuav}
J.~Wang, L.~Bai, Z.~Fang, R.~Han, J.~Wang, and J.~Choi,
\newblock ``Age of information based {URLLC} transmission for {UAVs} on pylon turn,''
\newblock {\em IEEE Trans. Veh. Technol.}, vol. 73, no. 6, pp. 8797--8809, 2024.

\bibitem{rex2008bookOMA}
A.~W. Scott and R.~Frobenius,
\newblock {\em Multiple Access Techniques: FDMA, TDMA, and CDMA}, pp. 413--429,
\newblock Wiley-IEEE Press, 2008.

\bibitem{zhiguo2017noma}
Z.~Ding, X.~Lei, G.~K. Karagiannidis, R.~Schober, J.~Yuan, and V.~K. Bhargava,
\newblock ``A survey on non-orthogonal multiple access for {5G} networks: Research challenges and future trends,''
\newblock {\em IEEE J. Sel. Areas Commun.}, vol. 35, no. 10, pp. 2181--2195, 2017.

\bibitem{letaief2019roadmap6G}
K.~B. Letaief, W.~Chen, Y.~Shi, J.~Zhang, and Y.~A. Zhang,
\newblock ``The roadmap to {6G}: {AI} empowered wireless networks,''
\newblock {\em IEEE Commun. Mag.}, vol. 57, no. 8, pp. 84--90, 2019.

\bibitem{shannon1961two}
C.~E. Shannon,
\newblock ``Two-way communication channels,''
\newblock in {\em Proc. Fourth Berkeley Symp. Math. Stat. Probab., Volume 1}. University of California Press, 1961, vol.~4, pp. 611--645.

\bibitem{Liu2024MAtutorial}
Y.~Liu, C.~Ouyang, Z.~Ding, and R.~Schober,
\newblock ``The road to next-generation multiple access: A 50-year tutorial review,''
\newblock {\em Proc. IEEE}, vol. 112, no. 9, pp. 1100--1148, 2024.

\bibitem{Ahlswede1973MultiwayCC}
R.~Ahlswede,
\newblock ``Multi-way communication channels,''
\newblock in {\em IEEE ISIT}, 1973, pp. 23--25.

\bibitem{shannon1959Gaussiancode}
C.~E. Shannon,
\newblock ``Probability of error for optimal codes in a {Gaussian} channel,''
\newblock {\em Bell Syst. Tech. J.}, vol. 38, no. 3, pp. 611--656, 1959.

\bibitem{verdu1998multiuser}
S.~Verd\'u,
\newblock {\em Multiuser detection},
\newblock Cambridge University Press, 1998.

\bibitem{PPV}
Y.~Polyanskiy, H.~Vincent Poor, and S.~Verd\'u,
\newblock ``Channel coding rate in the finite blocklength regime,''
\newblock {\em IEEE Trans. Inf. Theory}, vol. 56, no. 5, pp. 2307--2359, 2010.

\bibitem{zhou2023mono}
L.~Zhou and M.~Motani,
\newblock ``Finite blocklength lossy source coding for discrete memoryless sources,''
\newblock {\em Found. Trends Commun. Inf. Theory}, vol. 20, no. 3, pp. 157--389, 2023.

\bibitem{tan2014mono}
V.~Y.~F. Tan,
\newblock ``Asymptotic estimates in information theory with non-vanishing error probabilities,''
\newblock {\em Found. Trends Commun. Inf. Theory}, vol. 11, no. 1-2, pp. 1--184, 2014.

\bibitem{molavianjazi2015second}
E.~MolavianJazi and J.~N. Laneman,
\newblock ``A second-order achievable rate region for {Gaussian} multi-access channels via a central limit theorem for functions,''
\newblock {\em IEEE Trans. Inf. Theory}, vol. 61, no. 12, pp. 6719--6733, 2015.

\bibitem{yavas2021gaussianmac}
R.~C. Yavas, V.~Kostina, and M.~Effros,
\newblock ``Gaussian multiple and random access channels: Finite-blocklength analysis,''
\newblock {\em IEEE Trans. Inf. Theory}, vol. 67, no. 11, pp. 6983--7009, 2021.

\bibitem{yongpeng2020massiveaccess}
Y.~Wu, X.~Gao, S.~Zhou, W.~Yang, Y.~Polyanskiy, and G.~Caire,
\newblock ``Massive access for future wireless communication systems,''
\newblock {\em IEEE Wirel. Commun.}, vol. 27, no. 4, pp. 148--156, 2020.

\bibitem{poly2017massiveRac}
Y.~Polyanskiy,
\newblock ``A perspective on massive random-access,''
\newblock in {\em IEEE ISIT}, 2017, pp. 2523--2527.

\bibitem{poly2017lowComplexityRac}
O.~Ordentlich and Y.~Polyanskiy,
\newblock ``Low complexity schemes for the random access {Gaussian} channel,''
\newblock in {\em IEEE ISIT}, 2017, pp. 2528--2532.

\bibitem{poly2020energyEfficientCodeRAC}
S.~S. Kowshik, K.~Andreev, A.~Frolov, and Y.~Polyanskiy,
\newblock ``Energy efficient coded random access for the wireless uplink,''
\newblock {\em IEEE Trans. Commun.}, vol. 68, no. 8, pp. 4694--4708, 2020.

\bibitem{Fengler2021massiveMIMOUnsourcedRAC}
A.~Fengler, S.~Haghighatshoar, P.~Jung, and G.~Caire,
\newblock ``Non-{Bayesian} activity detection, large-scale fading coefficient estimation, and unsourced random access with a massive {MIMO} receiver,''
\newblock {\em IEEE Trans. Inf. Theory}, vol. 67, no. 5, pp. 2925--2951, 2021.

\bibitem{yavas2021rac}
R.~C. Yavas, V.~Kostina, and M.~Effros,
\newblock ``Random access channel coding in the finite blocklength regime,''
\newblock {\em IEEE Trans. Inf. Theory}, vol. 67, no. 4, pp. 2115--2140, 2021.

\bibitem{lapidoth1996mismatch}
A.~Lapidoth,
\newblock ``Nearest neighbor decoding for additive non-{G}aussian noise channels,''
\newblock {\em IEEE Trans. Inf. Theory}, vol. 42, no. 5, pp. 1520--1529, 1996.

\bibitem{scarlett2017mismatch}
J.~Scarlett, V.~Y.~F. Tan, and G.~Durisi,
\newblock ``The dispersion of nearest-neighbor decoding for additive {non-Gaussian} channels,''
\newblock {\em IEEE Trans. Inf. Theory}, vol. 63, no. 1, pp. 81--92, 2017.

\bibitem{zhou2019jscc}
L.~Zhou, V.~Y.~F. Tan, and M.~Motani,
\newblock ``The dispersion of mismatched joint source-channel coding for arbitrary sources and additive channels,''
\newblock {\em IEEE Trans. Inf. Theory}, vol. 65, no. 4, pp. 2234--2251, 2019.

\bibitem{zhou2019refined}
L.~Zhou, V.~Y.~F. Tan, and M.~Motani,
\newblock ``Refined asymptotics for rate-distortion using {G}aussian codebooks for arbitrary sources,''
\newblock {\em IEEE Trans. Inf. Theory}, vol. 65, no. 5, pp. 3145--3159, 2019.

\bibitem{zhou2023sr}
L.~Bai, Z.~Wu, and L.~Zhou,
\newblock ``Achievable refined asymptotics for successive refinement using {Gaussian} codebooks,''
\newblock {\em IEEE Trans. Inf. Theory}, vol. 69, no. 6, pp. 3525--3543, 2023.

\bibitem{wyner1975mac}
A.~Wyner,
\newblock ``Recent results in the {Shannon} theory,''
\newblock {\em IEEE Trans. on Inf. Theory}, vol. 20, no. 1, pp. 2--10, 1974.

\bibitem{el2011network}
A.~El Gamal and Y.-H. Kim,
\newblock {\em Network Information Theory},
\newblock Cambridge University Press, 2011.

\bibitem{zhou2017sr}
L.~Zhou, V.~Y.~F. Tan, and M.~Motani,
\newblock ``Second-order and moderate deviations asymptotics for successive refinement,''
\newblock {\em IEEE Trans. Inf. Theory}, vol. 63, no. 5, pp. 2896--2921, 2017.

\bibitem{effros2020mac}
Y.~Liu and M.~Effros,
\newblock ``Finite-blocklength and error-exponent analyses for {LDPC} codes in point-to-point and multiple access communication,''
\newblock in {\em IEEE ISIT}, 2020, pp. 361--366.

\bibitem{tantomamichel2015}
V.~Y.~F. Tan and T.~Tomamichel,
\newblock ``The third-order term in the normal approximation for the {AWGN} channel,''
\newblock {\em IEEE Trans. Inf. Theory}, vol. 61, no. 5, pp. 2430--2438, 2015.

\bibitem{polyanskiy2010thesis}
Y.~Polyanskiy,
\newblock {\em Channel Coding: Non-Asymptotic Fundamental Limits},
\newblock Ph.D. thesis, Department of Electrical Engineering, Princeton University, 2010.

\bibitem{kostina2019macracsourcecoding}
S.~Chen, M.~Effros, and V.~Kostina,
\newblock ``Lossless source coding in the point-to-point, multiple access, and random access scenarios,''
\newblock in {\em IEEE ISIT}, 2019, pp. 1692--1696.

\bibitem{gallager_ensemble}
R.~Gallager,
\newblock ``The random coding bound is tight for the average code {(Corresp.)},''
\newblock {\em IEEE Trans. Inf. Theory}, vol. 19, no. 2, pp. 244--246, 1973.

\bibitem{molavianjazi2013fading}
E.~MolavianJazi and J.~N. Laneman,
\newblock ``{On the second-order coding rate of non-ergodic fading channels},''
\newblock in {\em 51st Annu. Allerton Conf. Commun. Control Comput.}, 2013, pp. 583--587.

\bibitem{weiyang2014mimo}
W.~Yang, G.~Durisi, T.~Koch, and Y.~Polyanskiy,
\newblock ``Quasi-static multiple-antenna fading channels at finite blocklength,''
\newblock {\em IEEE Trans. Inf. Theory}, vol. 60, no. 7, pp. 4232--4265, 2014.

\bibitem{zhou2019multiConnectivity}
L.~Zhou, A.~Wolf, and M.~Motani,
\newblock ``On lossy multi-connectivity: Finite blocklength performance and second-order asymptotics,''
\newblock {\em IEEE J. Sel. Areas Commun.}, vol. 37, no. 4, pp. 735--748, 2019.

\bibitem{zhou2023covert}
L.~Bai, J.~Xu, and L.~Zhou,
\newblock ``Covert communication for spatially sparse {mmWave} massive {MIMO} channels,''
\newblock {\em IEEE Trans. Commun.}, vol. 71, no. 3, pp. 1615--1630, 2023.

\bibitem{yu2021awgncovert}
X.~Yu, S.~Wei, and Y.~Luo,
\newblock ``Finite blocklength analysis of {Gaussian} random coding in {AWGN} channels under covert constraint,''
\newblock {\em IEEE Trans. Inf. Forensic Secur.}, vol. 16, pp. 1261--1274, 2021.

\bibitem{shiyuan2021covertmimo}
S.~Wang and M.~R. Bloch,
\newblock ``Covert {MIMO} communications under variational distance constraint,''
\newblock {\em IEEE Trans. Inf. Forensic Secur.}, vol. 16, pp. 4605--4620, 2021.

\bibitem{zhou2020classification}
L.~Zhou, V.~Y.~F. Tan, and M.~Motani,
\newblock ``Second-order asymptotically optimal statistical classification,''
\newblock {\em Inf. Inference.}, vol. 9, no. 1, pp. 81--111, 2020.

\bibitem{zhou2023achievableerrorexponentstwophase}
L.~Zhou, J.~Diao, and L.~Bai,
\newblock ``Achievable error exponents for two-phase multiple classification,''
\newblock {\em arXiv:2210.12736}, 2023.

\bibitem{ShaoLun2023secondorderasymptoticscovert}
X.~Yu, S.~Wei, S.~Huang, and X.~Zhang,
\newblock ``The second order asymptotics of covert communication over {AWGN} channels,''
\newblock {\em arXiv:2305.17924}, 2023.

\bibitem{bloch2019secondorderCovert}
M.~Tahmasbi and M.~R. Bloch,
\newblock ``First- and second-order asymptotics in covert communication,''
\newblock {\em IEEE Trans. Inf. Theory}, vol. 65, no. 4, pp. 2190--2212, 2019.

\end{thebibliography}
